\documentclass[11pt,reqno]{amsart}
\usepackage[top=2.54cm, bottom=2.54cm, left=2.2cm, right=2.2cm]{
	geometry}
\usepackage{dsfont, amssymb,amsmath,amscd,latexsym, amsthm, amsxtra,amsfonts}
\usepackage[all]{xy}
\usepackage[active]{srcltx}
\usepackage{ulem}
\usepackage{chngpage}
\usepackage{array}
\usepackage{tabularx}
\usepackage{datetime}
\usepackage{bbm}
\usepackage{enumerate}
\usepackage{mathrsfs}
\usepackage{graphicx}
\usepackage{comment}
\usepackage{mathtools}
\usepackage{color}
\usepackage{xcolor}
\usepackage[hang,flushmargin]{footmisc}

\usepackage{tikz}
\usetikzlibrary{calc,arrows}
\usepackage{verbatim}
\usepackage{graphicx}
\usepackage{subfigure}

\newtheorem{theorem}{Theorem}[section]

\newtheorem{proposition}[theorem]{Proposition}
\newtheorem{lemma}[theorem]{Lemma}

\theoremstyle{definition}

\renewcommand{\theequation}{\arabic{section}.\arabic{equation}}

\theoremstyle{definition}

\theoremstyle{definition}
\newtheorem{remark}{Remark}
\theoremstyle{definition}

\newcommand{\rd}{\mathrm{d}}

\renewcommand{\cite}{\citet}

\usepackage[pdfstartview=FitH, bookmarksnumbered=true,bookmarksopen=true, colorlinks=true, pdfborder=001, citecolor=blue, linkcolor=blue,urlcolor=blue]{hyperref}
\usepackage{graphics}
\graphicspath{{figures/}}
\usepackage{apacite}
\usepackage[round]{natbib}
\usepackage{xpatch}
\usepackage{xcolor}

\begin{document}
	\makeatletter
	\def\@setauthors{%
		\begingroup
		\def\thanks{\protect\thanks@warning}%
		\trivlist \centering\footnotesize \@topsep30\p@\relax
		\advance\@topsep by -\baselineskip
		\item\relax
		\author@andify\authors
		\def\\{\protect\linebreak}%
		{\authors}%
		\ifx\@empty\contribs \else ,\penalty-3 \space \@setcontribs
		\@closetoccontribs \fi
		\endtrivlist
		\endgroup } \makeatother
	\baselineskip 18pt
	\title[{{\tiny Optimal management of DB pension fund}}]
	{{\tiny
			Optimal management of  DB pension fund under both underfunded and overfunded cases}} \vskip 10pt\noindent
	\author[{\tiny  Guohui Guan, Zongxia Liang, Yi xia}]
	{\tiny {\tiny  Guohui Guan$^{a,b,\dag}$, Zongxia Liang$^{c,\ddag}$, Yi Xia$^{c,*}$}
		\vskip 10pt\noindent
		{\tiny ${}^a$Center for Applied Statistics, Renmin University of China, Beijing, 100872, China
			\vskip 10pt\noindent\tiny ${}^b$School of Statistics, Renmin University of China, Beijing 100872, China
			\vskip 10pt\noindent\tiny ${}^c$Department of Mathematical Sciences, Tsinghua
			University, Beijing 100084, China}
\footnote{
$^{\dag}$  e-mail: guangh@ruc.edu.cn\\
 $^{\ddag}$  e-mail:  liangzongxia@mail.tsinghua.edu.cn\\
 $^*$  Corresponding author, e-mail:  xia-y20@mails.tsinghua.edu.cn}}
 %\numberwithin{equation}{section}
\date{}
\maketitle
\begin{abstract}
%This paper investigates the risk management of an aggregated defined benefit pension plan in a stochastic environment. %The interest rate follows the Ornstein-Uhlenbeck model, the benefits follow the geometric Brownian motion while the contribution rate is determined by the spread method of fund amortization.
%		The pension manager invests in the financial market with three assets: cash, bond and stock. Regardless of the initial status of the plan, we suppose that the pension fund may become underfunded or overfunded in the planning horizon. The optimization goal of the manager is to maximize the expected utility in the overfunded region minus the weighted solvency risk in the underfunded region. By introducing an auxiliary process and related equivalent optimization problems and using the martingale method, the optimal wealth process, optimal portfolio and efficient frontier are obtained under four cases (high tolerance towards solvency risk, low tolerance towards solvency risk,  a specific lower bound, and high lower bound). The tolerance towards solvency risk plays a prominent role in the management. We also obtain the probabilities that the optimal terminal wealth falls in the overfunded and underfunded regions. At last, we present numerical analyses to illustrate the manager's economic behaviors.
%		
This paper investigates the optimal management of an aggregated defined benefit pension plan in a stochastic environment. The interest rate follows the Ornstein-Uhlenbeck model, the benefits follow the geometric Brownian motion while the contribution rate is determined by the spread method of fund amortization. The pension manager invests in the financial market with three assets: cash, bond and stock. Regardless of the initial status of the plan, we suppose that the pension fund may become underfunded or overfunded in the planning horizon. The optimization goal of the manager is to maximize the expected utility in the overfunded region minus the weighted solvency risk in the underfunded region. By introducing an auxiliary process and related equivalent optimization problems and using the martingale method, the optimal wealth process, optimal portfolio and efficient frontier are obtained under four cases (high tolerance towards solvency risk, low tolerance towards solvency risk,  a specific lower bound, and high lower bound). Moreover, we also obtain the probabilities that the optimal terminal wealth falls in the overfunded and underfunded regions. At last, we present numerical analyses to illustrate the manager's economic behaviors.
\vskip 10 pt \noindent
	Submission Classification: IB13, IE13, IE43.
\end{abstract}
	\vskip15pt
	\setcounter{equation}{0}
	\section{\bf Introduction}
	 Optimal management of pension funds has become more and more important.  With the growth in both the size and the proportion of older persons in the population, the management of pension funds plays a prominent role in the social security system. Pension fund ensures a sustainable and adequate retirement payment for pension participants and is responsible for the welfare of beneficiaries. The main features in a pension fund are contributions from the working sponsors, benefits provided to the retirees. Classified by these two features, there are generally two kinds of pension funds worldwide: defined contribution (DC) pension fund and defined benefit (DB) pension fund. In a DC pension fund, the contribution is fixed in advance and the benefit after retirement depends on the fund wealth at retirement time, which may be paid by a lump sum or a life annuity. On the contrary, the benefit in a DB pension fund is fixed in advance and the contribution is determined to keep an actuarial balance.
	
	The benefits are determined by the fund surplus in a DC plan. As such, the risk is undertaken by the participants and the manager almost has no solvency risk. There are numerous researches concerning the risk management of DC pension fund, see \cite{boulier2001optimal},  \cite{emms2012lifetime},  \cite{temocin2018constant},  \cite{zeng2018ambiguity}, etc. However, in a DB pension plan, the manager is faced with the liability (promised benefits) paid to the retirees. Generally, the DB pension plan is classified into two {states}. As in \cite{carroll1998pension}, when the value of the fund assets is higher (smaller) than the actuarial liability, the pension fund is overfunded (underfunded). 	In reality, overfunded and underfunded may both happen in the management of a DB pension plan, which can be shown by the databases in \cite{franzoni2006pension} and \cite{beaudoin2010potential}. Besides, the state of a pension fund may also change.  Market changes can cause a fund to be either underfunded or overfunded, and it is fairly common for defined benefit plans to become significantly {underfunded or} overfunded.  \cite{asthana1999determinants} {{shows}} that the influence of actuarial choices on funding allows firms to become overfunded or underfunded over a period of time.  Wilshire Associates reported that 51\% of all state retirement
pension funds were found to be underfunded in 2001 up dramatically from the
previous year when only 31\% were underfunded (see \cite{eaton2004effect}).  However, most  researches suppose that the DB pension fund starts at underfunded (overfunded) case will always be underfunded (overfunded), such as  \cite{josa2012stochastic}, \cite{JJ18}, etc. The changes of the (underfunded, overfunded) states of the DB pension fund are prohibited {in  previous articles}. In this paper, we suppose that regardless the initial state of the fund, it can become both underfunded and overfunded in the future. This paper attempts to unify the overfunded and underfunded cases by a non-concave piece-wise function.

Researches have paid great attention to the management of DB pension fund under underfunded and overfunded cases separately. In \cite{sundaresan1997valuation}, a continuous-time framework is established to maximize the expected growth rate of surplus. Later, in \cite{siegmann2007optimal}, a discrete-time framework is considered to maximize the mean minus the downside risk.   In the underfunded case of a DB pension plan,  solvency risk is minimized to search the optimal strategies, see \cite{JR10}, i.e.,
\begin{equation}\label{iop1}
	\min \mathbb{E}[X^2(T)], X(T)<0,
\end{equation}
where $X(T)$ represents the fund surplus at time $T$. In the overfunded case of a DB pension plan, plan termination choice is focused, see \cite{stone1987financing}, \cite{thomas1989firms}, \cite{kapinos2009determinants}, etc.  For the optimal management in the overfunded case, the pension manager is concerned with the utility derived from the fund surplus, such as \cite{JJ18}, i.e.,
\begin{equation}\label{iop2}
	\max \mathbb{E}[U(X(T))], X(T)>0,
\end{equation}
where
\[U(x)=\frac{x^{1-\gamma}}{1-\gamma}\]
with $\gamma\in(0,1)$. The stochastic dynamic programming method is applied to obtain the optimal portfolio in their papers. %In \cite{JR19}, a stochastic differential game is played between the sponsors of the DB pension plan and participants in the overfunded case.

In \cite{JJ18}, the authors consider the DB pension fund in underfunded and overfunded cases separately. They suppose that a fund starting in the underfunded (overfunded) region never becomes overfunded (underfunded) under some condition. This phenomenon mainly arises from the power utility function adopted in \cite{JJ18}, which leads to an optimal wealth process evolving as the geometric Brownian motion. As such, the sign of the fund surplus never changes. However, in reality, either beginning in the underfunded region or overfunded region, the DB pension fund in later period has both the possibilities in these two regions. Besides, once the returns from the financial market are high, the manager may accept larger solvency risk to expect a larger utility of the positive fund surplus. In this paper, we suppose that the aggregated DB pension fund has both possibilities in the overfunded and underfunded regions. We combine the solvency risk in the underfunded region and the expected utility in the overfunded region in a unified framework. The solvency risk is modeled by the expected squared terminal surplus as in  \cite{JJ18}. We adopt the constant relative risk-averse (CRRA) utility in the overfunded region. This paper combines (\ref{iop1}) and (\ref{iop2}) in a unified framework and studies the following criterion,
\begin{equation}\label{iop3}\max \mathbb{E}\left\{-\alpha X^2(T)1_{\{X(T)<0\}} + \frac{X(T)^{1-\gamma}}{1-\gamma}1_{\{X(T)>0\}} \right\},
\end{equation}
where $\alpha>0$ measures the attitude towards solvency risk. In (\ref{iop3}), the manager is risk aversion towards gains and losses and different risk aversion coefficients are characterized.

(\ref{iop3}) characterizes different preferences over positive wealth and negative wealth. Originally from \cite{kahneman1979prospect}, a variety of studies concentrate on different attitudes towards gains and losses. In \cite{kahneman1979prospect}, the function is concave over gains while convex over losses, which shows the loss aversion behavior of an individual. Different preferences of DC pension plan above and below a threshold have also been modeled in DC pension plan under loss aversion in \cite{blake2013target}, \cite{dong2020optimal}.  Different from { that in} \cite{kahneman1979prospect}, (\ref{iop3}) shows that the manager is less risk aversion towards losses, which is consistent with the findings in \cite{march1996learning} that the human beings exhibit greater risk aversion for gains than for losses in a wide variety of situations. Besides, \cite{cardenas2014my} also reveal that  individuals in general are more risk tolerant towards losses than towards gains.

As the management horizon for a pension plan is usually long, the manager is faced with various risks, such as interest risk (\cite{JR10}, \cite{cairns2006stochastic}), volatility risk (\cite{JJ18}), jump risk (\cite{josa2012stochastic}, \cite{li2021alpha}), longevity risk (\cite{cox2013managing}), etc. The financial market in our paper is similar with \cite{JR10}, \cite{guan2016optimal}, which incorporate interest risk and consider three assets: cash, bond, and a stock. Besides, the benefits and contributions are important in the DB pension plan. To obtain explicit solutions, we need to assume the complete dependence of the benefits with the financial risks. The benefits are supposed to be stochastic and follow a geometric Brownian motion. For characterization of the contributions,   \cite{hainaut2011optimal} {minimize} the deviation of the terminal target asset from the mathematical reserve to find the optimal contribution rate in a DB pension fund. In this paper, we simply assume that the contribution rate is calculated by the spread method of fund amortization as in \cite{JR10}, \cite{JJ18}.  In this paper, a more flexible technical rate of actualization (the rate of valuation of the liabilities) can be studied, which is however determined by risk-neutral valuation in \cite{JR10}.  We see that as benefits and rate of actualization are affected by interest risk, compared with \cite{JR10}, the fund surplus is not self-financing and has a  complicated form. Besides, the optimization rule of the pension fund is a non-concave piece-wise function, which limits the application of the stochastic programming method in \cite{JJ18}.  %To disentangle the optimization problem, we first introduce  auxiliary processes by derivative pricing theory. Then using the martingale method and dual method, the results of optimal strategies and wealth processes are obtained for an equivalent problem and the original optimization problem.

We have the following contributions in this paper. First, we consider a stochastic interest rate, a more flexible technical rate of actualization, and stochastic benefits for the aggregated DB pension plan. Second,  the dependence of the interest rate, technical rate of actualization, and benefits leads to a very complicated form of the non-self-financing fund surplus. Using derivative pricing theory, the additional forms in the drift terms can be replicated by complicated processes. The optimization problem is transformed to an equivalent one with respect to a { self-financing} process. Third, this paper considers the underfunded case and overfunded cases together and proposes a unified optimization rule for the pension fund. A weight is introduced to characterize the manager's tolerance towards solvency risk. The optimization rule is a non-concave piece-wise  function. Martingale method and Lagrange dual method are employed to solve the optimization problem. Fourth, we obtain four cases for the optimal results corresponding to high tolerance towards solvency risk, low tolerance towards solvency risk, a specific lower bound, and high lower bound, respectively. Meanwhile, similar to the mean-variance analysis, the points of the solvency risk and expected utility form an efficient frontier. The efficient frontier and the possibilities in the underfunded and overfunded regions are also obtained. Last, we show detailed numerical examples of optimal wealth process, optimal investment strategies, efficient frontier, and possibilities in two regions. As we expect, the efficient frontier is similar to the parabolic form, which indicates that the manager can  {take} some solvency risk to expect a larger expected utility. The results can provide comprehensive guidance for the pension manager with different economic parameters, risk aversions, and tolerance levels over solvency risk.

The remainder of this paper is organized as follows. Section 2 presents the model of the financial market and aggregated DB pension fund. Section 3 formulates the optimization problem with solvency risk and expected utility. In Section 4, the optimal solutions are obtained. Section 5 illustrates the economic behaviors of the manager, and Section 6 is a conclusion. Most of the proofs are relegated in Appendices.
\vskip 15pt
\setcounter{equation}{0}
	
	\section{\bf Financial Model}
	Let $ (\Omega,\mathcal{F},\mathbb{F},\mathbb{P}) $ be a complete filtered  probability space where $ \mathbb{F}\!= \!\{\mathcal{F}_t|0\!\leqslant\! t\!\leqslant\! T\} $ and $ \mathcal{F}_t $ stands for the information available before time $ t $ in the market. The DB pension fund  starts at initial time 0 and the planning horizon for the pension fund manager is $[0, T]$. The pension fund manager adjusts the  portfolio  within time horizon $[0, T]$. All the processes introduced below are assumed to be well-defined and progressively measurable w.r.t.(with respect to) $\mathbb{F}$. We suppose that there are no transaction costs and short selling is also allowed in the financial market.

In this paper, we consider an aggregate defined benefit pension fund. At every time, the sponsors coexist with retired participants in the system. The benefits received by the retired participants are fixed in advance, which evolve with the financial market. Meanwhile, the pension fund receives contributions from the working sponsors. The contributions are calculated by the spread method of fund amortization as in \cite{josa2006optimal}, \cite{JJ18}. The main notations of the elements of the pension plan are as follows.

\begin{tabularx}{0.9\textwidth}{rX}
			$T$: & date of the end of the pension plan, with $0 < T $\\
			$F(t)$: & value of fund assets at time $t$\\
			$P(t)$: & benefits promised to the participants at time $t$\\
			$C(t)$: & contributions made by the sponsors to the funding process at time $t$\\
			$AL(t)$: & actuarial liability at time $t$, or total liabilities of the sponsors \\
			$NC(t)$: &normal cost at time $t$, which is the value of the contributions allowing equality between asset funds and obligation when the fund assets match the actuarial liability and the market is risk-free\\
			$X(t)$: & fund surplus at time $t$, equal to $F (t)-AL(t)$ ($-X(t)$ represents the unfunded actuarial liability)\\
			$SC(t)$: & supplementary cost at time $t$, equal to $C(t)-NC(t)$\\
			$M(u)$: & distribution {functions} of workers aged  $u\in[m, d]$, $ M(m)=0 $, $ M(d)=1 $\\
			$\hat{\delta}(t)$: & technical rate of actualization {at time $t$}, which is the rate of valuation of the liabilities and is specified by the regulatory authorities\\
			$r(t)$: & risk-free interest rate\\
		\end{tabularx}
	
	\subsection{\bf Financial market}
	
	First, we present the financial market for the pension fund, which consists of three assets: cash, zero-coupon bond and stock. Because the planning horizon of a pension fund may be long, the interest rate in the market fluctuates greatly. As such, it is natural to consider a stochastic interest rate. In this paper, the interest rate $ r=\left\{r(t)|0\leqslant t\leqslant T\right\} $  follows the Ornstein-Uhlenbeck (OU) model:
	\begin{equation}\label{equ-r}
	\begin{split}
\mathrm{d}r(t)=a(b-r(t))\mathrm{d}t-\sigma_r\mathrm{d}W_r(t), \ \ \ r(0)=r_0,
\end{split}
\end{equation}
where $a$, $b$ and $\sigma_r$ are positive constants, $W_r=\left\{W_r(t)|0\leqslant t\leqslant T\right\}$ is a standard Brownian motion on $ (\Omega,\mathcal{F},\mathbb{F},\mathbb{P}) $. The OU model possesses the mean-reverting property and $b$ represents the mean level of $r(t)$. $a$ is the mean reversion rate and {$\sigma_r$} characterizes the volatility. 	
%	Then we get	\begin{equation}\label{r}
%		r(t)=b-e^{-at}\left[(b-r_0)+\sigma_r\int_0^te^{as}\rd W_r(s)\right],
%	\end{equation}
The evolution of the cash $ S_0=\left\{S_0(t)|0\leqslant t\leqslant  T\right\} $ is
	\begin{equation}\label{equ-S0}
		\begin{split}
			\frac{\mathrm{d}S_0(t)}{S_0(t)}=r(t)\mathrm{d}t,\ \ \  S_0(0)=s_0.
		\end{split}
	\end{equation}
Because interest risk exists in the financial market, to hedge interest risk, we introduce here a zero-coupon bond with fixed payment \$1 at maturity time $\tilde{T}$. The zero-coupon bond can be viewed as the interest rate derivative. Based on \cite{hull1990pricing}, the price of the zero-coupon bond at time $t$ with maturity time $\tilde{T}$ is
	\begin{equation*}
		\begin{split}
			B(t,\tilde{T})=e^{C(t,\tilde{T})-A(t,\tilde{T})r(t)},
		\end{split}
	\end{equation*}
	where \ \ $A(t,\tilde{T})=\frac{1-e^{-a(\tilde{T}-t)}}{a}$,\ \ $C(t,\tilde{T})=-R(\tilde{T}-t)+A(t,\tilde{T})
[R-\frac{\sigma_r^2}{2a^2}]+\frac{\sigma_r^2}{4a^3}(1-e^{-2a(\tilde{T}-t)})$, \\ $R=b+\frac{\sigma_r\lambda_r}{a}-\frac{\sigma_r^2}{2a^2}$. Here, $\lambda_r>0$ represents the market price of risk of $W_r$. %When $\lambda_r$ gets larger, the expected return of the bond increases, and the price $B(t,\tilde{T})$ decreases.
	
	Moreover, $B(t,\tilde{T})$ satisfies the following backward stochastic differential equation (BSDE)
	\begin{equation*}
	\begin{split}		\frac{\mathrm{d}B(t,\tilde{T})}{B(t,\tilde{T})}&=r(t)\mathrm{d}t+h(\tilde{T}-t)
(\lambda_r\mathrm{d}t+\mathrm{d}W_r(t)), \ \ \
B(\tilde{T},\tilde{T})=1,
		\end{split}
	\end{equation*}
	where $h(t)=\frac{1-e^{-at}}{a}\sigma_r$.
	
   %As stated in \cite{boulier2001optimal}, the manager needs to purchase bonds with different times to maturity when allocating in $B(t,\tilde{T})$. However, in the financial market, the times to maturity of bonds are limited. Therefore, as in \cite{boulier2001optimal}, other than $B(t,\tilde{T})$ with a varying time to maturity
   Similar with \cite{boulier2001optimal}, we consider a rolling bond $B_K=\left\{B_K(t)|0\leqslant t\leqslant T \right\}$ with {constant time $K$ to maturity date} to hedge interest risk. $B_K$ satisfies the following stochastic differential equation (SDE)
	\begin{equation}\label{equ-B}
		\begin{split}
			\frac{\mathrm{d}B_K(t)}{B_K(t)}=r(t)\mathrm{d}t+h(K)(\lambda_r\mathrm{d}t+\mathrm{d}W_r(t)),
		\end{split}
	\end{equation}
	where $h(K)$ is a constant increasing with $K$ and represents the volatility of the rolling bond.
	
	The third asset in the financial market we consider is the stock with dynamics $ S=\left\{S(t)|0\leqslant t\leqslant T\right\} $ satisfying
	\begin{equation}\label{equ-S}
			\frac{\mathrm{d}S(t)}{S(t)}=r(t)\mathrm{d}t+\sigma_1(\lambda_r\mathrm{d}t+\mathrm{d}W_r(t))
			+\sigma_2(\lambda_S\mathrm{d}t+\mathrm{d}W_S(t)), \ \ \
			S(0)=S_0,
	\end{equation}
	where $W_S=\left\{W_S(t)|0\leqslant t\leqslant T\right\}$ is a standard Brownian motion on  $(\Omega,\mathcal{F},\mathbb{F},\mathbb{P}) $ and is independent of $W_r$. $\lambda_S$ characterizes the market price of risk of $W_S$. $\sigma_1$ and $\sigma_2$ {drive the volatility of the stock}.

	\subsection{\bf DB pension fund}
	We now  present the model for the DB pension fund. In reality, the benefits of the participants are often related to the average salary of the society, which evolves with the financial market. As such, it is reasonable to consider stochastic benefits. We suppose that the benefits process $P=\left\{P(t)|0\leqslant t\leqslant T\right\}$ satisfies the geometric Brownian motion
	\begin{equation}\label{equ-P}
	\begin{split}
		\frac{\mathrm{d}P(t)}{P(t)}=
		\mu\mathrm{d}t+\sigma_{P_1}\mathrm{d}W_r(t)
		+\sigma_{P_2}\mathrm{d}W_S(t), \ \ \
 P(0)=P_0,
\end{split}
	\end{equation}
	where $\mu$, $\sigma_{P_1}$ and $\sigma_{P_2}$ are non-negative constants.
%	\begin{remark}
%	In \cite{josa2008mean}, \cite{JR10}, the authors also assume stochastic benefits, which have additional risk compared with the financial market. However, if the benefits in our problem can be driven by other risks, the system is incomplete and the closed-form solution can not be obtained. As such, for simplicity, in Eq.~(\ref{equ-P}), we suppose that the risks of the benefits completely depend on the financial market.
%	\end{remark}

   The technical rate of actualization plays a prominent role in calculating the accurate valuation of liability.  Generally, the technical rate of actualization is determined to be specified by the regulatory authorities.	%In \cite{josa2004optimal}, the technical rate of actualization is set to be the risk-free interest rate. 
   In \cite{JR10}, $\hat{\delta}(t)$ is $r(t)$ plus a constant to attain a risk-neural valuation. In our work, the technical rate of actualization $\hat{\delta}(t)$ is supposed to rely on $r(t)$ and follows
	%	$ \delta=\left\{\hat{\delta}(t)|0\leqslant t\leqslant T\right\} $
%	which relies on $r(t)$:
	\begin{equation}\label{equ:delta}
		\hat{\delta}(t)=r(t)+\delta, \quad0\leqslant t\leqslant T,
	\end{equation}
where $\delta\geqslant0$ is a constant. The technical rates of actualization in \cite{JR10} also belong to the form given by Eq.~(\ref{equ:delta}). However, the constant $\delta$ needs to be specified in the literature to solve the optimization problem explicitly. In our model, we do not need to specify the value $\delta$.% and can study the effects of different  $\delta$ on the manager's strategies.

%	Assume the there is no other randomness in the benefits, that is, assume the benefits $P=\left\{P(t)|0\leqslant t\leqslant T\right\}$ satisfies the geometric Brownian motion
%	\begin{equation}\label{equ-P}
%		\frac{\mathrm{d}P(t)}{P(t)}=
%		\mu\mathrm{d}t+\sigma_{P_1}\mathrm{d}W_r(t)
%		+\sigma_{P_2}\mathrm{d}W_S(t),
%	\end{equation}
%	which indicates \begin{equation}\label{equ-P_1}
%		P(t)=P(0)\exp\left(\left(\mu-\frac12\left(\sigma_{P_1}^2+\sigma_{P_2}^2\right)\right)t+\sigma_{P_1}W_r(t)+\sigma_{P_2}W_S(t)\right).
%	\end{equation}
%	%	where $W_P$ is a Brownian motion which is independent of $W_r$ and $W_S$.
In the pension fund, the stochastic actuarial liability $AL=\left\{AL(t)|0\leqslant t\leqslant T\right\}$ and the stochastic normal cost $NC=\left\{NC(t)|0\leqslant t\leqslant T\right\}$ are given by
	\begin{equation*}
		\begin{split}
			A L(t)=\mathbb{E}\left[\int_{m}^{d} e^{-\int_{t}^{t+d-x} \hat{\delta}(s) d s} M(x) P(t+d-x) \rd x \Big| \mathcal{F}_{t}\right],\quad 0\leqslant t\leqslant T,\\
			N C(t)=\mathbb{E}\left[\int_{m}^{d} e^{-\int_{t}^{t+d-x} \hat{\delta}(s) d s} M^{\prime}(x) P(t+d-x) \rd x \Big| \mathcal{F}_{t}\right],\quad 0\leqslant t\leqslant T,
		\end{split}
	\end{equation*}
where $\mathbb{E}\left[\cdot| \mathcal{F}_{t}\right]$ represents the conditional expectation given information before time $t$. $M(x)$ is the distribution function of workers aged  $x\in[m, d]$ with $ M(m)=0 $ and $ M(d)=1 $. Besides, $M'$ means the derivative of $M$. $AL$ is the present value of the benefits promised to workers aged between $[m,d]$. All participants in the plan start at age $m$ and retire at time $d$. $NC$ represents the present value of the normal costs of the pension plan. At  time $t$, the manager takes conditional expectation to calculate the future liabilities and normal costs. Based on the properties of the conditional expectations, $AL(t)$ and $NC(t)$ can be rewritten as
	\begin{equation*}
	\begin{split}
		A L(t)=\int_{m}^{d}M(x) \mathbb{E}\left[e^{-\int_{t}^{t+d-x} \hat{\delta}(s) d s}  P(t+d-x)\Big| \mathcal{F}_{t}\right] \rd x ,\quad 0\leqslant t\leqslant T,\\
		N C(t)=\int_{m}^{d}M^\prime(x)\mathbb{E}\left[e^{-\int_{t}^{t+d-x} \hat{\delta}(s) d s}  P(t+d-x)\Big| \mathcal{F}_{t}\right] \rd x,\quad 0\leqslant t\leqslant T.
	\end{split}
\end{equation*}
As {both $\hat{\delta}(\cdot)$ and $P(\cdot)$} depend on  the interest risk $W_r$, {the processes $e^{-\int_{t}^{t+d-x} \hat{\delta}(s) d s}$  and  $P(t+d-x)$ in $AL$ and $NC$ are not independent}. The accurate values $AL$ and $NC$ are complicated, which are presented in the following proposition.
\begin{proposition}\label{ALNC}
The explicit forms of $AL(t)$ and $NC(t)$ are given by
\begin{equation*}
	\begin{split}
		AL(t)=\int_{m}^{d} e^{-r(t)A(x,d)+D(x,d)}M(x)\rd x P(t),\\
		NC(t)=\int_{m}^{d} e^{-r(t)A(x,d)+D(x,d)}M'(x)\rd x P(t),
	\end{split}
\end{equation*}
where $A(x,d)=\frac{1-e^{-a(d-x)}}{a}$ and
\begin{equation*}\begin{aligned}
		D(x,d)= -\frac{\sigma_r^2}{4a}A(x,d)^2+\left(b-\frac{\sigma_r \sigma_{P_1}}{a}-\frac{\sigma_r^2}{2a^2}\right)A(x,d) +\left(\frac{\sigma_r \sigma_{P_1}}{a}+\frac{\sigma_r^2}{2a^2}- b-\delta+\mu\right)(d-x).
	\end{aligned}
\end{equation*}
Besides, $AL$ satisfies the following SDE:
 \begin{equation*}\begin{aligned}
		\rd AL(t)=&[P(t)(f_1(r(t))+\sigma_{P_1}f_2(r(t)))+\mu AL(t)]\rd t\\&+[P(t)f_2(r(t))+\sigma_{P_1}AL(t)]\rd W_r(t)+\sigma_{P_2}AL(t)\rd W_S(t),
	\end{aligned}
\end{equation*}
where \begin{equation*}
	\begin{aligned}
		f_1(r(t))=&\int_{m}^{d}  e^{-r(t)A(x,d)+D(x,d)}M(x)A(x,d)\left(\frac{\sigma_r^2}{2}A(x,d)-a(b-r(t))\right)\rd x,\\
		f_2(r(t))=&\sigma_r\int_{m}^{d}  e^{-r(t)A(x,d)+D(x,d)}M(x)A(x,d)\rd x.
	\end{aligned}
\end{equation*}
\end{proposition}
\begin{proof}
	See Appendix \ref{aALNC}.
\end{proof}	
We see that  different from $P$, $AL$ is no longer a geometric Brownian motion. Because  $\hat{\delta}(t)$ and $P$ are not independent, the closed form of $AL$ is very complicated. The dependence of $\hat{\delta}(t)$ and $P$ results in the term $P(t)(f_1(r(t))+\sigma_{P_1}f_2(r(t)))$ in the drift term and $P(t)f_2(r(t))$ in the volatility term of $W_r$, which are not linear. We will see that this dependence also brings difficulty in deriving the optimal strategies.
	
In this paper, we assume that the contribution rate is calculated by the spread method of fund amortization, which means that the supplementary contribution rate is proportional to the unfunded actuarial liability, i.e.,
	\[C(t)=NC(t)+SC(t)=NC(t)+k(AL(t)-F(t)),\]
where $k$ is a constant fixed in advance by the manager, representing the rate at which surplus or deficit is amortized. When $k>0$, if the fund is underfunded (overfunded), the contribution  will be larger (smaller) than the normal cost of the fund.

The DB pension fund receives contributions and has an outflow of promised benefits. Meanwhile, the fund manager participates in the financial market to hedge the risks. Denote the money invested in the cash, bond and stock at time $t$ by $u_0(t)$, $u_B(t)$ and $u_S(t)$, respectively ($ 0\leqslant t\leqslant T $).  The fund wealth $F=\left\{F(t)|0\leqslant t\leqslant T\right\}$ satisfies
	\begin{equation*}
		\begin{split}
			\mathrm{d}F(t)=u_0(t)\frac{\mathrm{d}S_0(t)}{S_0(t)}+u_B(t)\frac{\mathrm{d}B_K(t)}{B_K(t)}+u_S(t)\frac{\mathrm{d}S(t)}{S(t)}
			+(C(t)-P(t))\mathrm{d}t.
		\end{split}
	\end{equation*}
	The first three terms on the right side of the last equation represent the changes of fund wealth associated with the investment strategy. The last term on the right side of the last equation reflects the effect of the contributions and benefits on the fund wealth. Substituting Eqs.~(\ref{equ-S0}), (\ref{equ-B}) and (\ref{equ-S}) into the above equation, we derive the explicit form of the wealth based on the relation $F(t)=u_0(t)+u_B(t)+u_S(t)$ as follows:
	\begin{equation}\label{equ-X}
	\left\{	\begin{split}			\mathrm{d}F(t)=&r(t)F(t)\mathrm{d}t+[u_B(t)h(K)+u_S(t)\sigma_1](\lambda_r\mathrm{d}t+\mathrm{d}W_r(t))\\
			&+u_S(t)\sigma_2(\lambda_S\mathrm{d}t+\mathrm{d}W_S(t))+(C(t)-P(t))\mathrm{d}t,\\
 F(0)=&F_0.
		\end{split}\right.
	\end{equation}
	 Denote $u=\left\{u(t)\triangleq(u_B(t),u_S(t))|0\leqslant t\leqslant T\right\}$ as the investment strategy chosen by the fund manager. $u$ is said to be admissible if
     it satisfies the following conditions:\begin{enumerate}
		\item [(i)]$u$ is progressively measurable w.r.t. the filtration $\mathbb{F}$,
		\item [(ii)]$\mathbb{E}\left\{\int_0^T[(u_B(t)h(K)+u_S(t)\sigma_1)^2+u_S(t)^2\sigma_2^2]\mathrm{d}t\right\}<+\infty$,
		\item [(iii)]Eq.~(\ref{equ-X}) admits a unique strong solution starting with $F(t)=f$ for any  $(t,f)\in [0,T]\times \mathbb{R}$.
	\end{enumerate}
	\vskip 5pt
	Denote the set of all admissible investment strategies of $u$ by  $\Pi$. We are only concerned with the admissible strategies. The pension fund manager searches the optimal strategy within the admissible set under some optimization criterion.
	
   At each time $ t $, the pension fund is faced with a liability of $AL(t)$. The fund surplus at time $ t $ is given by $X(t)=F(t)-AL(t)$. The fund manager is concerned with the fund surplus when making a decision. By It\^{o}'s formula and based on the relationships among $AL(t)$, $C(t)$, $NC(t)$ and $F(t)$, we obtain the dynamics of the fund surplus in the following proposition.
\begin{proposition}\label{prop:x}
The evolution of $X=\left\{X(t)|0\leqslant t\leqslant T\right\}$ is as follows:
\begin{equation}\label{SDE-X}
	\left\{\begin{aligned}
		\mathrm{d}X(t)=&(r(t)-k)X(t)\rd t+\lambda_rP(t)f_2(r(t))\rd t+\left(\lambda_r\sigma_{P_1}+\lambda_S\sigma_{P_2}-\delta\right)P(t)f_0(r(t))\rd t\\
		&+[u_B(t)h(K)+u_S(t)\sigma_1-P(t)f_2(r(t))-\sigma_{P_1}P(t)f_0(r(t))](\lambda_r\mathrm{d}t+\mathrm{d}W_r(t))\\
		&+[u_S(t)\sigma_2-\sigma_{P_2}P(t)f_0(r(t))](\lambda_S\mathrm{d}t+\mathrm{d}W_S(t)),\\
		 X(0)=&X_0:=F_0-AL(0)
	\end{aligned}\right.
\end{equation}
with $f_0(r(t))=\int_{m}^{d}  e^{-r(t)A(x,d)+D(x,d)}M(x)\rd x $.
	\end{proposition}
	\begin{proof}
    See Appendix \ref{aprop:x}.
	\end{proof}
We see that the SDE (\ref{SDE-X}) of $X$ is very complicated. Different from { that in} \cite{JR10}, $P(t)$, $f_0(r(t))$ and $f_2(r(t))$ appear in the SDE of $X$. Note here that when $X_0>0$ ($X_0<0$), the pension fund is overfunded (underfunded) at initial time. The rate of interest for $X(t)$ becomes $r(t)-k$, which relies on the constant $k$. A positive $k$ results in a smaller interest rate charged on the unfunded liability.
\vskip 15pt
\setcounter{equation}{0}
\section{\bf Problem Formulation}
In this section, we formulate the problem of DB pension plan under  both underfunded and overfunded cases. Previous researches studied the underfunded and overfunded cases for the DB pension plan separately. In \cite{JJ18}, the authors suppose that the fund surplus starts in the underfunded (overfunded) region will never becomes overfunded (underfunded). However, in reality, either starting in the underfunded or overfunded case, the fund may become underfunded or overfunded because of the investment performance in the financial market. Meanwhile, the fund manager may also {sacrifice solvency} in the underfunded case to expect a larger return in the overfunded case. As such, it is more reasonable to consider these two cases together. First, we present the solvency risk (expected utility) considered in the underfunded (overfunded) case in previous researches.
\begin{itemize}
	\item[{\bf i }]{\bf Underfunded case}\vskip1pt \noindent
	In the underfunded case with $X_0<0$, \cite{JR10} and \cite{JJ18} consider minimizing the terminal solvency risk, i.e.,
	\begin{equation}\label{op1}
		\min \mathbb{E}[X^2(T)].
	\end{equation}
	\item[{\bf ii }]{\bf Overfunded case}\vskip1pt \noindent
In the overfunded case with $X_0>0$, \cite{JJ18} consider maximizing the expected CRRA utility of terminal wealth, i.e.,
	\begin{equation}\label{op2}
		\max \mathbb{E}[U(X(T))],
	\end{equation}
	where
	\[U(x)=\frac{x^{1-\gamma}}{1-\gamma}\]
	with $\gamma\in(0,1)$.
\end{itemize}

In \cite{JR10}, \cite{JJ18}, the optimal fund surplus evolves as the geometric Brownian motion. As such, the sign of $X(T)$ is the same as the sign of initial value $X_0$, and the authors consider these two cases separately.

In this paper, we suppose that either beginning with $X_0>0$ or $X_0<0$, the fund wealth at time $T$ may become positive or negative. The manager is concerned with the weighted sum of the solvency risk and the expected CRRA utility
	\begin{equation}\label{op3}\max \mathbb{E}\left\{-\alpha X^2(T)1_{\{X(T)<0\}} + \frac{X(T)^{1-\gamma}}{1-\gamma}1_{\{X(T)>0\}} \right\},
	\end{equation}
where $\alpha>0$ measures the attitude towards solvency risk. $\alpha$ can be viewed as the coefficient of tolerance level towards solvency risk. When $\alpha$ is large (small), the manager is more concerned with solvency risk (expected utility) in the underfunded (overfunded) region. Compared with the goal (\ref{op2}), starting from the overfunded region, the fund manager with the goal (\ref{op3}) can {sacrifice some solvency} to achieve a larger expected utility. Meanwhile, starting from the underfunded region, the fund manager in our paper may also enjoy some expected utility in the overfunded region. Besides, the fund manager can not {sacrifice too much solvency} and we impose a negative bond on the fund surplus at time $T$, i.e.,
\[X(T)\geqslant-B,\]
where $B>0$ is a constant.

Then the optimization problem of the pension manager under  both underfunded and overfunded cases is as follows
\begin{equation}\label{PX}
	\left\{\begin{aligned}
		\max_{u\in\Pi}~~&{\mathbb{E}\left[-\alpha X^2(T)1_{\{X(T)<0\}} + \frac{X(T)^{1-\gamma}}{1-\gamma}1_{\{X(T)>0\}}\right]},\\
		\text{s.t.}\quad&{X(t)}\text{ satisfies Eq.~(\ref{SDE-X}),}\\&X(T)\geqslant-B.
	\end{aligned}\right.
\end{equation}
We do not require $X_0>0$ or $X_0<0$ in this paper. Regardless of the initial case of the fund, the fund at time $T$ may become underfunded or overfunded. Problem (\ref{PX}) formulates the optimization problem for different coefficient $\alpha>0$. For each $\alpha>0$, we can {solve the optimization problem and} calculate the related solvency risk $\mathbb{E}[{X_\alpha^*}^2(T)1_{\{X^*(T)<0\}}]$ and expected utility   $ \mathbb{E}[U(X_\alpha^*(T))1_{\{X^*(T)>0\}}]$. As such, for each $\alpha$, we have a point $(\mathbb{E}[{X_\alpha^*}^2(T)1_{\{X^*(T)<0\}}], \mathbb{E}[U(X_\alpha^*(T))1_{\{X^*(T)>0\}}])$ on the plane. All points form a frontier which we also call it efficient frontier as in the mean-variance analysis.

\begin{figure}[tbph]
	\centering
	\includegraphics[width=0.5\linewidth]{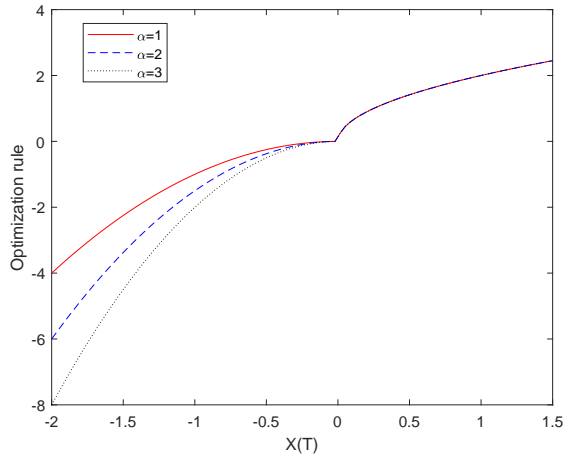}
	\caption{The optimization rule with $B=2$ and $\gamma=0.5$.}
	\label{fig1}
\end{figure}

\begin{remark}
	Problem (\ref{op3}) is concerned with a non-concave piece-wise function, which belongs to the form in \cite{carassus2009portfolio}. The manager is risk-averse towards overfunded wealth and underfunded wealth while the risk aversion of the manager is modified at the level 0. Fig.~\ref{fig1} shows the optimization rule  for different $\alpha$. In \cite{berkelaar2004optimal}, the authors consider a concave power utility function with a kink at the reference point, which supposes that the  manager is risk-averse towards losses and gains while has different attitudes. Besides, Problem (\ref{op3}) shows that the manager is more risk averse towards overfunded wealth, which can also be observed from Fig.~\ref{fig1}. Although the criterion in Problem (\ref{op3}) is different from the loss aversion towards losses in \cite{tversky1991loss} and \cite{tversky1992advances}, the criterion is consistent with the empirical results in \cite{march1996learning}, which reveals that that human beings exhibit greater risk aversion for gains than for losses in a wide variety of situations. Besides, we see from Fig.~\ref{fig1} that for a larger $\alpha$, the risk aversion attitude over losses increases.
\end{remark}

\vskip 15pt
\setcounter{equation}{0}
\section{\bf Optimal Portfolio and the Efficient Frontier}
We see that the optimization rule in Problem ($\ref{PX}$) is a non-concave piece-wise function of $X(T)$. The stochastic dynamic method can not be efficiently applied {to} this kind of problem. In this paper, we employ the martingale method to derive the optimal portfolio and efficient frontier. First, we see that Eq.~(\ref{SDE-X}) of $X$ is very complicated.  Eq.(\ref{SDE-X}) has additional terms both in the drift and volatility terms. To explicitly solve the optimization problem, we need to eliminate these additional terms first. Notice that the interest rate in Eq.~(\ref{SDE-X}) becomes $r(t)-k$. Observing the structure of Eq.~(\ref{SDE-X}), the pricing kernel $\rho=\!\{\rho(t)|0\!\leqslant\! t\!\leqslant\! T\}$  of the system is  given by
\begin{eqnarray}\label{equ:rho}
\left\{	\begin{split}
	\frac{\rd \rho(t)}{\rho(t)}&=-(r(t)-k)\rd t-\lambda_r \rd W_r(t)-\lambda_S\rd W_S(t),\\
 \rho(0)&= 1,
 \end{split}\right.
\end{eqnarray}
which is different from the real pricing kernel in the financial market with interest rate $r(t)$. When $k=0$, $\rho$ represents the real pricing kernel in the financial market. In the following, when we calculate the price of some derivative, the pricing kernel is $\rho(t)$, i.e., the price of payment $Z$ at maturity $s$ is $\mathbb{E}[Z\rho(s)/\rho(t)]$ at time $t\leqslant s$.

%\subsection{\bf Auxiliary processes and equivalent problem}

\subsection{\bf Equivalent problem}
The additional drift terms in Eq.~(\ref{SDE-X}) are $\lambda_rP(t)f_2(r(t))$ and $(\lambda_r\sigma_{P_1}+\lambda_S\sigma_{P_2}-\delta)P(t)f_0(r(t))$. The additional diffusion terms in Eq.~(\ref{SDE-X}) can be absorbed in the investment strategy and do not affect the structure of the fund surplus. We apply the derivative pricing theory to replicate the additional drift terms in  Appendix \ref{replication} and transform the original problem into an equivalent one.  We introduce two auxiliary processes $H$ and $\tilde{H}$ in Appendix \ref{replication} to replicate these two terms.

Next, we construct an auxiliary process $Y=\left\{Y(t)|0\leqslant t\leqslant T\right\} $ to transform the original problem (\ref{PX}) into an equivalent one with a self-financing wealth process. $Y$ is defined as
\begin{equation*}	Y(t)=X(t)+\lambda_rH(t,r(t),P(t),T)+\left(\lambda_r\sigma_{P_1}+\lambda_S\sigma_{P_2}-\delta\right)\tilde{H}(t,r(t),P(t),T),\quad 0\leqslant t\leqslant T\\
\end{equation*}
with $Y(0)=Y_0\!=\!X_0+\lambda_rH(0,r_0,P_0,T)+\left(\lambda_r\sigma_{P_1}\!+\!\lambda_S\sigma_{P_2}\!
-\!\delta\right)\tilde{H}(0,r_0,P_0,T)$. As $H(T,r(T),P(T),T)\!=\!0$ and $\tilde{H}(T,r(T),P(T),T)=0$, we have $Y(T)=X(T)$. Combining Eqs.~(\ref{SDE-X}), (\ref{sde:H}) and (\ref{equ:h2}), the differential form of $Y$ is obtained as follows:
\begin{equation}\label{equ:y1}
		\begin{aligned}
			\rd Y(t)=&(r(t)-k)Y(t)\rd t\\
			&+[u_B(t)h(K)+u_S(t)\sigma_1-P(t)f_2(r(t))-\sigma_{P_1}P(t)f_0(r(t))\\			&+\int_t^T\left(\lambda_r\sigma_{P_1}+\lambda_S\sigma_{P_2}-\delta\right)(\sigma_{P_1}
\tilde{g}(t,r(t),s)-\sigma_r\tilde{g}_r(t,r(t),s))\\&+\lambda_r(\sigma_{P_1}g(t,r(t),s)-\sigma_rg_r(t,r(t),s))\rd s](\lambda_r\mathrm{d}t+\mathrm{d}W_r(t))\\			&+[u_S(t)\sigma_2-\sigma_{P_2}P(t)f_0(r(t))\\&+\sigma_{P_2}\int_t^T\lambda_rg(t,r(t),s)+\left(\lambda_r\sigma_{P_1}+\lambda_S\sigma_{P_2}-\delta\right)\tilde{g}(t,r(t),s)\rd s](\lambda_S\mathrm{d}t+\mathrm{d}W_S(t)).
		\end{aligned}
	\end{equation}
We see that in Eq.~(\ref{equ:y1}), the additional drift terms in Eq.~(\ref{SDE-X}) are eliminated. The interest rate for process $Y$ is $r(t)-k$. The strategies $u_B$ and $u_S$ affect the diffusion terms in Eq.~(\ref{equ:y1}). In order to simplify the diffusion terms in Eq.~(\ref{equ:y1}), we introduce two strategies $\pi=\{\pi(t)=(\pi_1(t),\pi_2(t))|0\leqslant t\leqslant T\}$ for $Y$ as follows:
 \begin{equation}\label{piu}
\left\{ \begin{aligned}
			\pi_1(t)=&u_B(t)h(K)+u_S(t)\sigma_1-P(t)f_2(r(t))-\sigma_{P_1}P(t)f_0(r(t))\\
&+\int_t^T\lambda_r(\sigma_{P_1}g(t,r(t),s)-\sigma_rg_r(t,r(t),s))\\
&+\left(\lambda_r\sigma_{P_1}+\lambda_S\sigma_{P_2}-\delta\right)(\sigma_{P_1}\tilde{g}(t,r(t),s)
-\sigma_r\tilde{g}_r(t,r(t),s))\rd s,\\			\pi_2(t)=&u_S(t)\sigma_2-\sigma_{P_2}P(t)f_0(r(t))+\sigma_{P_2}\int_t^T\lambda_rg(t,r(t),s)\\
&+\left(\lambda_r\sigma_{P_1}+\lambda_S\sigma_{P_2}-\delta\right)\tilde{g}(t,r(t),s)\rd s.
\end{aligned}	\right.
	\end{equation}
We call $\pi$ admissible if the related strategy $u$ derived from Eq.~(\ref{piu}) belongs to $\Pi$. We also use the notation $\pi\in\Pi$ if the strategy $\pi$ is admissible. Combining Eqs.~(\ref{equ:y1}) and (\ref{piu}), $Y$ satisfies the following SDE
 \begin{equation}\label{SDE-Y}
		\rd Y(t)=(r(t)-k)Y(t)\rd t+\pi_1(t)(\lambda_r\mathrm{d}t+\mathrm{d}W_r(t))+\pi_2(t)(\lambda_S\mathrm{d}t+\mathrm{d}W_S(t)), Y(0)=Y_0.
	\end{equation}
As $Y(T)=X(T)$, Problem (\ref{PX}) is equivalent to the following problem w.r.t. $Y$.
\begin{equation}\label{PY}
	\left\{\begin{aligned}
		\max_{\pi\in\Pi}~~&{\mathbb{E}\left[-\alpha Y^2(T)1_{\{Y(T)<0\}} + \frac{Y(T)^{1-\gamma}}{1-\gamma}1_{\{Y(T)>0\}}\right]},\\
		\text{s.t.}\quad&{Y(t)}\text{ satisfies Eq.~(\ref{SDE-Y}),}\\&Y(T)\geqslant-B.
	\end{aligned}\right.
\end{equation}	

\subsection{\bf Optimal solution}
Problem (\ref{PY}) is not the traditional expected utility maximization problem. Problem (\ref{PY}) is concerned with a non-concave piece-wise optimization rule and also has a negative bound on the terminal wealth. As such, the stochastic dynamic programming method can not be applied. Based on \cite{cox1989optimal}, the martingale problem is
\begin{equation}\label{PY2}
	\left\{\begin{aligned}
		\max_{Z\geqslant-B,Z\in\mathcal{M}}\quad&{\mathbb{E}\left[-\alpha Z^21_{\{Z<0\}} + \frac{Z^{1-\gamma}}{1-\gamma}1_{\{Z>0\}}\right]},\\
		\text{s.t.}\quad& \mathbb{E}[\rho(T)Z]\leqslant Y_0,
	\end{aligned}\right.
\end{equation}	
where $\mathcal{M}$ represents the set of $\mathcal{F}_{T}- $measurable random variables.  The following lemma ensures the equivalence between Problem (\ref{PY}) and Problem (\ref{PY2}).
\begin{lemma}\label{TR1}
	Problem (\ref{PY}) is equivalent to  Problem (\ref{PY2}).
\end{lemma}
\begin{proof}
	Denote $ \{u_S^*, u_B^*\} $ as the solution to Problem (\ref{PY}). Let $ Y^*$ be the auxiliary process corresponding to $ \{u_S^*, u_B^*\} $. It is obvious that $ \{\rho(t)Y^*(t)|0\leqslant t\leqslant T\} $ is a local martingale with lower bound and thus is a super-martingale with respect to $ \mathbb{F} $, thus we have $ \mathbb{E}[\rho(T)Y^*(T)]\leqslant Y_0 $. Besides, $ Y^*(T) $ attains the maximal of the optimization goal, i.e., $ Y^*(T) $ is the solution to Problem (\ref{PY2}).
	
	On the other hand, for the solution $ Z^* $ of Problem (\ref{PY2}), let $ Y^*(t)=\rho(t)^{-1} \mathbb{E}[\rho(T)Z^*|\mathcal{F}_t]$. Then $ \rho(t)Y^*(t) $ is a martingale with respect to $ \mathbb{F} $. As such, by the martingale representation theorem, we obtain $$\rd \rho(t)Y^*(t)=A\rd W_r(t) +B\rd W_S(t)$$ for some {integrable} processes $ A $ and $B$. Based on Eqs.~(\ref{equ:rho}) and (\ref{SDE-Y}), we have
	\begin{equation}\label{rhoY}
		\rd \rho(t)Y^*(t)=\rho(t)[(\pi_1(t)-\lambda_rY^*(t))\rd W_r(t)+(\pi_2(t)-\lambda_SY^*(t))\rd W_S(t)].
	\end{equation}
	Comparing  the last two equations, we have the following results for Problem (\ref{PY})
	\begin{equation*}
		\left\{\begin{aligned}
		\pi_1^*(t)=&\rho(t)^{-1}A+\lambda_rY^*(t),\\
		\pi_2^*(t)=&\rho(t)^{-1}B+\lambda_SY^*(t).
		\end{aligned}\right.
	\end{equation*} Based on Eq.~(\ref{piu}), the solutions $ u_B^* $ and $ u_S^* $  to the original problem (\ref{PX}) can be obtained.
\end{proof}
Lemma \ref{TR1} also shows the procedure to derive the optimal strategy of Problem (\ref{PY}) from the solution to Problem (\ref{PY2}). Problem (\ref{PY2}) is a static maximization problem w.r.t. some random variable.  However, Problem (\ref{PY2}) imposes a  budget constraint on the random variable and is a constrained optimization problem. To eliminate the budget constraint, we employ the Lagrange dual method first. For simplicity, denote
\begin{equation*}
	f(x)=-\alpha x^21_{\{x<0\}} + \frac{x^{1-\gamma}}{1-\gamma}1_{\{x>0\}}.
\end{equation*}
By the Lagrange dual method, Problem (\ref{PY2}) can be transformed to the following problem with a Lagrange multiplier $ \beta\geqslant0 $
\begin{equation}\label{PY3}
	\max_{Z\geqslant-B,Z\in\mathcal{M}}\quad{\mathbb{E}\left[f(Z)-\beta\rho(T)Z \right]}.
\end{equation}
In order to solve  Problem (\ref{PY3}), we only need to solve the following equivalent {deterministic} problem
\begin{equation}\label{DPY3}
	\max_{x\geqslant-B}\quad\{f(x)-yx\}.
\end{equation}	
The following lemma shows the results for Problem (\ref{DPY3}).
\begin{lemma}\label{lemma:xy}
Denote $ k_1= (\frac{4\alpha\gamma}{1-\gamma})^{\frac{\gamma}{1+\gamma}} $. There are two cases for the solution to Problem (\ref{DPY3}).
\begin{itemize}
	\item[(i)]If $k_1<2\alpha B $, the solution to Problem (\ref{DPY3}) is  given by
	\begin{equation*}
		x^*(y)=\left\{\begin{aligned}
			&y^{-\frac{1}{\gamma}},&0<y< k_1,\\
			&k_1^{-\frac{1}{\gamma}} \text{ or } -\frac{k_1}{2\alpha},&y=k_1\\
			&-\frac{y}{2\alpha},&k_1<y\leqslant 2\alpha B,\\
			&-B,&y> 2\alpha B.
		\end{aligned}\right.
	\end{equation*}
	\item[(ii)]If $k_1\geqslant2\alpha B $, define $ k_2=z_0^{-\gamma} $,  $ z_0 $ is the solution of the following equation
\begin{equation*}
		\frac{z^{1-\gamma}}{1-\gamma}+\alpha B^2=z^{-\gamma}(z+B),
\end{equation*}
then the solution to Problem (\ref{DPY3}) is
	\begin{equation*}
		x^*(y)=\left\{\begin{aligned}
			&y^{-\frac{1}{\gamma}},&0<y< k_2,\\
			&k_2^{-\frac{1}{\gamma}}\text{ or }-B,&y= k_2\\
			&-B,&y>k_2.
		\end{aligned}\right.
	\end{equation*}
\end{itemize}
\end{lemma}
\begin{proof}
See Appendix \ref{alemma:xy}.
\end{proof}
We see from Lemma \ref{lemma:xy} that when $\alpha$ is large, the solution is presented in Case (i).  Then the solution to Problem (\ref{PY3}) can be directly obtained by the above lemma. The following lemma shows the solution to Problem (\ref{PY3}).
\begin{lemma}\label{lemma:py3}
For $ y>0 $, let the Borel measurable function $ x^*(y) $ be the optimal solution to Problem (\ref{DPY3}). Then \begin{equation}\label{Z}
	Z^*(\beta)=x^*(\beta \rho(T))
\end{equation} is the optimal solution to Problem (\ref{PY3}).
\end{lemma}
\begin{proof}
It is easy to see that $ Z^*(\beta)=x^*(\beta \rho(T)) $ satisfies the constraint $ Z\geqslant-B $ and $ Z\in\mathcal{M} $,  we only need to show the optimality of $ x^*(\beta \rho(T)) $.  Due to the optimality of $ x^*(\cdot) $ to Problem (\ref{DPY3}), for any Z satisfying the constraint $ Z\geqslant-B $ and $ Z\in\mathcal{M} $, we have
\begin{align*}
	\mathbb{E}\left[f(Z)-\beta\rho(T)Z \right]
	=&\int_{\Omega}f(Z(\omega))-\beta\rho(T)(\omega)Z(\omega)\mathbb{P}(\rd\omega)\\
	\leqslant&\int_{\Omega}f(x^*(\beta \rho(T)(\omega)))-\beta\rho(T)(\omega)x^*(\beta \rho(T)(\omega))\mathbb{P}(\rd\omega)\\
	=&\mathbb{E}\left[f(x^*(\beta \rho(T)))-\beta\rho(T)x^*(\beta \rho(T)) \right],
\end{align*}
thus the proof follows.	
\end{proof}
\begin{remark}
Note that the solution to Problem (\ref{DPY3}) is not unique according to Lemma \ref{lemma:xy}. However, since $\mathbb{P}\left(\beta \rho(T)\in\{k_1,k_2\}\right)=0$, the optimal solution to Problem (\ref{PY3}) is  almost surely unique according to Lemma \ref{lemma:py3}.
\end{remark}
Problem (\ref{PY2}) can be solved by its Lagrange dual problem (\ref{PY3}), which has been solved in Lemma \ref{lemma:py3}. The following theorem shows the solution of Problem (\ref{PY2}). Denote
\begin{align}
	&M\triangleq(k-b-\frac12(\lambda_r^2+\lambda_S^2))T+(b-r_0)\frac{1-e^{-aT}}{a},\\
	&V\triangleq\sqrt{\left(\left(\frac{\sigma_r}{a}-\lambda_r\right)^2+
\lambda_S^2\right)T-2\frac{\sigma_r}{a}\left(\frac{\sigma_r}{a}-\lambda_r\right)
\frac{1-e^{-aT}}{a}+\frac{\sigma^2_r}{2a^2}\frac{1-e^{-2aT}}{a}}.
\end{align}
\begin{theorem}\label{TR2}
\begin{itemize}
\item[(i)]
	When $ Y_0>-B\exp\left(M+\frac{1}{2}V^2\right) $, there exists a unique $ \beta^*\in(0,+\infty) $ such that $ Z^*({\beta^*}) $ is the optimal solution of Problem (\ref{PY3}) with
\begin{equation}\label{beta}
			\mathbb{E}[\rho(T)Z^*(\beta^*)]= Y_0,
\end{equation}
and $ Z^*\triangleq Z^*({\beta^*}) $ is the optimal  solution to Problem (\ref{PY2}).
\item[(ii)]When  $ Y_0=-B\exp\left(M+\frac{1}{2}V^2\right) $, we have $ \beta^*=+\infty $  and $  Z^*(\beta^*)=x^*(\beta^* \rho(T))=-B $.
\item[(iii)] When  $ Y_0<-B\exp\left(M+\frac{1}{2}V^2\right) $, there is no solution to Problem (\ref{PY2}).
	\end{itemize}
\end{theorem}
\begin{proof}
Define a function $ g(\beta)\triangleq\mathbb{E}[\beta\rho(T)x^*(\beta \rho(T))] $. Based on the structure of $ x^*(\cdot) $, we see that $ g(\cdot) $ is continuous and monotonic. Besides,  $\displaystyle \lim_{\beta\to\infty}g(\beta)=-B\exp\left(M+\frac{1}{2}V^2\right) $ and $\displaystyle \lim_{\beta\to0}g(\beta)=+\infty $. As such, if $ Y_0\geqslant-B\exp\left(M+\frac{1}{2}V^2\right) $, there exists $ \beta^*\in(0,+\infty] $ such that $ Z^*({\beta^*}) $ is the optimal solution of Problem (\ref{PY3}). If  $ Y_0<-B\exp\left(M+\frac{1}{2}V^2\right) $, the constraints $ Z\geqslant-B $ and $ \mathbb{E}[\rho(T)Z]\leqslant Y_0 $ in Problem (\ref{PY2}) are incompatible and there will be no solution to Problem (\ref{PY2}).
	
If there exists $ \beta^*\in(0, +\infty] $ such that $ Z^*({\beta^*}) $ is the optimal solution of Problem (\ref{PY3}) satisfying Eq.~(\ref{beta}), then for any $ Z $ satisfying the constraints in Problem (\ref{PY2}), we have
\begin{eqnarray*}
\mathbb{E}[f(Z)]&\leqslant&\mathbb{E}[f(Z)]+\beta^*( Y_0-\mathbb{E}[\rho(T)Z])\\
			&=&\mathbb{E}[f(Z)-\beta^*\rho(T)Z]+\beta^*Y_0\\
			&\leqslant &\mathbb{E}[f(Z^*({\beta^*}))-\beta^*\rho(T)Z^*({\beta^*})]+\beta^*Y_0\\
			&=&\mathbb{E}[f(Z^*({\beta^*}))].
\end{eqnarray*}
		Thus, $ Z^*\triangleq Z^*({\beta^*}) $ is the optimal  solution to Problem (\ref{PY2}).
\end{proof}
 In Case (i) of Theorem \ref{TR2}, the optimal terminal wealth is obtained by $Y^*(T)=x^*(\beta^* \rho(T))$. By Lemmas \ref{lemma:xy} and \ref{lemma:py3}, $Y^*(T)=x^*(\beta^* \rho(T))$ has two forms in Case (i) corresponding to $k_1<2\alpha B$ and  $k_1\geqslant 2\alpha B$. As $ k_1= (\frac{4\alpha\gamma}{1-\gamma})^{\frac{\gamma}{1+\gamma}} $, equation $(\frac{4\alpha\gamma}{1-\gamma})^{\frac{\gamma}{1+\gamma}}=2\alpha B$ admits a solution $\alpha^*$. We can show that when $\alpha>\alpha^*$ ($\alpha<\alpha^*$), $k_1<2\alpha B$ ($k_1>2\alpha B$) {holds}. As such, when $k_1<2\alpha B$ ($k_1\geqslant 2\alpha B$), $\alpha$ is relatively large (small) and the manager has low (high) tolerance towards solvency risk.  In the case of high tolerance level towards solvency risk, the manager expects a constant $-B$  of terminal wealth in the underfunded region to achieve high expected utility in the overfunded region. The optimal terminal wealth in the overfunded region is the same as the results with optimization rule (\ref{op2}) and the optimal terminal wealth has a two-region form. In the case of low tolerance level towards solvency risk, the manager should take the solvency risk into account. The optimal terminal wealth is divided into three situations: result with optimization rule (\ref{op2}), result with optimization rule (\ref{op1}), and the negative bound. As such, when the manager has a low tolerance level towards solvency risk, the manager will not accept too large a solvency risk and then the optimal terminal wealth has a three-region form.

 In Case (ii) of Theorem \ref{TR2},  the negative bound $-B$ is high for the manager. A  simple calculation shows that $\exp\left(M+\frac{1}{2}V^2\right) $ is the price of a zero-coupon bond with maturity $T$ in a financial system with pricing kernel $\rho(t)$. We see that when $Y_0=-B\exp\left(M+\frac{1}{2}V^2\right) $, the lower bound on the terminal fund level is so large that replicating it from the initial fund level is the only feasible solution. In Case (iii) of Theorem \ref{TR2}, the lower bound is {too} large and the terminal wealth can not be maintained above $-B$. As such, the optimization problem is not well-posed and the problem has no solution.

After obtaining the optimal solution to the martingale problem (\ref{PY2}), we need to obtain the solution to the original problem (\ref{PY}). Theorem \ref{TR2} shows the optimal terminal wealth of Problem (\ref{PY}). Based on the replication {procedure in Lemma} \ref{TR1}, we  obtain the optimal wealth process at time $t$ and optimal strategies for Problem (\ref{PY}). We introduce the following notations first.
\begin{eqnarray*}
&&\kappa(t)\triangleq\exp\left(\frac{1-e^{-a(T-t)}}{a}(b-r(t))\right),\ \ \tilde{M}(t)\triangleq(k-b-\frac12(\lambda_r^2+\lambda_S^2))(T-t),\\
&&\Upsilon(x,\beta)\triangleq\frac{\ln(x)-\ln(\beta)-M}{V},\ \ \Gamma(x,\beta,t)\triangleq\frac{\ln(x)-\ln(\beta\rho(t)\kappa(t))-\tilde{M}(t)}{\tilde{V}(t)},\\
&&\tilde{V}(t)\triangleq\sqrt{\left(\left(\frac{\sigma_r}{a}-
\lambda_r\right)^2+\lambda_S^2\right)(T-t)-2\frac{\sigma_r}{a}\left(\frac{\sigma_r}{a}-
\lambda_r\right)\frac{1-e^{-a(T-t)}}{a}+\frac{\sigma^2_r}{2a^2}\frac{1-e^{-2a(T-t)}}{a}}.
\end{eqnarray*}
Meanwhile, denote $\phi(\cdot)$ and $\Phi(\cdot)$ as the probability density function and cumulative distribution function of standard normal distribution, respectively.
\begin{theorem}\label{th:y}
There are four cases for the optimal auxiliary process $Y^*$ and optimal strategies $\pi^*=(\pi_1^*,\pi_2^*)$ as follows:
\begin{itemize}
\item[(i)]{\bf Low tolerance towards solvency risk}
\\If $k_1<2\alpha B$ and $Y_0>-B\exp\left(M+\frac{1}{2}V^2\right)$, the optimal auxiliary process is
\!\!\!\!\begin{equation*}
\begin{aligned}	Y^*(t)&=\exp\left(\frac{\gamma-1}{\gamma}\tilde{M}(t)\!+\!\frac12\left(\frac{\gamma-1}{\gamma}\tilde{V}(t)\right)^2\right){\beta^*}^{-\frac{1}{\gamma}}\rho(t)^{-\frac{1}{\gamma}}\kappa(t)^{1-\frac{1}{\gamma}}\Phi\left(\Gamma(k_1,{\beta^*},t)-\frac{\gamma-1}{\gamma}\tilde{V}(t)\right)\\
	&-\frac{{\beta^*}}{2\alpha}\exp(2(\tilde{M}(t)+\tilde{V}(t)^2))\rho(t)\kappa(t)^{2}(\Phi(\Gamma(2\alpha B,{\beta^*},t)-2\tilde{V}(t))-\Phi(\Gamma(k_1,{\beta^*},t)-2\tilde{V}(t)))\\&-B\kappa(t)\exp(\tilde{M}(t)+\tilde{V}(t)^2)(1-\Phi(\Gamma(2\alpha B,{\beta^*},t)-\tilde{V}(t))),
\end{aligned}
\end{equation*}
where $\beta^*$ satisfies the budget constraint $\mathbb{E}[\rho(T)Y^*(T)]= Y_0$. The optimal strategies are
\!\!\!
\begin{eqnarray*}
\begin{aligned}		
\!\!\pi_1^*(t)&=\!\!\Bigg\{\frac{\gamma\!-\!1}{\gamma}\exp\left(\frac{\gamma\!-\!1}
{\gamma}\tilde{M}(t)\!+\!\frac12\left(\frac{\gamma\!-\!1}{\gamma}\tilde{V}(t)\right)^2\right)
{\beta^*}^{-\frac{1}{\gamma}}\rho(t)^{-\frac{1}{\gamma}}\kappa(t)^{1\!-\!\frac{1}{\gamma}}
\Phi\left(\Gamma(k_1,{\beta^*},t)\!-\!\frac{\gamma\!-\!1}{\gamma}\tilde{V}(t)\right)\\		&\!-\!\frac{1}{\tilde{V}(t)}\exp\left(\frac{\gamma\!-\!1}{\gamma}\tilde{M}(t)\!+\!
\frac12\left(\frac{\gamma\!-\!1}{\gamma}\tilde{V}(t)\right)^2\right){\beta^*}^{-\frac{1}
{\gamma}}\rho(t)^{-\frac{1}{\gamma}}\kappa(t)^{1\!-\!\frac{1}{\gamma}}
\phi\left(\Gamma(k_1,{\beta^*},t)\!-\!\frac{\gamma-1}{\gamma}\tilde{V}(t)\right)\\
&-\frac{{\beta^*}}{\alpha}\exp(2(\tilde{M}(t)+\tilde{V}(t)^2))\rho(t)\kappa(t)^{2}(\Phi(\Gamma(2\alpha B,{\beta^*},t)-2\tilde{V}(t))-\Phi(\Gamma(k_1,{\beta^*},t)-2\tilde{V}(t)))
\end{aligned}
\end{eqnarray*}
\!\!\!
\begin{eqnarray*}
\begin{aligned}
&+\frac{{\beta^*}}{2\alpha\tilde{V}(t)}\exp(2(\tilde{M}(t)+
\tilde{V}(t)^2))\rho(t)\kappa(t)^{2}(\phi(\Gamma(2\alpha B,{\beta^*},t)-2\tilde{V}(t))-\phi(\Gamma(k_1,{\beta^*},t)-2\tilde{V}(t)))\\
&-B\kappa(t)\exp(\tilde{M}(t)+\tilde{V}(t)^2)(1-\Phi(\Gamma(2\alpha B,{\beta^*},t)-\tilde{V}(t)))\\
&-\frac{1}{\tilde{V}(t)}B\kappa(t)\exp(\tilde{M}(t)+\tilde{V}(t)^2)\phi(\Gamma(2\alpha B,{\beta^*},t)-\tilde{V}(t))\Bigg\}(\sigma_rA(t,T)-\lambda_r)+\lambda_rY^*(t)
\end{aligned}
\end{eqnarray*}
and
\begin{equation*}
\begin{aligned}		\!\!\!\!\pi_2^*(t)&=\!\!\Bigg\{\frac{\gamma\!-\!1}{\gamma}\exp\left
(\frac{\gamma\!-\!1}{\gamma}\tilde{M}(t)\!+\!\frac12\left(\frac{\gamma\!-\!1}{\gamma}\tilde{V}(t)\right)^2\right)
{\beta^*}^{-\frac{1}{\gamma}}\rho(t)^{-\frac{1}{\gamma}}\kappa(t)^{1\!-\!\frac{1}{\gamma}}
\Phi\left(\Gamma(k_1,{\beta^*},t)-\frac{\gamma\!-\!1}{\gamma}\tilde{V}(t)\right)\\
&-\frac{1}{\tilde{V}(t)}\exp\left(\frac{\gamma-1}{\gamma}\tilde{M}(t)+\frac12\left(\frac{\gamma-1}{\gamma}\tilde{V}(t)\right)^2\right){\beta^*}^{-\frac{1}{\gamma}}\rho(t)^{-\frac{1}{\gamma}}\kappa(t)^{1-\frac{1}{\gamma}}\phi\left(\Gamma(k_1,{\beta^*},t)-\frac{\gamma-1}{\gamma}\tilde{V}(t)\right)\\
&-\frac{{\beta^*}}{\alpha}\exp(2(\tilde{M}(t)+\tilde{V}(t)^2))\rho(t)\kappa(t)^{2}(\Phi(\Gamma(2\alpha B,{\beta^*},t)-2\tilde{V}(t))-\Phi(\Gamma(k_1,{\beta^*},t)-2\tilde{V}(t)))\\&+\frac{{\beta^*}}{2\alpha\tilde{V}(t)}\exp(2(\tilde{M}(t)+\tilde{V}(t)^2))\rho(t)\kappa(t)^{2}(\phi(\Gamma(2\alpha B,{\beta^*},t)-2\tilde{V}(t))-\phi(\Gamma(k_1,{\beta^*},t)-2\tilde{V}(t)))\\&-B\kappa(t)\exp(\tilde{M}(t)+\tilde{V}(t)^2)(1-\Phi(\Gamma(2\alpha B,{\beta^*},t)-\tilde{V}(t)))\\&-\frac{1}{\tilde{V}(t)}B\kappa(t)\exp(\tilde{M}(t)+\tilde{V}(t)^2)\phi(\Gamma(2\alpha B,{\beta^*},t)-\tilde{V}(t))\Bigg\}(-\lambda_S)+\lambda_SY^*(t).
	\end{aligned}
\end{equation*}
\item[(ii)]{\bf High tolerance towards solvency risk}
\\If $k_1\geqslant2\alpha B $ and $Y_0>-B\exp\left(M+\frac{1}{2}V^2\right)$,  the optimal auxiliary process is
\begin{equation*}
\begin{aligned}		Y^*(t)&=\exp\left(\frac{\gamma-1}{\gamma}\tilde{M}(t)+\frac12\left(\frac{\gamma-1}{\gamma}\tilde{V}(t)\right)^2\right){\beta^*}^{-\frac{1}{\gamma}}\rho(t)^{-\frac{1}{\gamma}}\kappa(t)^{1-\frac{1}{\gamma}}\Phi\left(\Gamma(k_2,{\beta^*},t)-\frac{\gamma-1}{\gamma}\tilde{V}(t)\right)\\
&-B\kappa(t)\exp(\tilde{M}+\tilde{V}^2)(1-\Phi(\Gamma(k_2,{\beta^*},t)-\tilde{V})),
	\end{aligned}
\end{equation*}
where $\beta^*$ satisfies the budget constraint $\mathbb{E}[\rho(T)Y^*(T)]= Y_0$. The optimal strategies are
\!\!\!\begin{equation*}
\begin{aligned}
\!\!\!\!\!\!\pi_1^*(t)&=\!\!\Bigg\{\frac{\gamma\!-\!1}{\gamma}\exp\left(\frac{\gamma\!-\!1}{\gamma}\tilde{M}(t)
\!+\!\frac12\left(\frac{\gamma\!-\!1}{\gamma}\tilde{V}(t)\right)^2\right)
{\beta^*}^{-\frac{1}{\gamma}}\rho(t)^{-\frac{1}{\gamma}}\kappa(t)^{1\!-\!\frac{1}{\gamma}}
\Phi\left(\Gamma(k_2,{\beta^*},t)-\frac{\gamma\!-\!1}{\gamma}\tilde{V}(t)\right)\\
&-\frac{1}{\tilde{V}}\exp\left(\frac{\gamma-1}{\gamma}\tilde{M}(t)+\frac12\left(\frac{\gamma-1}{\gamma}\tilde{V}(t)\right)^2\right){\beta^*}^{-\frac{1}{\gamma}}\rho(t)^{-\frac{1}{\gamma}}\kappa(t)^{1-\frac{1}{\gamma}}\phi\left(\Gamma(k_2,{\beta^*},t)-\frac{\gamma-1}{\gamma}\tilde{V}(t)\right)\\
&-B\kappa(t)\exp(\tilde{M}+\tilde{V}^2)(1-\Phi(\Gamma(k_2,{\beta^*},t)-\tilde{V}))\\
&-\frac{1}{\tilde{V}}B\kappa(t)\exp(\tilde{M}+\tilde{V}^2)\phi(\Gamma(k_2,{\beta^*},t)
-\tilde{V})\Bigg\}(\sigma_rA(t,T)-\lambda_r)+\lambda_rY^*(t)
\end{aligned}
\end{equation*}
and
\begin{equation*}
\begin{aligned}
\!\!\!\!\pi_2^*(t)\!\!&=\!\!\Bigg\{\frac{\gamma\!-\!1}{\gamma}
\exp\left(\frac{\gamma\!-\!1}{\gamma}\tilde{M}(t)\!+\!
\frac12\left(\frac{\gamma\!-\!1}{\gamma}\tilde{V}(t)\right)^2\right)
{\beta^*}^{-\frac{1}{\gamma}}\rho(t)^{-\frac{1}{\gamma}}
\kappa(t)^{1\!-\!\frac{1}{\gamma}}\Phi\left(\Gamma(k_2,{\beta^*},t)\!-\!\frac{\gamma-1}{\gamma}\tilde{V}(t)\right)\\
&-\frac{1}{\tilde{V}}\exp\left(\frac{\gamma\!-\!1}{\gamma}\tilde{M}(t)\!+\!
\frac12\left(\frac{\gamma\!-\!1}{\gamma}\tilde{V}(t)\right)^2\right)
{\beta^*}^{-\frac{1}{\gamma}}\rho(t)^{-\frac{1}{\gamma}}\kappa(t)^{1\!-\!
\frac{1}{\gamma}}\phi\left(\Gamma(k_2,{\beta^*},t)-\frac{\gamma-1}{\gamma}\tilde{V}(t)\right)\\
&-B\kappa(t)\exp(\tilde{M}+\tilde{V}^2)(1-\Phi(\Gamma(k_2,{\beta^*},t)-\tilde{V}))\\&-\frac{1}{\tilde{V}}B\kappa(t)\exp(\tilde{M}+\tilde{V}^2)\phi(\Gamma(k_2,{\beta^*},t)-\tilde{V})\Bigg\}(-\lambda_S)+\lambda_SY^*(t).
	\end{aligned}
\end{equation*}
\item[(iii)]{\bf A specific lower bound}
\\When  $ Y_0=-B\exp\left(M+\frac{1}{2}V^2\right) $, the optimal auxiliary process is \begin{equation*}
	Y^*(t)=-B\kappa(t)\exp(\tilde{M}(t)+\frac12\tilde{V}(t)^2).
\end{equation*}
The optimal strategies are
\begin{equation*}
	\pi_1^*(t)=-B\kappa(t)\exp(\tilde{M}(t)+\frac12\tilde{V}(t)^2)\sigma_rA(t,T)
\end{equation*}and \begin{equation*}
	\pi_2^*(t)=0.
\end{equation*}
\item[(iv)]{\bf High lower bound}
\\When  $ Y_0<-B\exp\left(M+\frac{1}{2}V^2\right) $,  Problem (\ref{PY}) admits no solution.
\end{itemize}
\end{theorem}
\begin{proof}
The calculation is tedious and we only show the procedure to calculate the optimal wealth and strategies. The procedure has also been shown in the proof of Lemma \ref{TR1}. Theorem \ref{TR2} shows the results of the optimal terminal wealth $Y^*(T)$ in different cases. The optimal wealth  $Y^*(t)$ at time $t$ can be directly obtained by taking the conditional expectation
\[Y^*(t)=\rho(t)^{-1}\mathbb{E}[\rho(T)Y^*(T)|\mathcal{F}_t].\]
After deriving the explicit form of $Y^*(t)$, the differential form of $Y^*(t)$ is   also easily obtained. Then comparing the diffusion terms with {that in} the original equation (\ref{SDE-Y}), the optimal investment strategies can be obtained.
\end{proof}
We see from a specific lower bound case that the manager should purchase a zero-coupon bond with a maturity date $T$ and avoid allocation in the stock.  As stated before, {solving} the original problem (\ref{PX}) is equivalent to solving Problem (\ref{PY}). After obtaining the optimal solution to Problem (\ref{PY}), Problem (\ref{PX}) is directly solved.
\begin{theorem}\label{th:x}
The optimal wealth process and investment strategies for Problem (\ref{PX}) are
\[X^*(t)=Y^*(t)-\lambda_rH(t,r(t),P(t),T)-
\left(\lambda_r\sigma_{P_1}+\lambda_S\sigma_{P_2}-\delta\right)\tilde{H}(t,r(t),P(t),T)\]
and
\begin{equation}\label{U1}
	\left\{\begin{aligned}
		u_S^*(t)=\frac{1}{\sigma_2}\bigg\{\sigma_{P_2}&\int_t^T\lambda_rg(t,r(t),s)+\left(\lambda_r\sigma_{P_1}+\lambda_S\sigma_{P_2}-\delta\right)\tilde{g}(t,r(t),s)\rd s\\&\qquad\qquad\qquad\qquad\qquad\qquad\qquad\quad+\pi_2^*(t)+\sigma_{P_2}P(t)f_0(r(t))\bigg\},\\
		u_B^*(t)=\frac{1}{h(K)}\Bigg\{&\pi_1^*(t)+P(t)f_2(r(t))+\sigma_{P_1}P(t)f_0(r(t))\\&-\int_t^T\left(\lambda_r\sigma_{P_1}+\lambda_S\sigma_{P_2}-\delta\right)(\sigma_{P_1}\tilde{g}(t,r(t),s)-\sigma_r\tilde{g}_r(t,r(t),s))\\&\qquad\qquad\qquad\qquad\qquad+\lambda_r(\sigma_{P_1}g(t,r(t),s)-\sigma_rg_r(t,r(t),s))\rd s\\&-\frac{\sigma_1}{\sigma_2}\bigg\{\sigma_{P_2}\int_t^T\lambda_rg(t,r(t),s)+\left(\lambda_r\sigma_{P_1}+\lambda_S\sigma_{P_2}-\delta\right)\tilde{g}(t,r(t),s)\rd s\\&\qquad\qquad\qquad\qquad\qquad\qquad\qquad\qquad+\pi_2^*(t)+\sigma_{P_2}P(t)f_0(r(t))\bigg\}
		\Bigg\},
	\end{aligned}\right.
\end{equation}
where $Y^*(t)$, $\pi_1^*(t)$ and $\pi^*_2(t)$ are shown in Theorem \ref{th:y}.
\end{theorem}
\begin{proof}
	Using the relationship $$Y(t)=X(t)+\lambda_rH(t,r(t),P(t),T)
+\left(\lambda_r\sigma_{P_1}+\lambda_S\sigma_{P_2}-\delta\right)\tilde{H}(t,r(t),P(t),T)$$ and	Eq.~(\ref{piu}), the proof easily follows.
\end{proof}
Besides, we can also calculate the expected utility and solvency risk for the pension fund at time $T$.
\begin{theorem}\label{thm:ef}
\begin{itemize}
	\item[(i)]{\bf Low tolerance towards solvency risk}
	\\If $k_1<2\alpha B $ and $Y_0>-B\exp\left(M+\frac{1}{2}V^2\right)$, the solvency risk in the underfunded region is
	\begin{equation*}
		\begin{aligned}
			{\mathbb{E}\left[ {X^*}^2(T)1_{\{{X^*}(T)<0\}}\right]}=&\frac{{\beta^*}^2}{4\alpha^2}\exp(2(M+V^2))(\Phi(\Upsilon(2\alpha B,\beta^*)-2V)-\Phi(\Upsilon(k_1,\beta^*)-2V))\\&+ B^2(1-\Phi(\Upsilon(2\alpha B,\beta^*))).
		\end{aligned}
	\end{equation*}
	The expected utility in the overfunded region is
	\begin{equation*}
	\begin{aligned}
		\mathbb{E}\left[ \frac{X^*(T)^{1-\gamma}}{1-\gamma}1_{\{X^*(T)>0\}}\right]=\frac{{\beta^*}^{\frac{\gamma-1}{\gamma}}}{1-\gamma}\exp\left(\frac{\gamma-1}{\gamma}M+\frac{V^2(\gamma-1)^2}{2\gamma^2}\right)\Phi\left(\Upsilon(k_1,\beta^*)-V\frac{\gamma-1}{\gamma}\right).
	\end{aligned}
\end{equation*}
	\item[(ii)]{\bf High tolerance towards solvency risk}
	\\If $k_1\geqslant2\alpha B $ and $Y_0>-B\exp\left(M+\frac{1}{2}V^2\right)$, the solvency risk in the underfunded region is 	\begin{equation*}
		\begin{aligned}
			{\mathbb{E}\left[ {X^*}^2(T)1_{\{{X^*}(T)<0\}}\right]}
			= B^2(1-\Phi(\Upsilon(k_2,\beta^*))).
		\end{aligned}
	\end{equation*}
	The expected utility in the overfunded region is
	\begin{equation*}
	\begin{aligned}
		\mathbb{E}\left[ \frac{X^*(T)^{1-\gamma}}{1-\gamma}1_{\{X^*(T)>0\}}\right]
		=\frac{{\beta^*}^{\frac{\gamma-1}{\gamma}}}{1-\gamma}\exp\left(\frac{\gamma-1}{\gamma}M+\frac{V^2(\gamma-1)^2}{2\gamma^2}\right)\Phi\left(\Upsilon(k_2,\beta^*)-\frac{\gamma-1}{\gamma}V\right).\\\end{aligned}
\end{equation*}	
\item [(iii)]{\bf A specific lower bound}\\
If $Y_0=-B\exp\left(M+\frac{1}{2}V^2\right)$, the solvency risk in the underfunded region is $B^2$ and the expected utility in the overfunded region is 0.
\item[(iv)]{\bf High lower bound}\\If $Y_0<-B\exp\left(M+\frac{1}{2}V^2\right)$, Problem (\ref{PX}) is not well-defined.
\end{itemize}
For different tolerance level $\alpha>0$ over the solvency risk,\\ $\left(\mathbb{E}\!\left[\! {X_\alpha^*}^2(T)1_{\{{X_\alpha^*}(T)<0\}}\right]\!,\mathbb{E}\!\left[ \frac{X_\alpha^*(T)^{1-\gamma}}{1-\gamma}1_{\{X_\alpha^*(T)>0\}}\right]\right)$ is a point on the plane. All points construct an efficient frontier on the plane.
\end{theorem}
\begin{proof}
See Appendix \ref{athm:ef}.
\end{proof}
Theorem \ref{thm:ef} shows the solvency risk and expected utility with different tolerance levels $\alpha$. Although the explicit form of the efficient frontier can not be obtained, we can plot the curve of efficient frontier numerically. Later we will see that the efficient frontier is very similar to the parabolic form in the mean-variance analysis, which characterizes a balance between the return (expected utility in the overfunded region) and the risk (solvency risk in the underfunded region).  Besides, we can compute the possibilities that the terminal wealth falls in the overfunded and underfunded regions, which are shown in the following theorem.
\begin{theorem}\label{thm:bd}
\begin{itemize}
		\item[(i)]{\bf Low tolerance towards solvency risk}
		\\If $k_1<2\alpha B $ and $Y_0>-B\exp\left(M+\frac{1}{2}V^2\right)$, the probability that the optimal terminal wealth falls in the {overfunded} region is
		\begin{equation*}
			\mathbb{P}\left[ X^*(T)>0\right]=\Phi\left(\Upsilon(k_1,\beta^*)\right).
		\end{equation*}
		and  the probability that the optimal terminal wealth falls in the {underfunded} region is $1-\Phi\left(\Upsilon(k_1,\beta^*)\right)$.
		
		\item[(ii)]{\bf High tolerance towards solvency risk}
		\\If $k_1\geqslant2\alpha B $ and $Y_0>-B\exp\left(M+\frac{1}{2}V^2\right)$, the probability that the optimal  terminal wealth falls in the {overfunded} region is
		\begin{equation*}
			\mathbb{P}\left[ X^*(T)>0\right]=\Phi\left(\Upsilon(k_2,\beta^*)\right)
		\end{equation*}
		and  the probability that the  optimal terminal wealth falls in the {underfunded} region is $1-\Phi\left(\Upsilon(k_2,\beta^*)\right)$.
		
		\item [(iii)]{\bf A specific lower bound}
		\\If $Y_0=-B\exp\left(M+\frac{1}{2}V^2\right)$, then we see that $ X^*(T)=-B ,a.s.$, and the probability that the  optimal terminal wealth falls in the underfunded region is $ 1 $.
		\item[(iv)]{\bf High lower bound}
		\\If $Y_0<-B\exp\left(M+\frac{1}{2}V^2\right)$, Problem (\ref{PX}) is not well-defined.
	\end{itemize}
%	For different tolerance level $\alpha>0$ over the solvency risk, $\left(\mathbb{P}\left[ X^*(T)>0\right],\mathbb{P}\left[ X^*(T)\leqslant0\right]\right)$ is a point on the plane. All points construct an efficient frontier on the plane.
\end{theorem}
\begin{proof}
Based on the closed-form of $X^*(T)$ in Theorem \ref{th:y}, we easily obtain the probability that the terminal wealth falls in the underfunded region. If $k_1<2\alpha B $ and $Y_0>-B\exp\left(M+\frac{1}{2}V^2\right)$,  we have
\begin{equation*}
		\begin{aligned}
			\mathbb{P}\left[X^*(T)>0\right]=\mathbb{P}\left[x^*(\beta^* \rho(T))>0\right]=\mathbb{P}\left[0<\beta^* \rho(T)\leqslant k_1\right]	=\Phi\left(\Upsilon(k_1,\beta^*)\right)
		\end{aligned}
	\end{equation*}
	and
	\begin{equation*}
		\begin{aligned}
			\mathbb{P}\left[X^*(T)<0\right]
			=1-\Phi\left(\Upsilon(k_1,\beta^*)\right).
		\end{aligned}
	\end{equation*}
	If $k_1\geqslant2\alpha B $ and $Y_0>-B\exp\left(M+\frac{1}{2}V^2\right)$,  we have
\begin{equation*}
		\mathbb{P}\left[X^*(T)>0\right]=\mathbb{P}\left[x^*(\beta^* \rho(T))>0\right]
=\mathbb{P}\left[0<\beta^* \rho(T)\leqslant k_2\right]=\Phi\left(\Upsilon(k_2,\beta^*)\right)
\end{equation*}
	and
	\begin{equation*}
		\begin{aligned}
			\mathbb{P}\left[X^*(T)<0\right]
			=1-\Phi\left(\Upsilon(k_2,\beta^*)\right).
		\end{aligned}
	\end{equation*}
In the last two cases, the results are straightforward.
\end{proof}
In this section, we transform the original problem into an equivalent one after introducing auxiliary processes. Then by the Lagrange method and martingale method, we obtain the results for the optimal terminal wealth. We see that the results are categorized into four cases: low tolerance level towards solvency risk, high tolerance level towards solvency risk, a specific lower bound, and high lower bound. Then using replication, the optimal wealth process and optimal strategies are obtained.

\vskip 15pt
\setcounter{equation}{0}
\section{\bf Numerical Analysis}
In this section, we show the impacts of different parameters on the manager's optimal wealth processes, optimal portfolios, efficient frontier, and probabilities in two regions. As the wealth processes and portfolios are stochastic, we conduct the Monte Carlo Method to illustrate the economic behaviors of the manager. Unless otherwise stated, the parameters we adopt are: $ T=10$, $ K=8$, $ r_0=0.04$, $ a=0.2$, $  b=0.02 $, $ \sigma_r=0.02 $, $\delta=0.005$, $k=0.06$,  $\sigma_{1}=0.2$, $ \sigma_2=0.4$, $ \lambda_r=0.15$, $\lambda_S=0.2$, $ P_0=0.15 $, $ \sigma_{P_1}=0.1$, $ \sigma_{P_2}=0.1$, $ X_0=0$, $ \mu=0.04 $, $ m=25 $, $d= 55 $, $ \gamma=0.4 $, $ \alpha=0.1 $, $ B=5 $, $M(x)=\frac{x-m}{d-m}$.
\subsection{\bf Optimal wealth processes}
\begin{figure}[htbp]
	\centering
		\begin{minipage}{0.5\textwidth}
		\centering
		\includegraphics[totalheight=5cm]{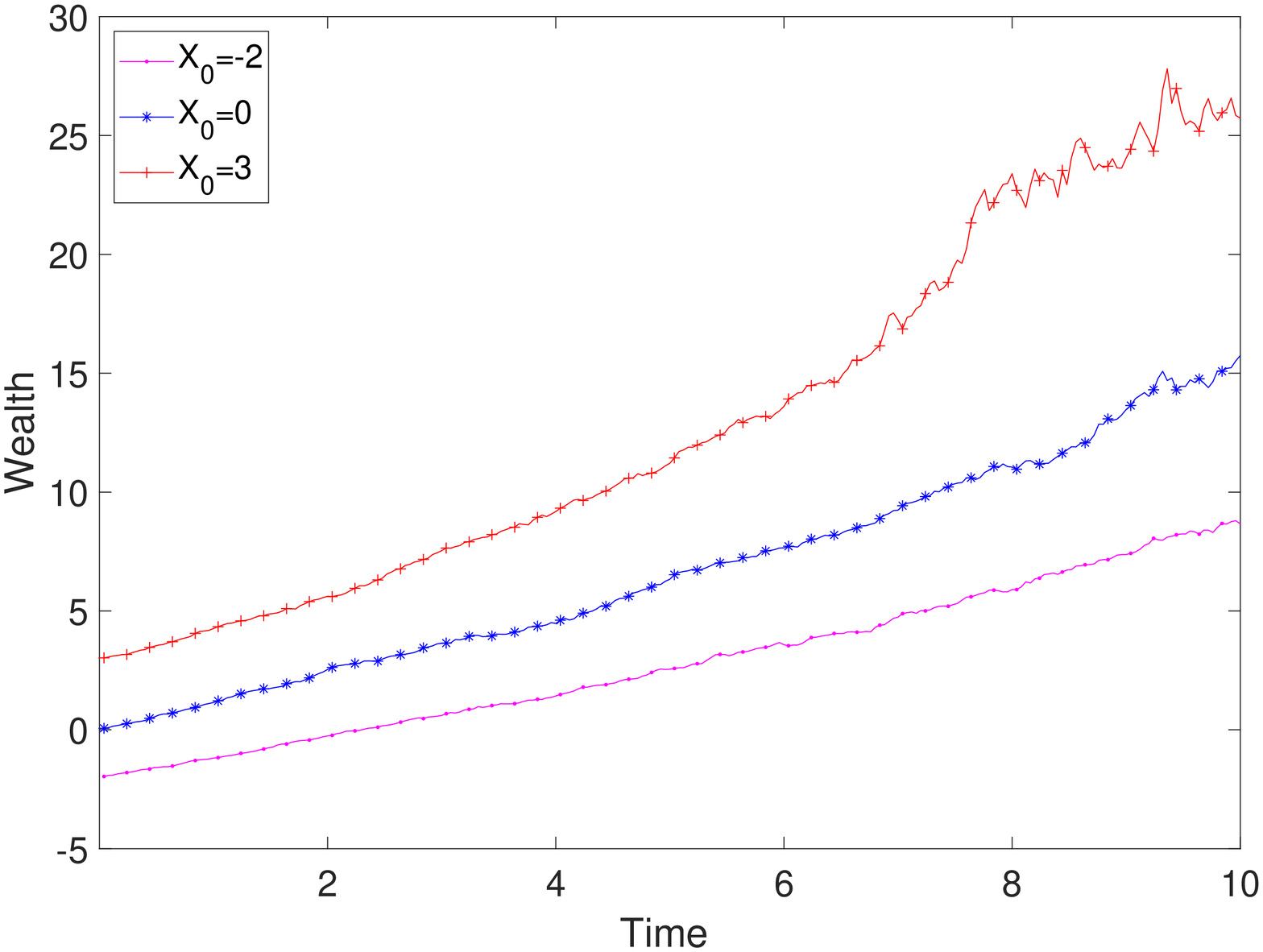}
		\caption{Effect of $X_0$.}
		\label{wealth:x0}
	\end{minipage}\hfill
	\begin{minipage}{0.5\textwidth}
	\centering
	\includegraphics[totalheight=5cm]{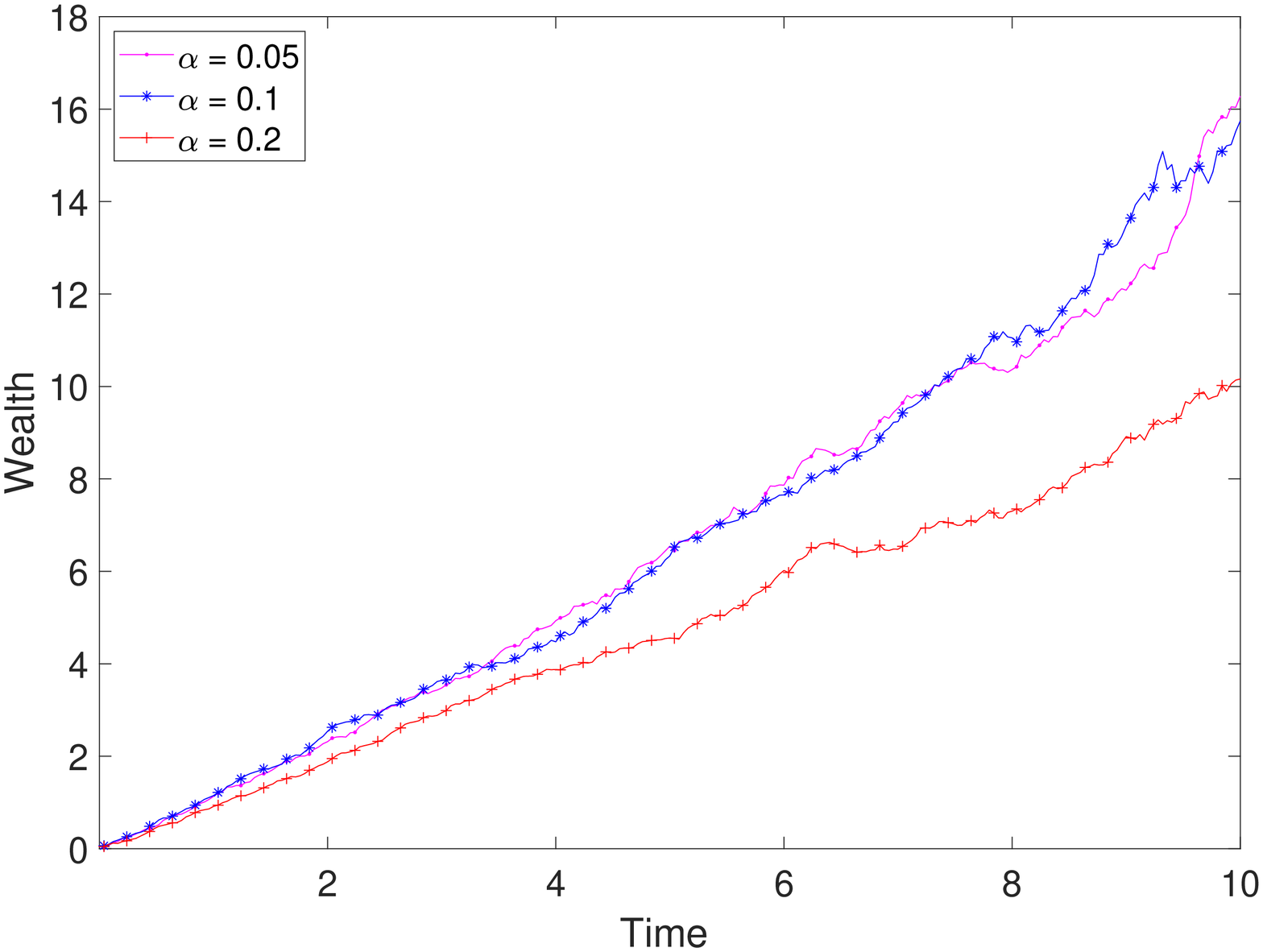}
	\caption{Effect of $\alpha$.}
	\label{wealth:alpha}
	\end{minipage}\hfill
		\begin{minipage}{0.5\textwidth}
		\centering
		\includegraphics[totalheight=5cm]{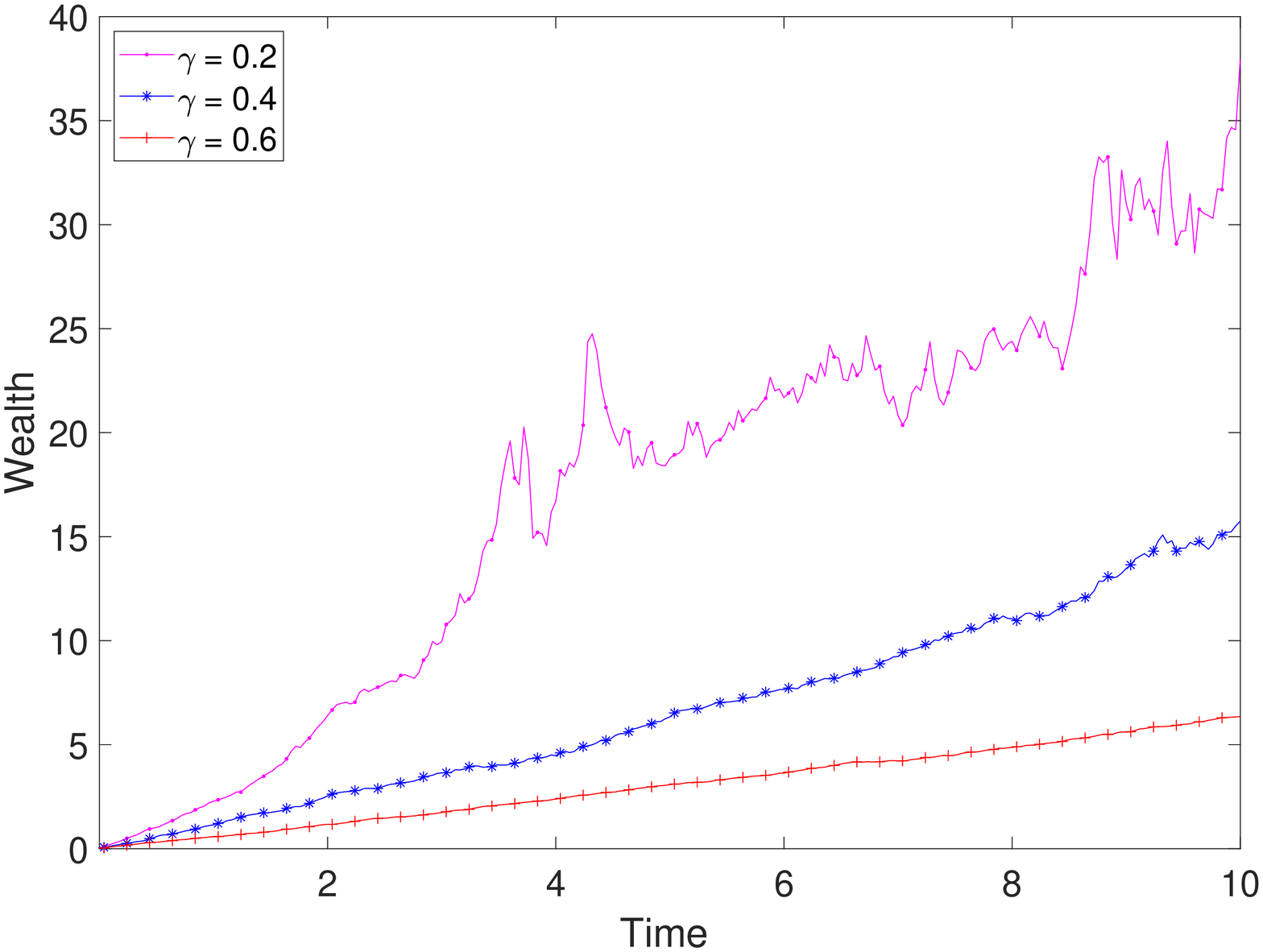}
			\caption{Effect of $\gamma$.}
		\label{wealth:gamma}
	\end{minipage}\hfill
	\begin{minipage}{0.5\textwidth}
		\centering
		\includegraphics[totalheight=5cm]{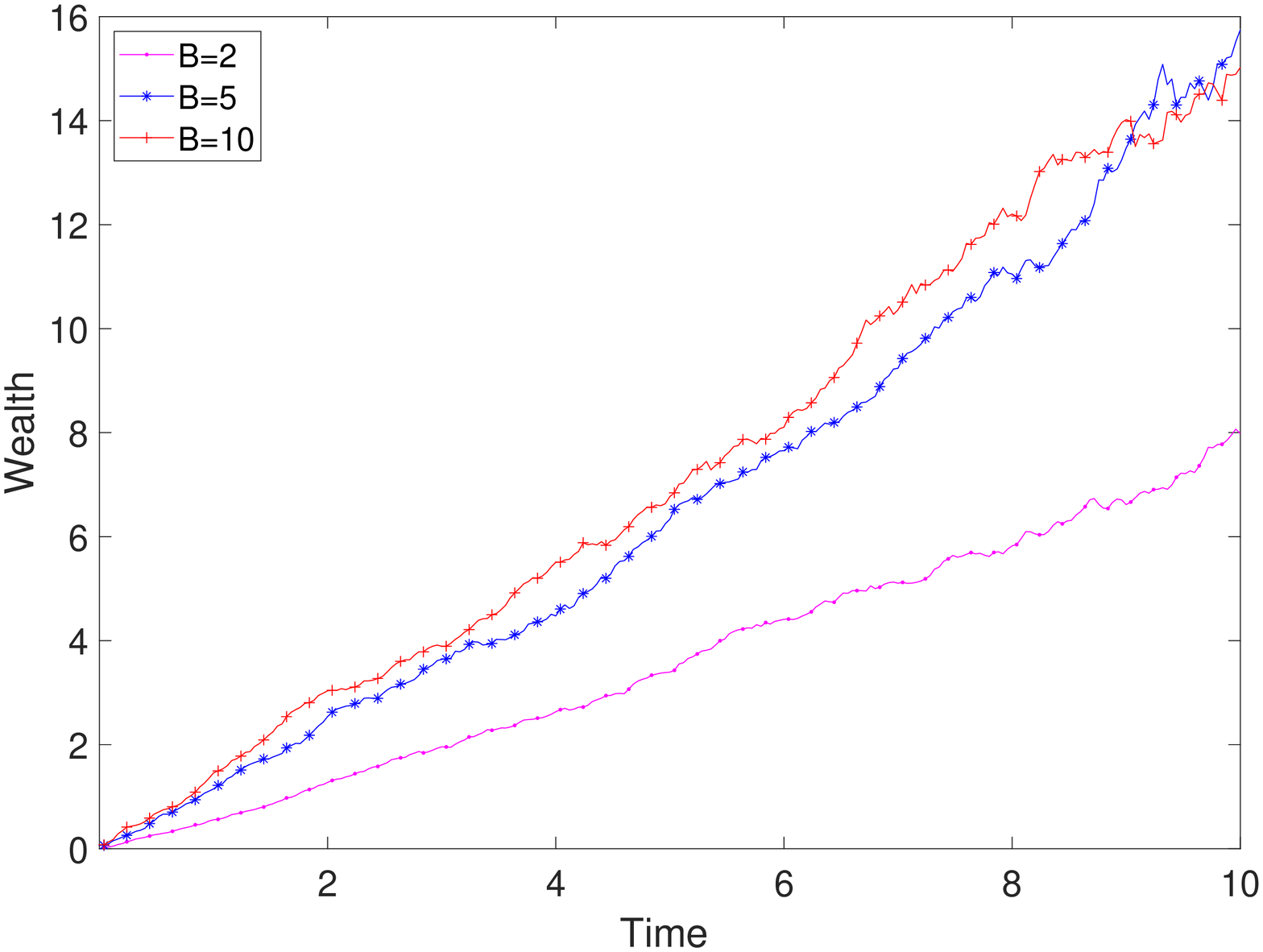}
		\caption{Effect of $B$.}
		\label{wealth:B}
	\end{minipage}\hfill
\end{figure}
Figs.~\ref{wealth:x0}-\ref{wealth:B} illustrate the mean evolution of optimal wealth with respect to  different parameters. Observing these figures, we see that the mean wealth increases with time, i.e., the manager has a larger expected fund surplus over time. In Fig.~\ref{wealth:x0}, the impact of the initial status of the fund on the optimal wealth is shown. We see that even in the case $X_0=-2$, the mean wealth increases to the overfunded region at about time 4. Meanwhile, the increase of wealth is faster for larger initial wealth $X_0$. Fig.~\ref{wealth:x0} shows that the optimal terminal wealth increases to about 25, 15, and 8 from $X_0=3,0,-2$ at the initial time. Besides, when $\alpha$ becomes larger, the manager has less tolerance level towards solvency risk and will be more risk-averse. As such, the optimal wealth increases slower for larger $\alpha$, which is depicted in Fig.~\ref{wealth:alpha}. The effect of $\gamma$ on the optimal wealth is shown in Fig.~\ref{wealth:gamma}. $\gamma$ represents the manager's risk aversion level towards financial risk. Obviously, for a manager with higher $\gamma$, he/she will be more conservative when allocating in the risky asset, which results in a smaller expected wealth in the future. Moreover, the positive relationship of $B$ and the optimal wealth is revealed in Fig.~\ref{wealth:B}. The bound $B$ also reflects the manager's tolerance over solvency risk. If $B$ is large, the manager becomes less sensitive to the negative value in the underfunded region. Then the manager acts more aggressively in the financial market and has a bigger expected wealth in the future.

\subsection{\bf  Optimal portfolio}
 \begin{figure}[htbp]
	\centering
		\begin{minipage}{0.9\textwidth}
		\centering
		\includegraphics[totalheight=8cm]{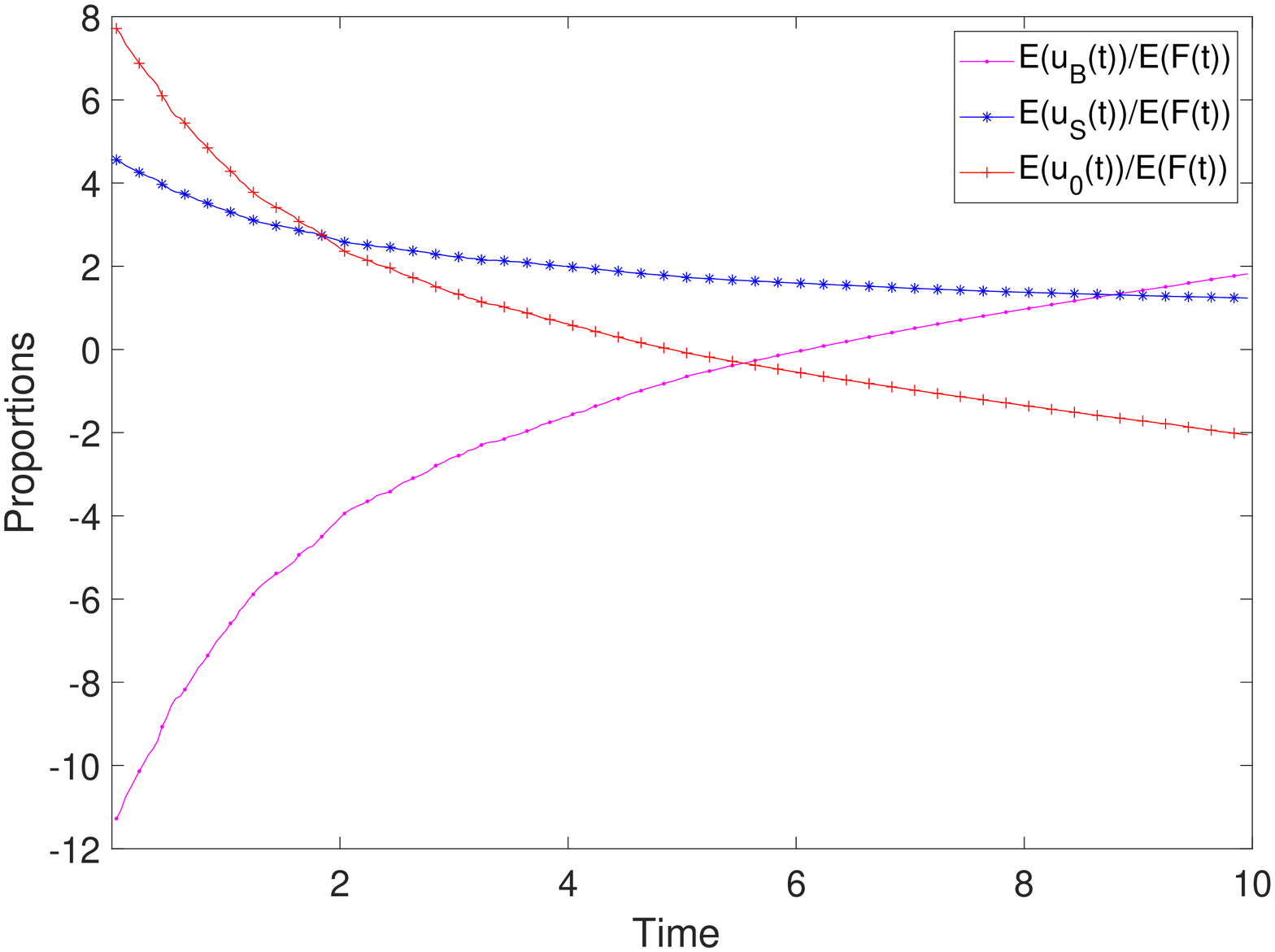}
		\caption{Benchmark parameters: $B=5,X_0=0,\gamma=0.4, \alpha=0.1, \mu=0.04, \delta=0.005, k=0.06, \sigma_r=0.02$.}
		\label{portfolio1}
	\end{minipage}\hfill
	\begin{minipage}{0.5\textwidth}
		\centering
		\includegraphics[totalheight=5cm]{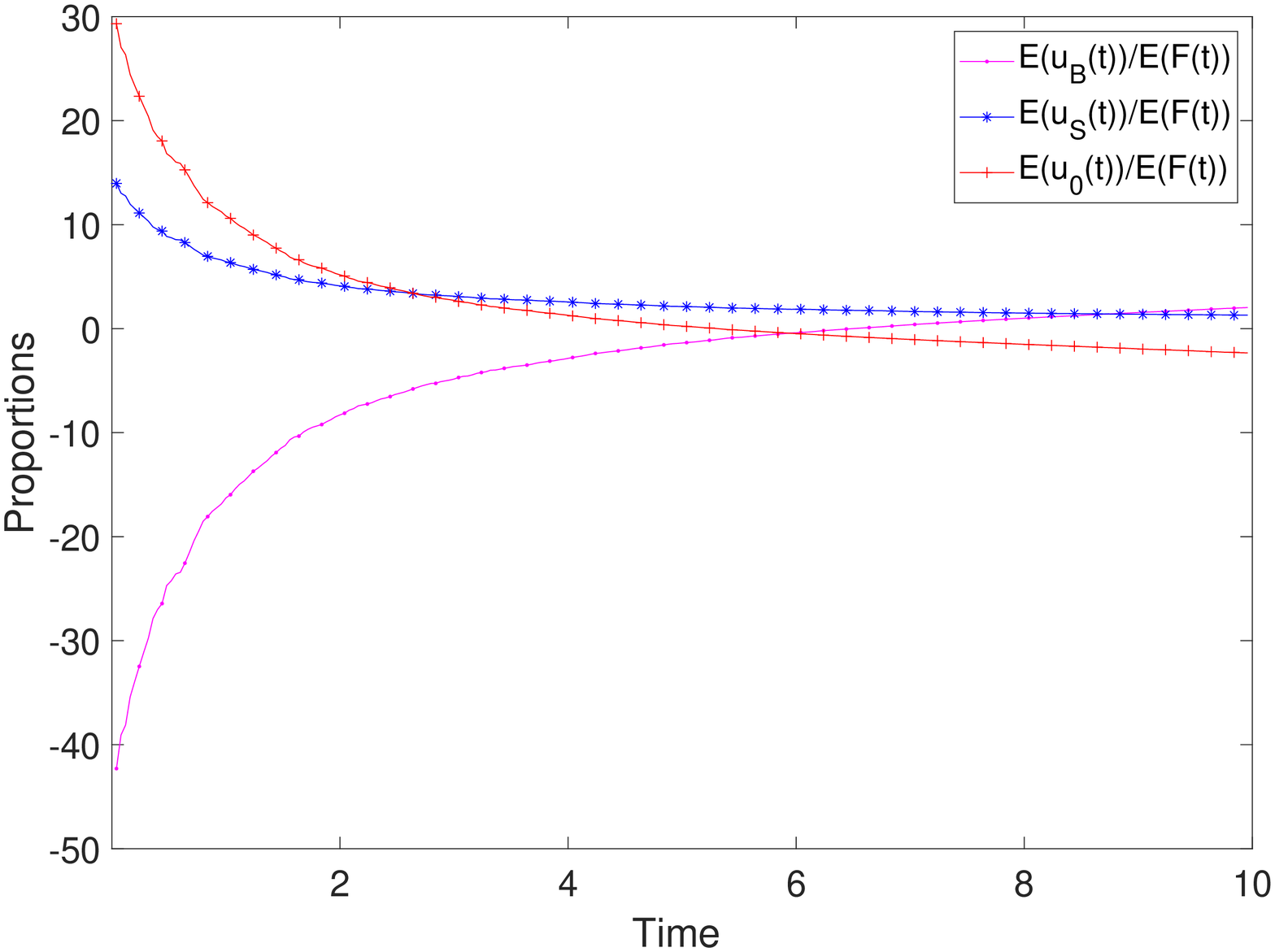}
		\caption{$X_0=0\rightarrow X_0=-2$.}
		\label{portfolio2}
	\end{minipage}\hfill
	\begin{minipage}{0.5\textwidth}
		\centering
		\includegraphics[totalheight=5cm]{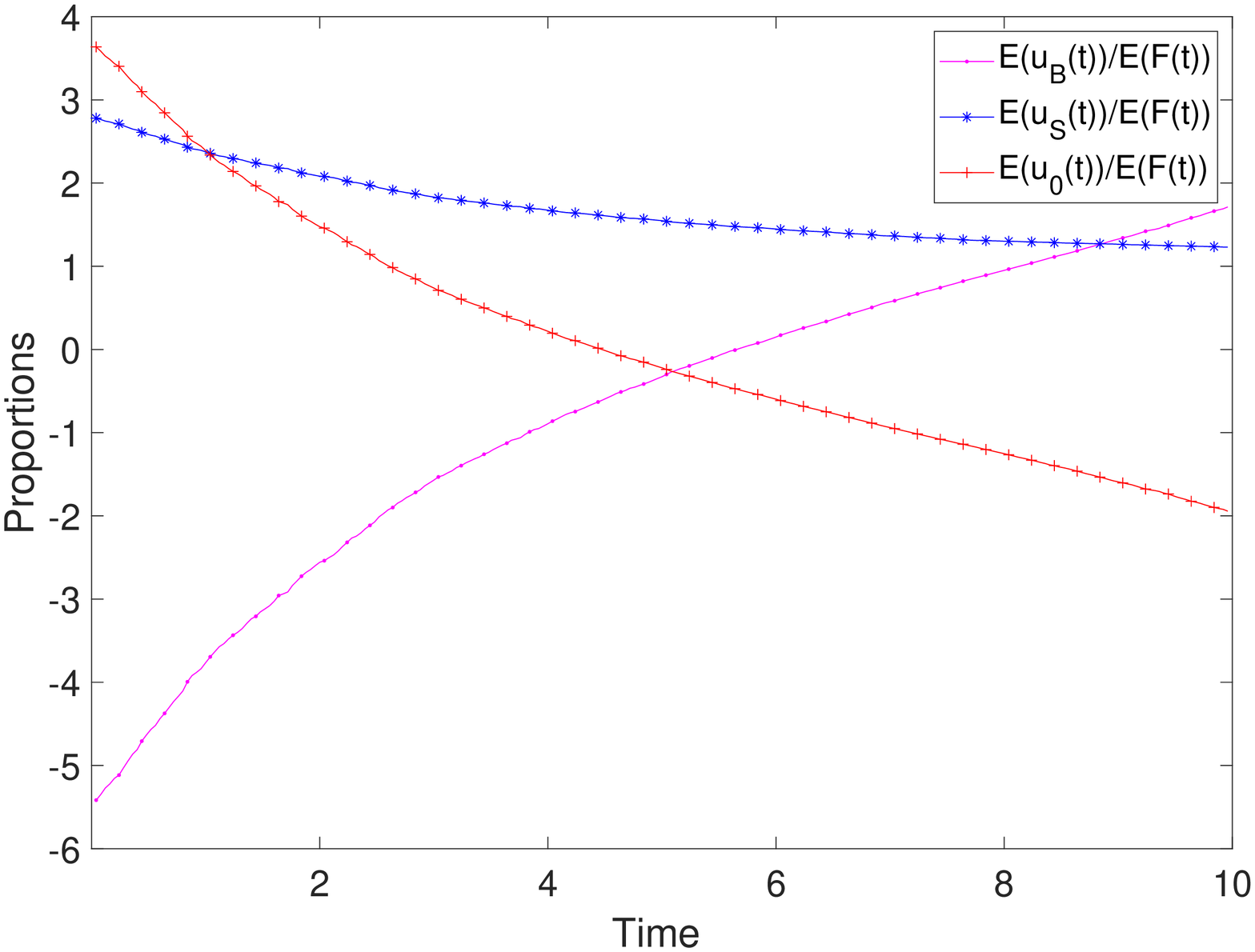}
		\caption{$X_0=0\rightarrow X_0=3$.}
		\label{portfolio3}
	\end{minipage}\hfill
\end{figure}
Next, we illustrate the effects of different parameters on the optimal investment strategies, which are shown in Figs.~\ref{portfolio1}-\ref{portfolio17}. We compare the other figures with the benchmark figure~\ref{portfolio1}. Fig.~\ref{portfolio1} reveals that the manager invests most in cash at the initial time, decreasing from about $8$ at the initial time to about $-2$ at terminal time. Meanwhile, the manager takes a long position in stock within the planning horizon, which decreases from about 4.2 to 1.8. The proportion in bond is negative at the initial time and becomes positive after time 6. Observing Figs.~\ref{portfolio1}-\ref{portfolio15}, we see that the manager always shorts bond and buys cash and stock at the initial time. The manage longs stock in the planning horizon while investment in cash (bond) may change from a positive (negative) value to a negative (positive) value at the terminal time. Fig.~\ref{portfolio2}  and Fig.~\ref{portfolio3} depict the optimal allocations when the DB pension fund begins in the underfunded and overfunded regions, respectively.  When the initial status of the fund is underfunded, the manager becomes risk-seeking. As such, compared with Fig.~\ref{portfolio1}, the manager in Fig.~\ref{portfolio2} acts aggressively and allocates more in stock with an initial proportion of 14. Besides, the manager takes a huge short position in bond, which increases from -43 at the initial time to 2 at terminal time. The absolute values of proportions in the three assets when $X_0=-2$ are all greater than the values in the case $X_0=0$. We also observe that the proportion in stock when $X_0=-2$ decreases from 13 to 2. When the initial status of the fund is overfunded, Fig.~\ref{portfolio3} shows that the proportion in cash (stock) decreases from about 3.8 (2.9)  at the initial time to -2 (1.1) at the terminal time. Similar to Figs.~\ref{portfolio1} and \ref{portfolio2}, the investment in bond also increases with time in this case.
 \begin{figure}[htbp]
	\centering
	\begin{minipage}{0.5\textwidth}
		\centering
		\includegraphics[totalheight=5cm]{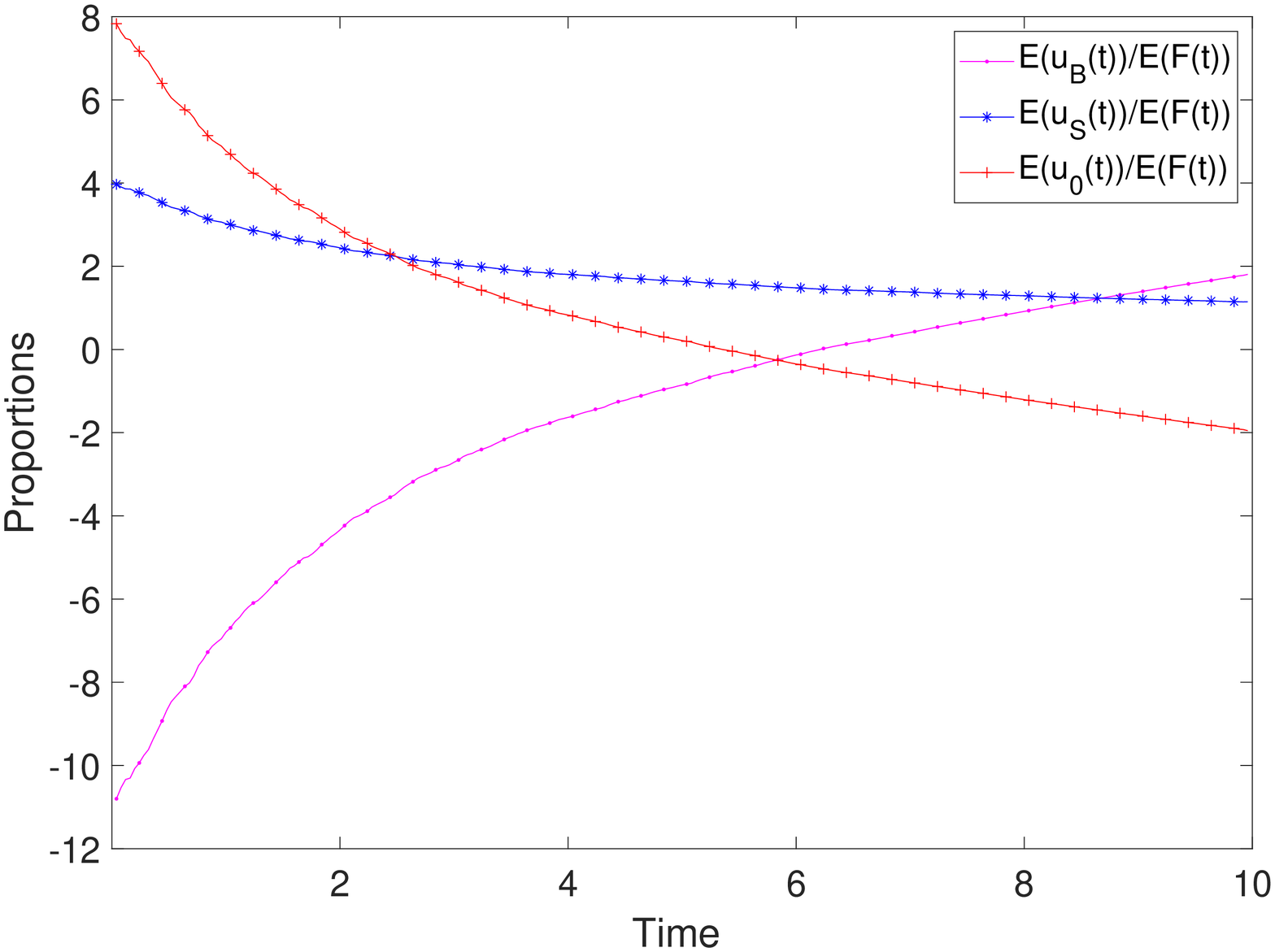}
		\caption{$\alpha=0.1\rightarrow\alpha=0.2$.}
		\label{portfolio13}
	\end{minipage}\hfill
	\begin{minipage}{0.5\textwidth}
		\centering
		\includegraphics[totalheight=5cm]{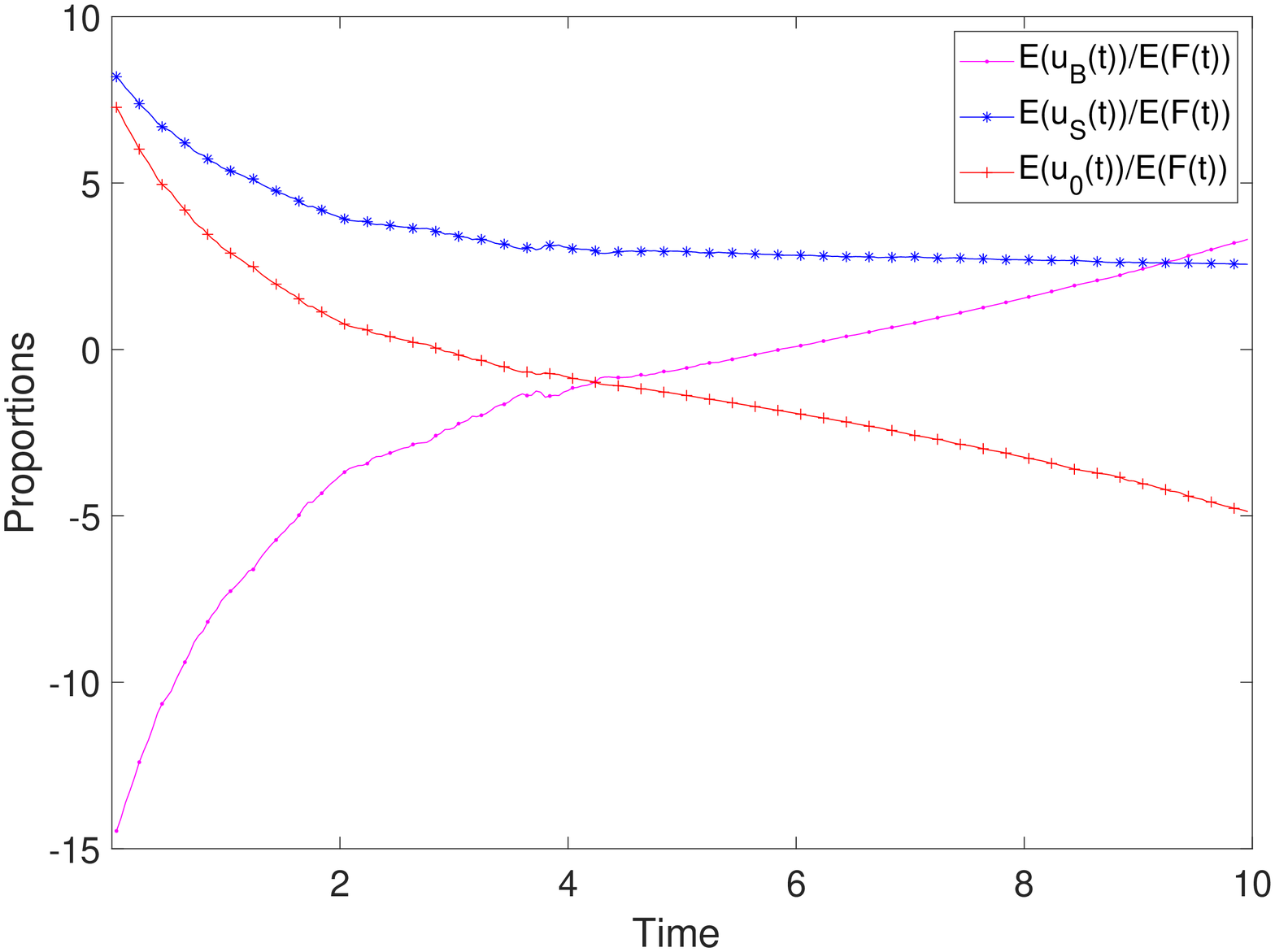}
		\caption{$\gamma=0.4\rightarrow\gamma=0.2$.}
		\label{portfolio6}
	\end{minipage}\hfill
	\begin{minipage}{0.5\textwidth}
	\centering
	\includegraphics[totalheight=5cm]{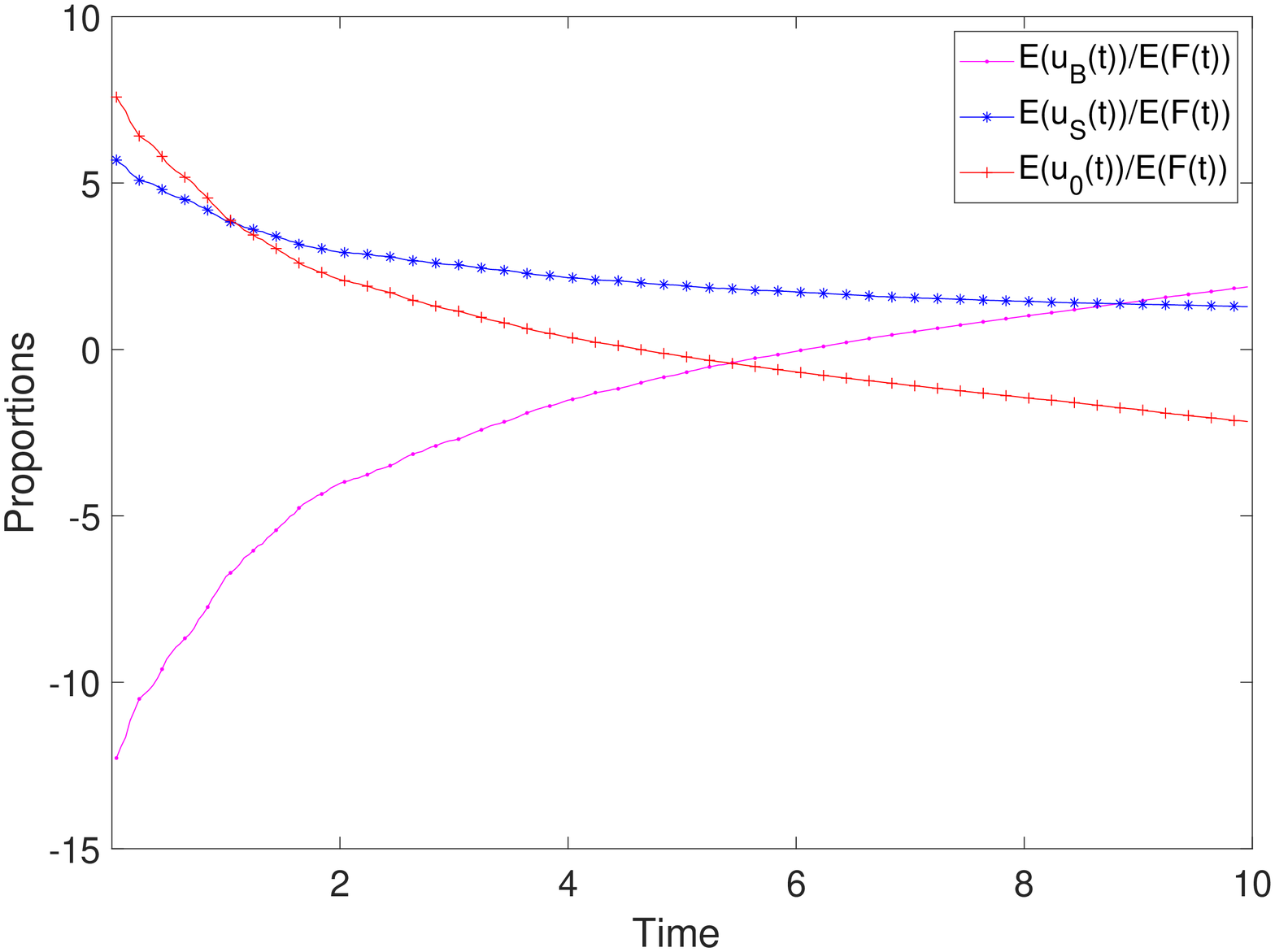}
	\caption{$B=5\rightarrow B=10$.}
	\label{portfolio5}
\end{minipage}\hfill
	\begin{minipage}{0.5\textwidth}
	\centering
	\includegraphics[totalheight=5cm]{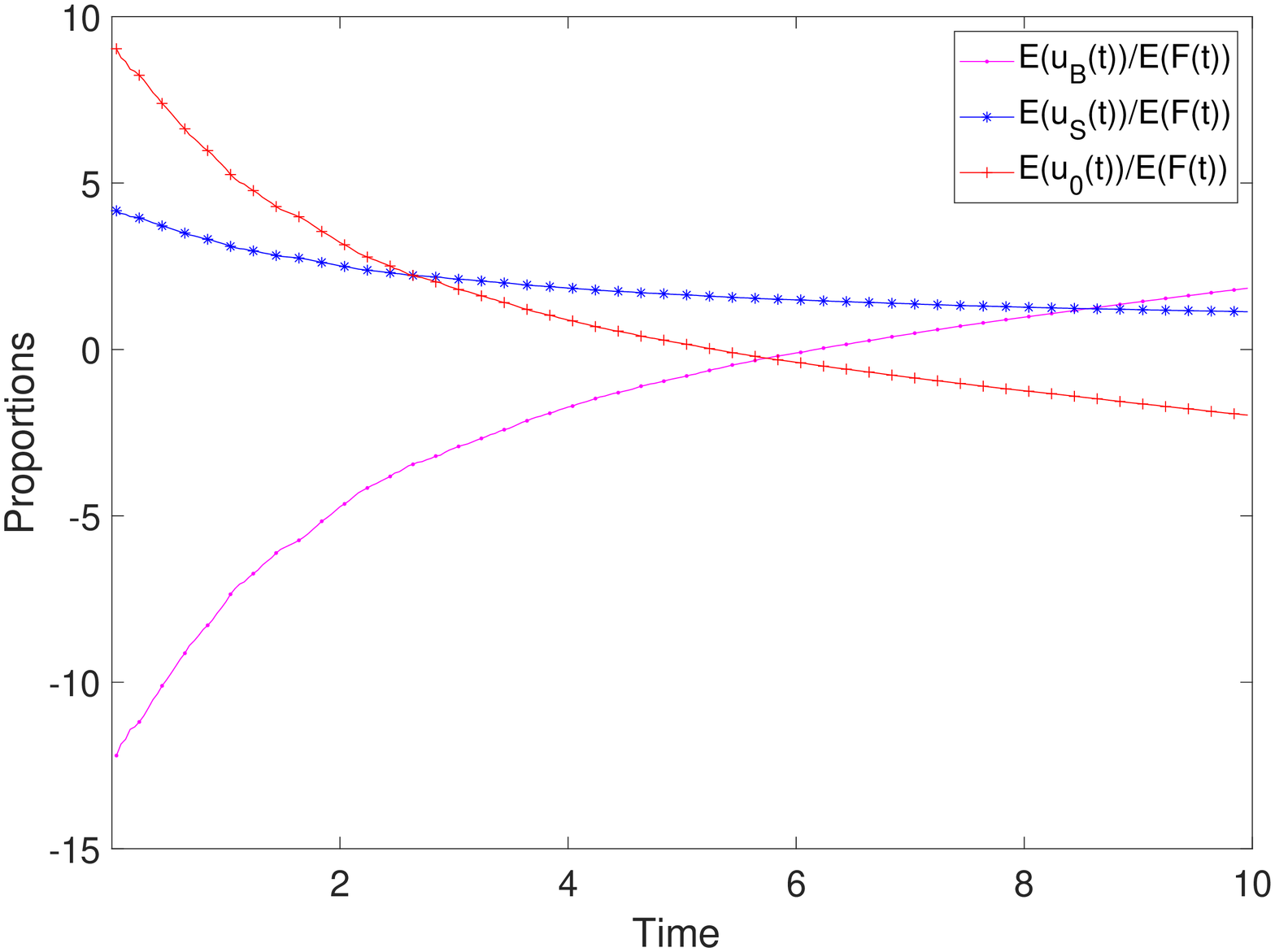}
	\caption{$\mu=0.04\rightarrow\mu=0.06$.}
	\label{portfolio11}
\end{minipage}\hfill
\end{figure}

Fig.~\ref{portfolio13} shows the case when $\alpha$ increases from 0.1 to 0.2.  The difference between Fig.~\ref{portfolio13} and Fig.~\ref{portfolio1} is small. In the financial market, the manager with larger $\alpha$ becomes more {averse to} solvency risk and will decrease the allocation in the risky asset, which is indicated by the initial proportions of stock in Fig.~\ref{portfolio13} and Fig.~\ref{portfolio1}. Besides, if the manager has a smaller risk aversion coefficient $\gamma$, he/she will allocate more in the risky asset. Fig.~\ref{portfolio6} shows that when $\gamma=0.2$, the proportion in stock decreases from about 8 to 2.5, which evolves bigger than that in Fig.~\ref{portfolio1}. The investment in bond also has a  drastic change, which begins with -15 at the initial time and increases to about 3 at terminal time. Furthermore, the economic interpretation of $B$ is similar with $\alpha$: a bigger $B$ leads to a manager with high tolerance over solvency risk. As such, compared with Fig.~\ref{portfolio1}, the manager allocates more in the risky assets: stock and cash. To compensate for the increase of proportion in stock and cash, the proportion in bond is smaller in the case $B=10$. In Fig.~\ref{portfolio11}, $\mu$ increases from 0.04 to 0.06. As illustrated in Fig.~\ref{portfolio11}, the proportion in stock decreases from about 4.8 to 2, which is slightly bigger than the curve of proportion in stock presented in Fig.~\ref{portfolio1}. Meanwhile, the manager invests more in cash with a proportion of 10 at the initial time. As $\mu$ characterizes the increasing rate of the promised benefits, the manager with larger $\mu$ is faced with greater liability, as such, the available fund surplus decreases. Therefore, the comparison of Fig.~\ref{portfolio11} and Fig.~\ref{portfolio1} is very similar to that of Fig.~\ref{portfolio2} and Fig.~\ref{portfolio1}.
 \begin{figure}[htbp]
	\centering
	\begin{minipage}{0.5\textwidth}
	\centering
	\includegraphics[totalheight=5cm]{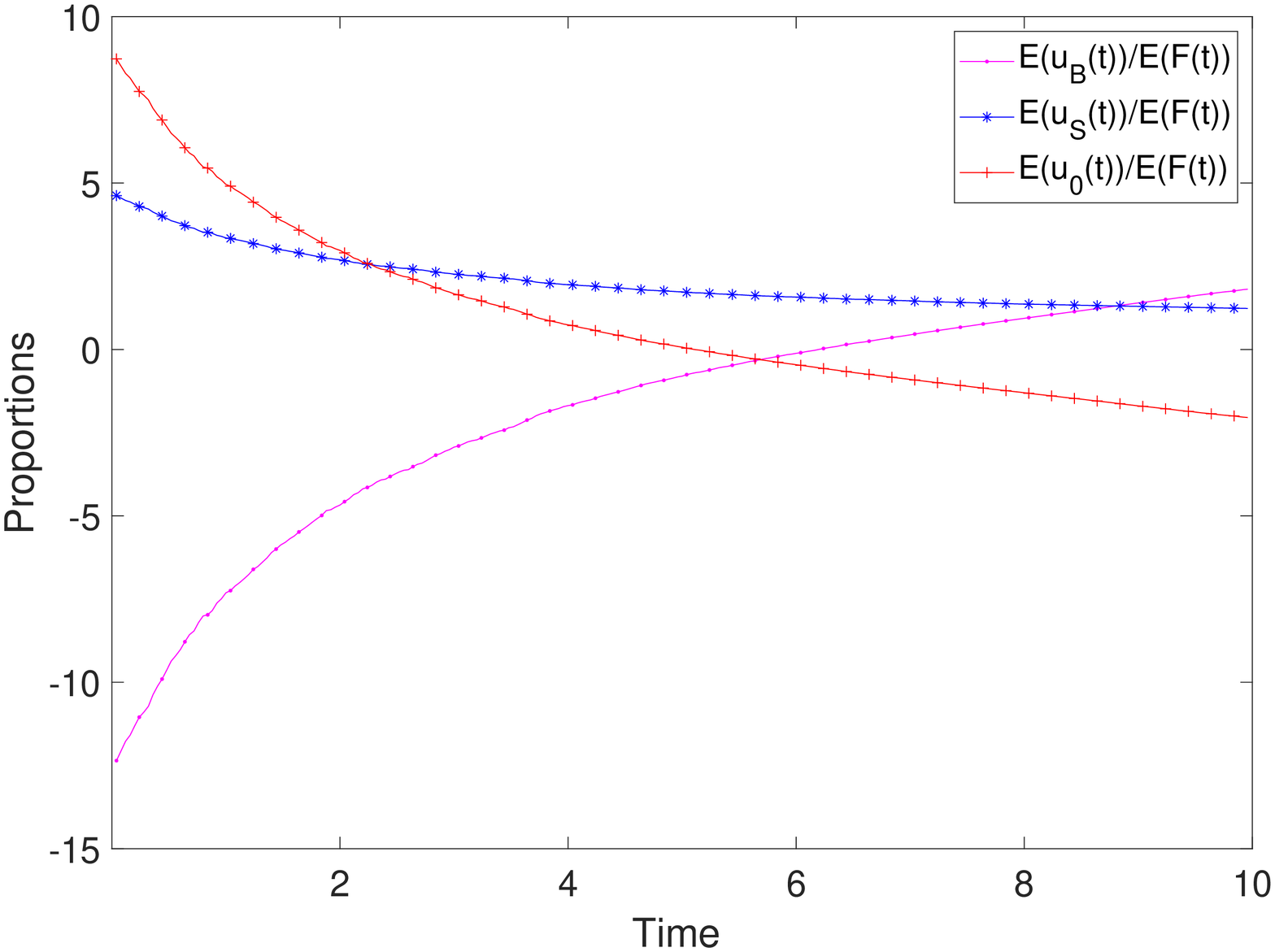}
	\caption{$\delta=0.005\rightarrow\delta=0.001$.}
	\label{portfolio8}
\end{minipage}\hfill
\begin{minipage}{0.5\textwidth}
	\centering
	\includegraphics[totalheight=5cm]{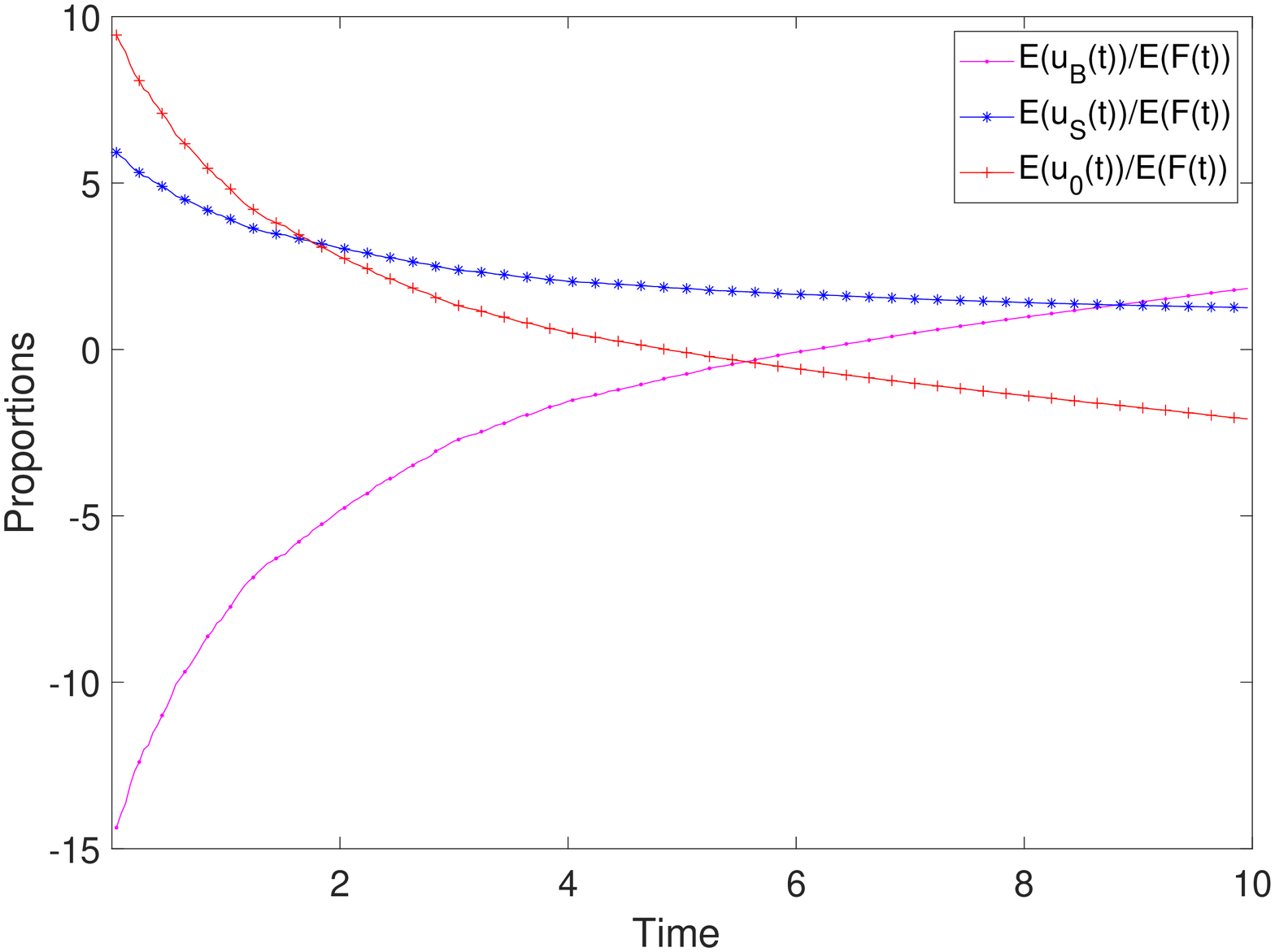}
	\caption{$k=0.06\rightarrow k=0.09$.}
	\label{portfolio15}
\end{minipage}\hfill
\begin{minipage}{0.5\textwidth}
	\centering
	\includegraphics[totalheight=5cm]{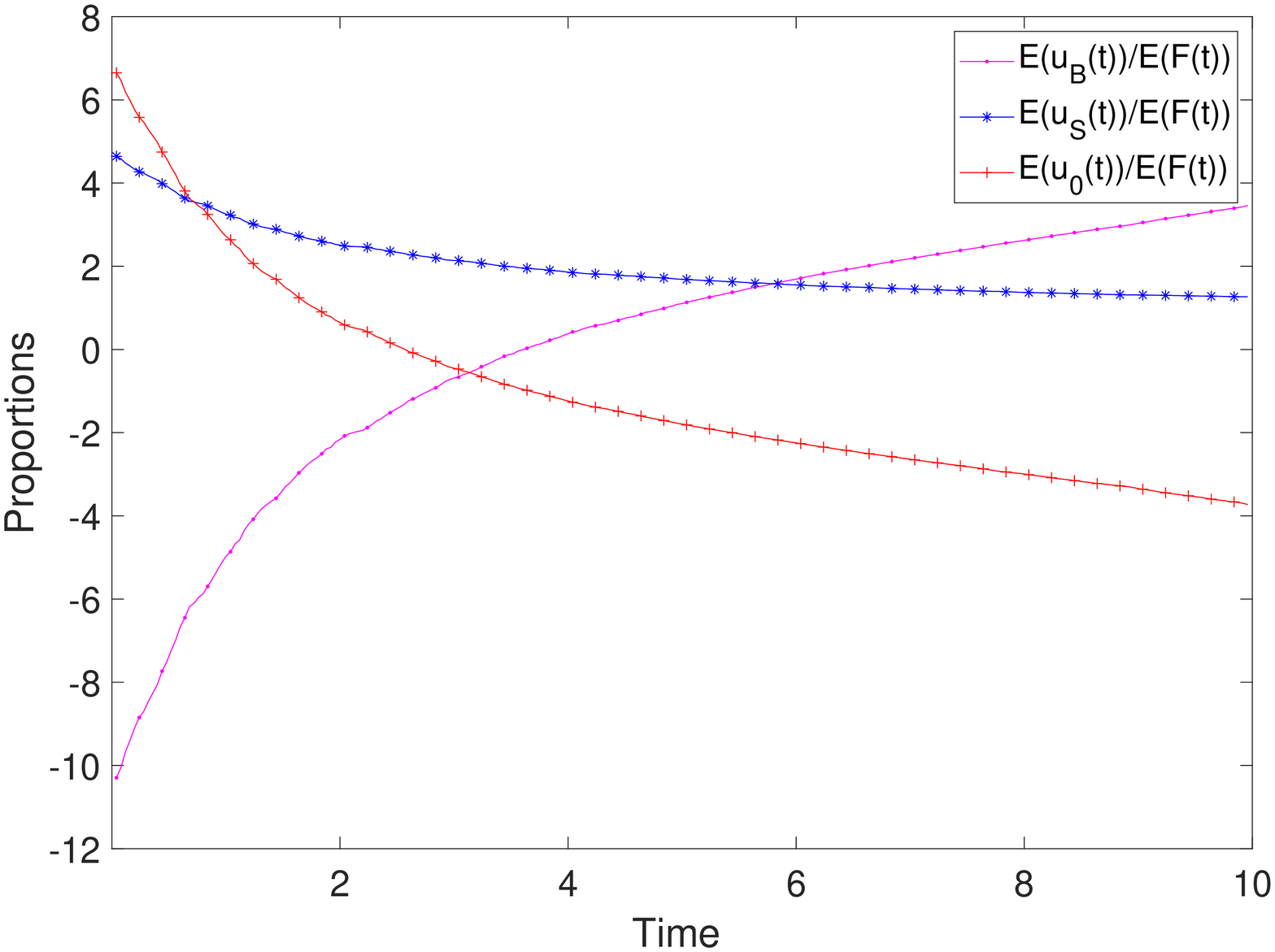}
	\caption{$\sigma_r=0.02\rightarrow\sigma_r=0.01$.}
	\label{portfolio17}
\end{minipage}\hfill
\end{figure}

Figs.~\ref{portfolio8}, \ref{portfolio15} and \ref{portfolio17} depict  the influences of $\delta$, $k$ and  $\sigma_r$ on the optimal investment strategies. $\delta$ characterizes the discount rate of the future benefits. The manager is faced with higher liabilities when $\delta$  decreases to 0.001. Then the manager becomes aggressive and allocates more in stock, which decreases from about 5  at the initial time to about 2 at the terminal time in Fig.~\ref{portfolio8}. Comparing with Fig.~\ref{portfolio1}, we also observe that proportion in cash is bigger while the proportion in bond is smaller in the case when $\delta=0.001$. We are also interested in the effect of constant $k$ on the optimal investment strategies. The contribution rate decreases in the overfunded region while increases in the underfunded region with larger $k$. As shown in Fig.~\ref{portfolio15}, when $k$ increases, the initial proportions in bond, cash, and stock are -15, 10, and 6, respectively. The impact of the interest risk $\sigma_r$ can be illustrated by comparing Figs.~\ref{portfolio1} and \ref{portfolio17}.  When $\sigma_r$ decreases to 0.01, the financial risk decreases. It is natural to expect larger investment in the cash, which is however in the opposite direction in Fig.~\ref{portfolio17}. As the actualization discount rate is set as $r(t)+\delta$, $\sigma_r$ also has a great effect on the discount rate, and therefore the influence of $\sigma_r$ on the strategies shall not be clearly revealed intuitively.

\subsection{\bf Efficient frontier}
Figs.~\ref{boundry3}-\ref{boundry6} illustrate the impacts of different parameters of efficient frontier. The efficient frontier is derived numerically based on Theorem \ref{thm:ef}. Although the closed form of the efficient frontier can not be obtained, Figs.~\ref{boundry3}-\ref{boundry6} indicate that the efficient frontier is of the parabolic form as in the mean-variance analysis. There is a balance between solvency risk in the underfunded region and expected utility in the overfunded region: if the manager has a lower (higher) tolerance level towards solvency risk, he/she can achieve higher (lower) expected utility in the overfunded region. Thus, as in the mean-variance analysis, we  regard the solvency risk (expected utility) as risk (return) and the manager needs to make a decision up to some  tolerance level. However, different from the efficient frontier in the mean-variance analysis, the efficient frontier in our problem has left and right endpoints. On the one hand, the solvency risk may not be zero and the efficient frontier may not intersect with the y-axis. On the other hand, as the terminal wealth has a negative bound, the solvency risk shall not be arbitrarily large and the efficient frontier has a right endpoint. The intersection point of the efficient frontier with the y-axis is related to the result under optimization rule (\ref{op2}). The increasing properties of curves in Figs.~\ref{boundry3}-\ref{boundry6} show that in our financial model, the manager can {sacrifice some solvency} to attain higher expected utility in the overfunded region.
\begin{figure}[htbp]
	\centering
	\begin{minipage}{0.48\textwidth}
		\centering
		\includegraphics[totalheight=5cm]{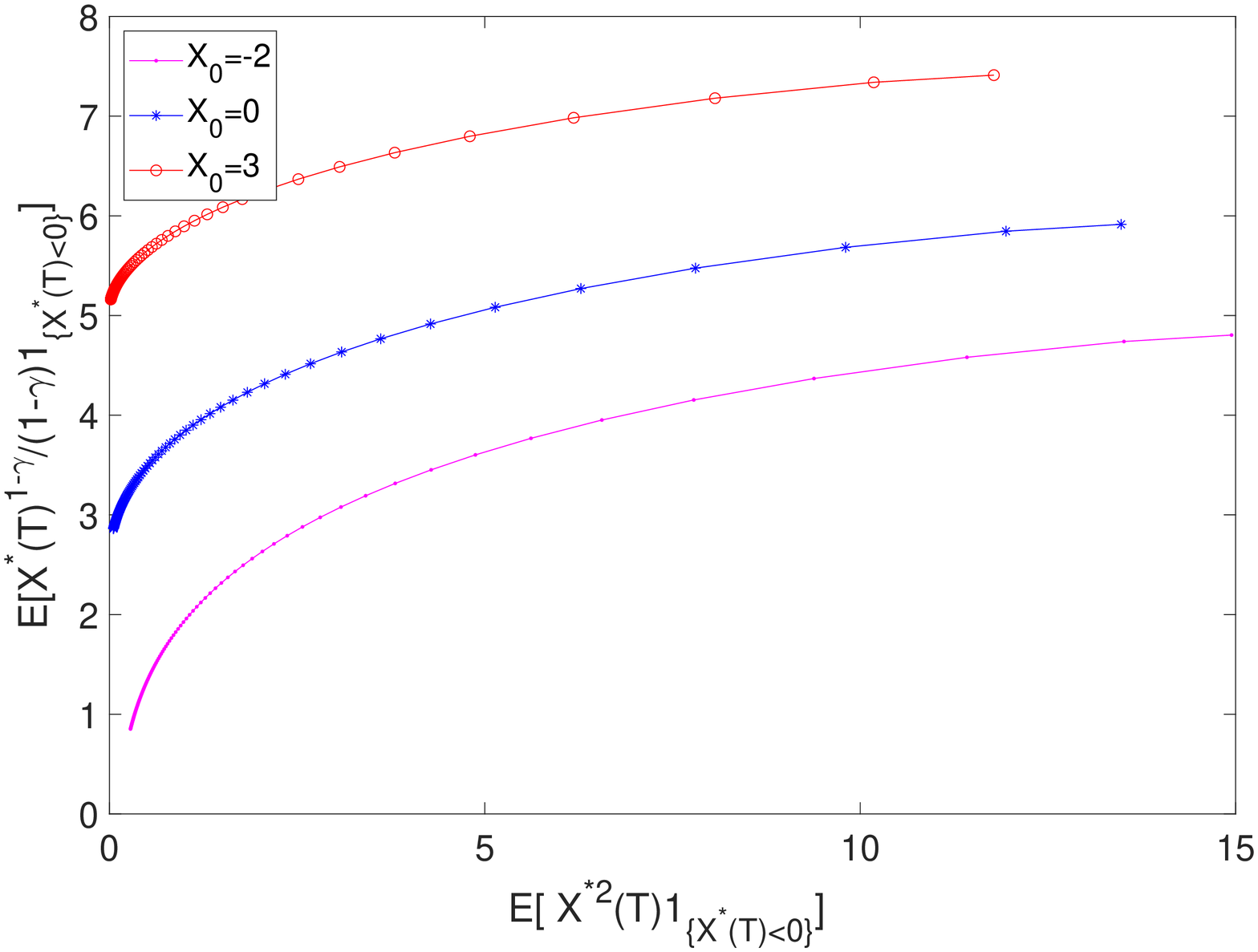}
		\caption{Effect of $X_0$.}
		\label{boundry3}
	\end{minipage}\hfill
	\begin{minipage}{0.48\textwidth}
		\centering
		\includegraphics[totalheight=5cm]{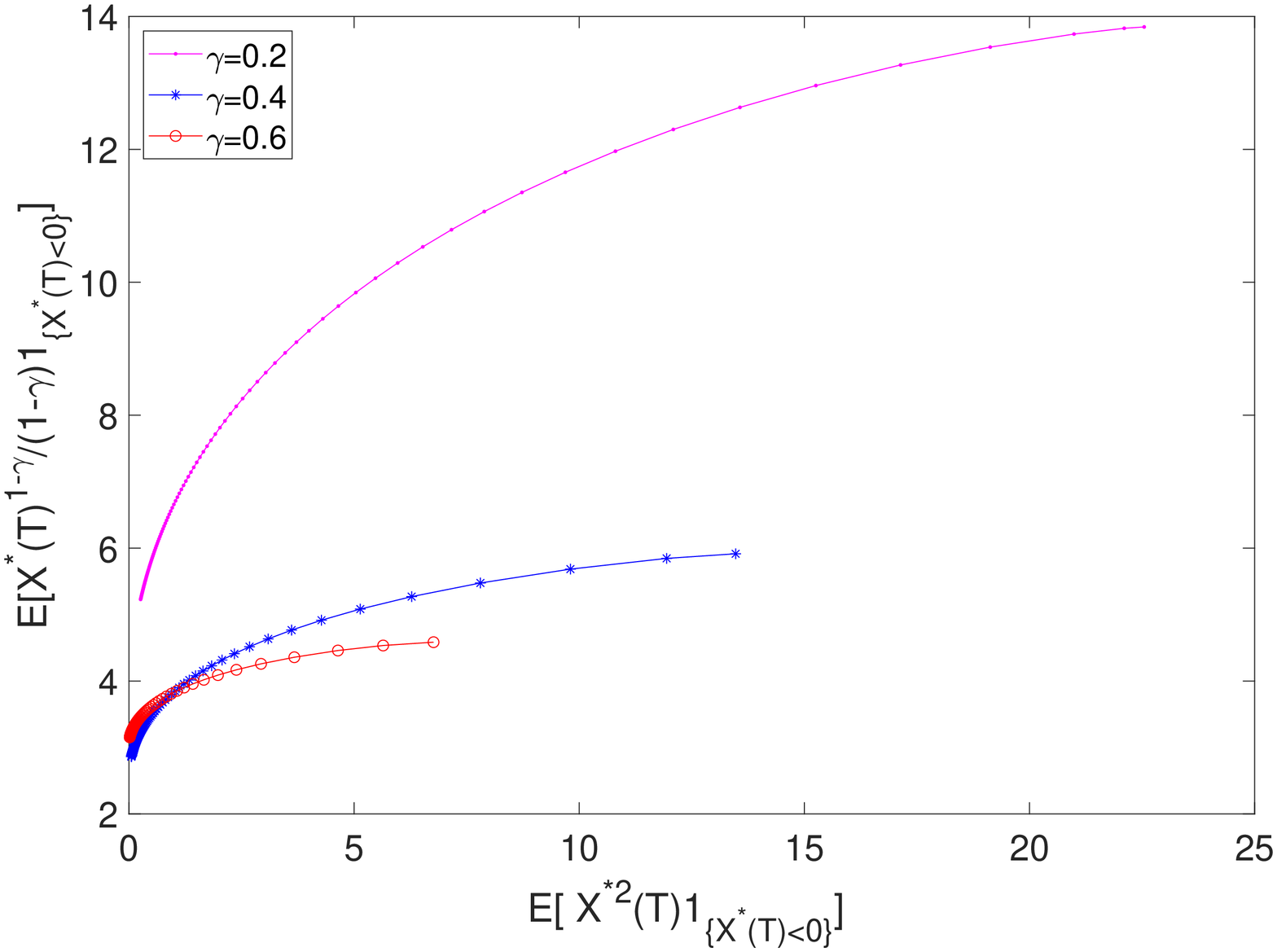}
		\caption{Effect of $\gamma$.}
		\label{boundry1}
	\end{minipage}\hfill
	\begin{minipage}{0.48\textwidth}
		\centering
		\includegraphics[totalheight=5cm]{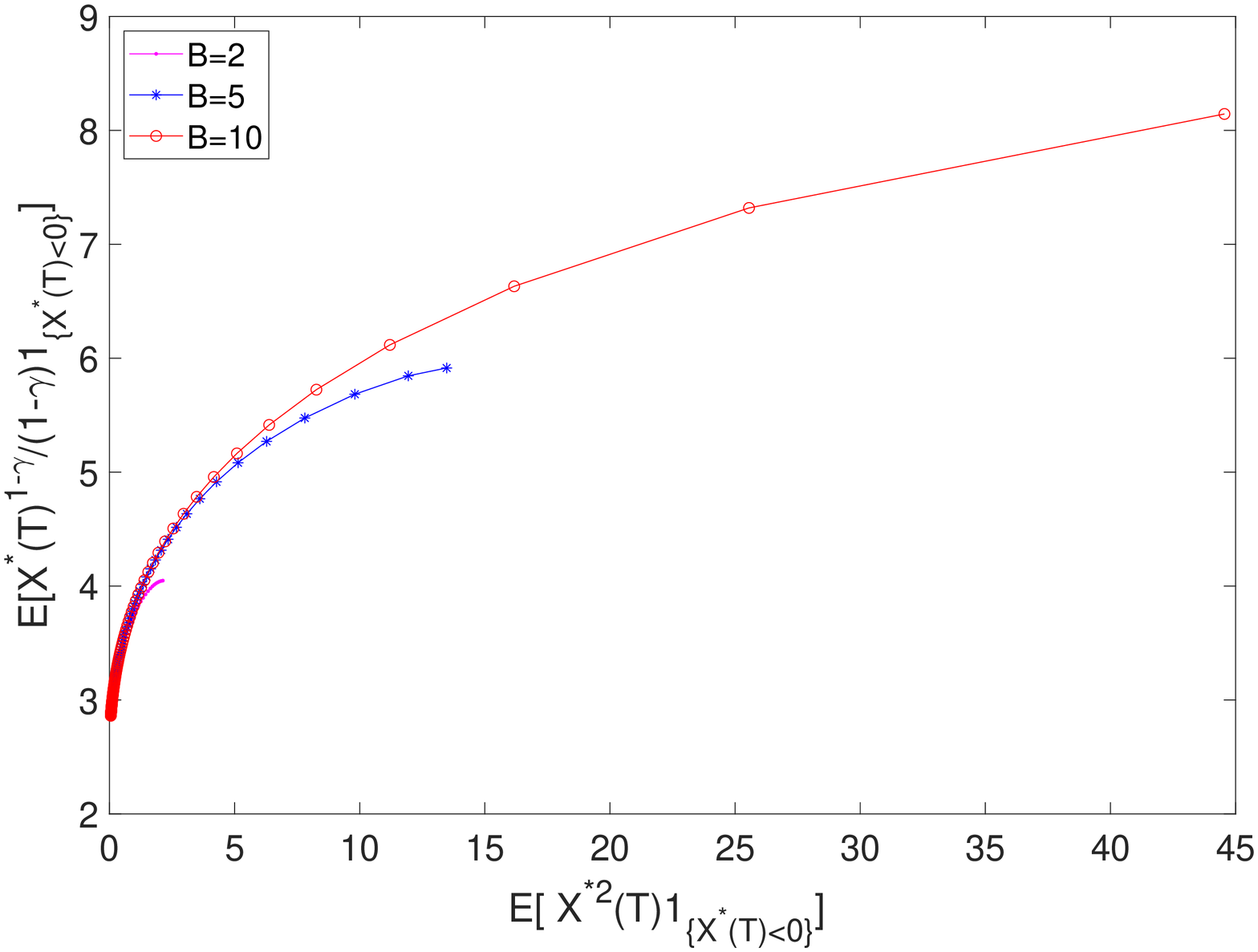}
		\caption{Effect of $B$.}
		\label{boundry2}
	\end{minipage}\hfill
	\begin{minipage}{0.48\textwidth}
		\centering
		\includegraphics[totalheight=5cm]{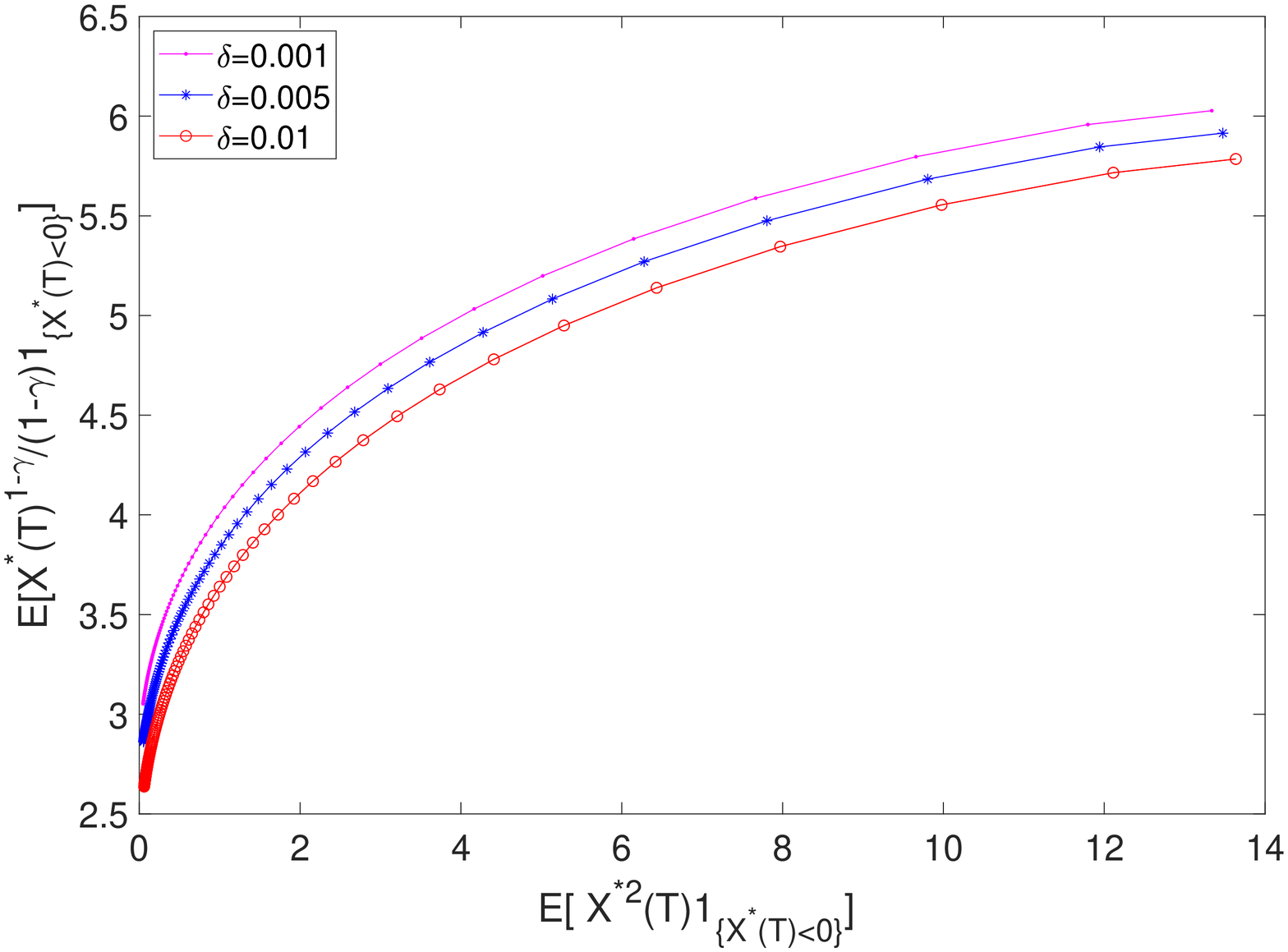}
		\caption{Effect of $\delta$.}
		\label{boundry5}
	\end{minipage}\hfill
\end{figure}
Fig.~\ref{boundry3} illustrates the relationship between the efficient frontier and the initial fund surplus $X_0$. When the manager has larger initial wealth, he/she can attain a larger expected utility for given solvency risk. As such, the efficient frontier in Fig.~\ref{boundry3} moves upward when $X_0$ increases. Besides, when $X_0=-2$, the pension fund is insolvent at the initial time and the solvency risk can not be completely eliminated. We see that in Fig.~\ref{boundry3}, the efficient frontier when $X_0=-2$ does not intersect with the y-axis. Besides, when $X_0$ decreases, the solvency risk faced by the manager increases. Fig.~\ref{boundry3} also indicates that the solvency risk of right endpoint in the efficient frontier increases when $X_0$ gets  smaller. The effect of the risk aversion parameter on the efficient frontier is indicated by Fig.~\ref{boundry1}. When $\gamma$ increases, the manager is more risk aversion towards financial risk and becomes more conservative when making decisions. Then for a larger $\gamma$, the expected utility obtained in the overfunded region decreases, which is well plotted in  Fig.~\ref{boundry1}. However, a careful observation over Fig.~\ref{boundry1} shows that the efficient frontiers when $\gamma=0.5$ and $\gamma=0.4$ interact when the solvency risk is small. These abnormal phenomena may be explained as follows: when the acceptable solvency risk is relatively small, the manager is highly risk aversion and allocates wealth in the risk-free asset. Then the multiplier $\frac{1}{1-\gamma}$ which increases with $\gamma$, has a big influence on the expected utility in the overfunded region. Then the efficient frontiers may intersect with small solvency risk.
\vskip 5pt
Fig.~\ref{boundry2} clearly reveals the connection of the bound $B$ and the efficient frontier. If the terminal wealth can fall to a smaller negative bound, the manager is risk-seeking and can attain a higher expected utility in the overfunded region. Therefore, the efficient frontier moves upward with $B$. Meanwhile, for larger $B$, the solvency risk w.r.t. the right endpoint of the efficient frontier also increases. In Fig.~\ref{boundry5}, the effect of $\delta$ on the efficient frontier is illustrated. $\delta$ reflects the discount rate of the liabilities. With a bigger $\delta$, the value of future benefits promised to participants increases, and the manager is faced with greater liabilities.  Then the manager is endowed with less fund surplus and has an efficient frontier more closing to the x-axis.

\begin{figure}[htbp]
	\centering
	\begin{minipage}{0.5\textwidth}
		\centering
		\includegraphics[totalheight=5cm]{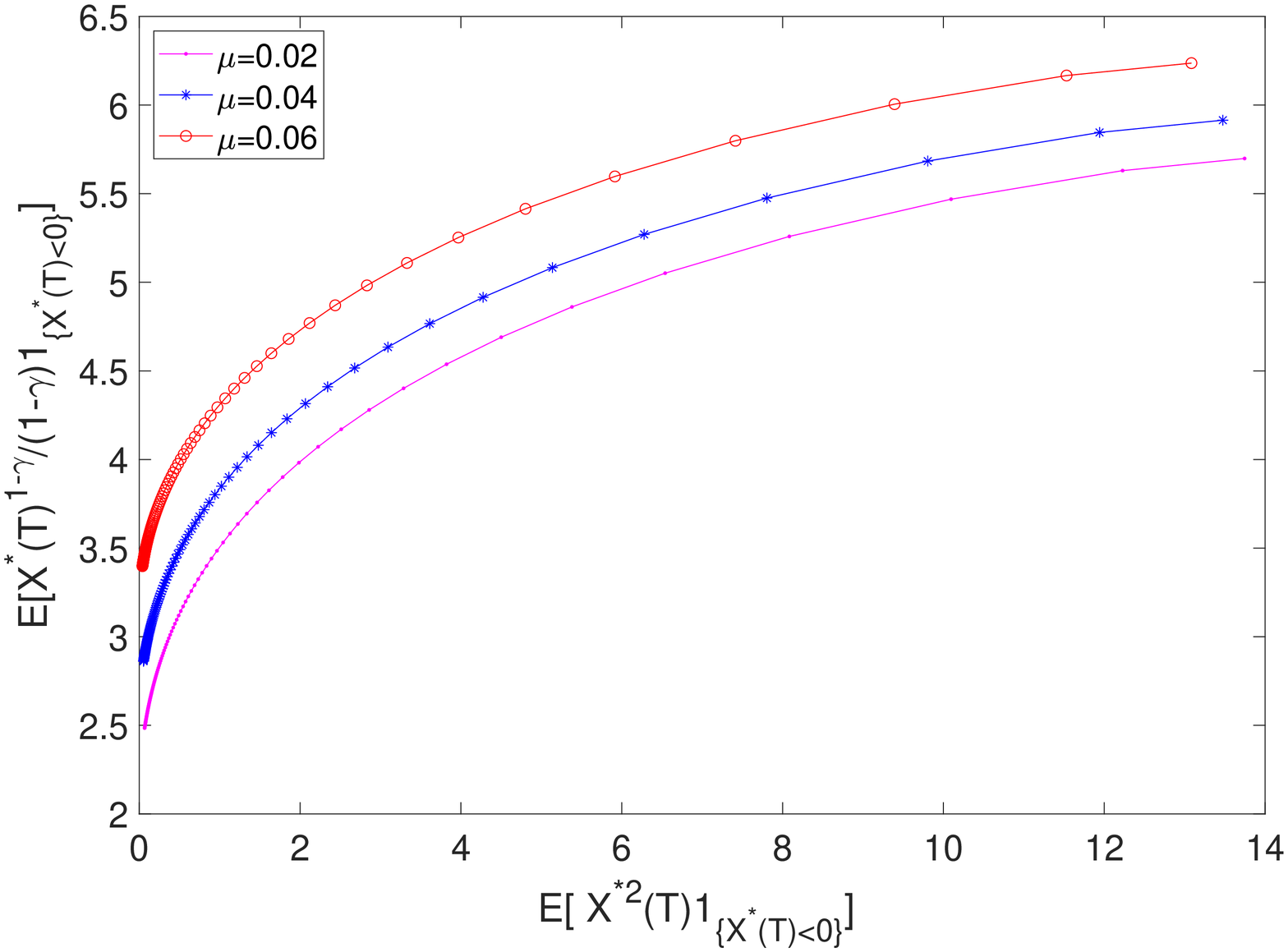}
		\caption{Effect of $\mu$.}
		\label{boundry4}
	\end{minipage}\hfill
	\begin{minipage}{0.5\textwidth}
		\centering
		\includegraphics[totalheight=5cm]{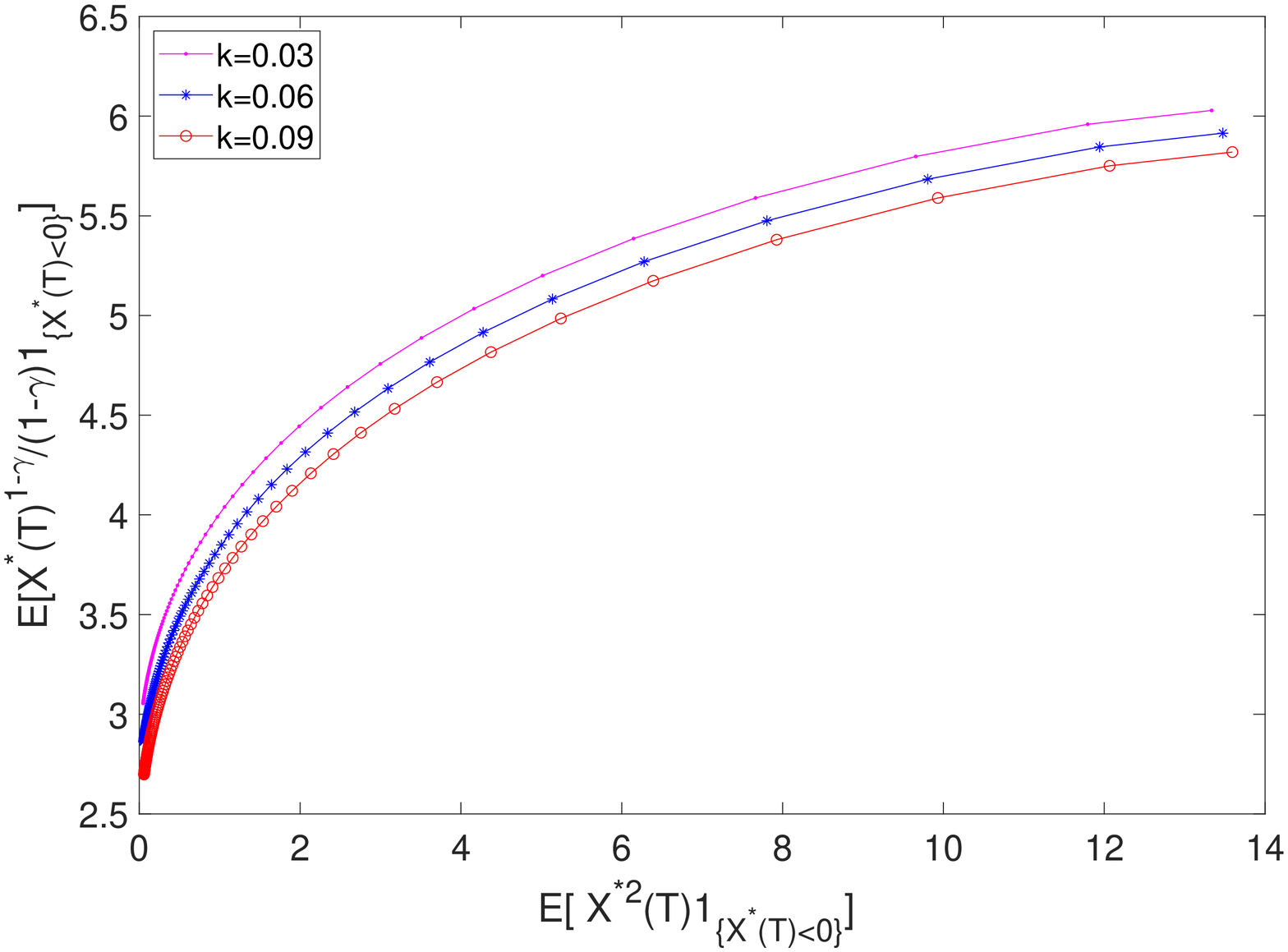}
		\caption{Effect of $k$.}
		\label{boundry7}
	\end{minipage}\hfill
	\begin{minipage}{0.5\textwidth}
		\centering
		\includegraphics[totalheight=5cm]{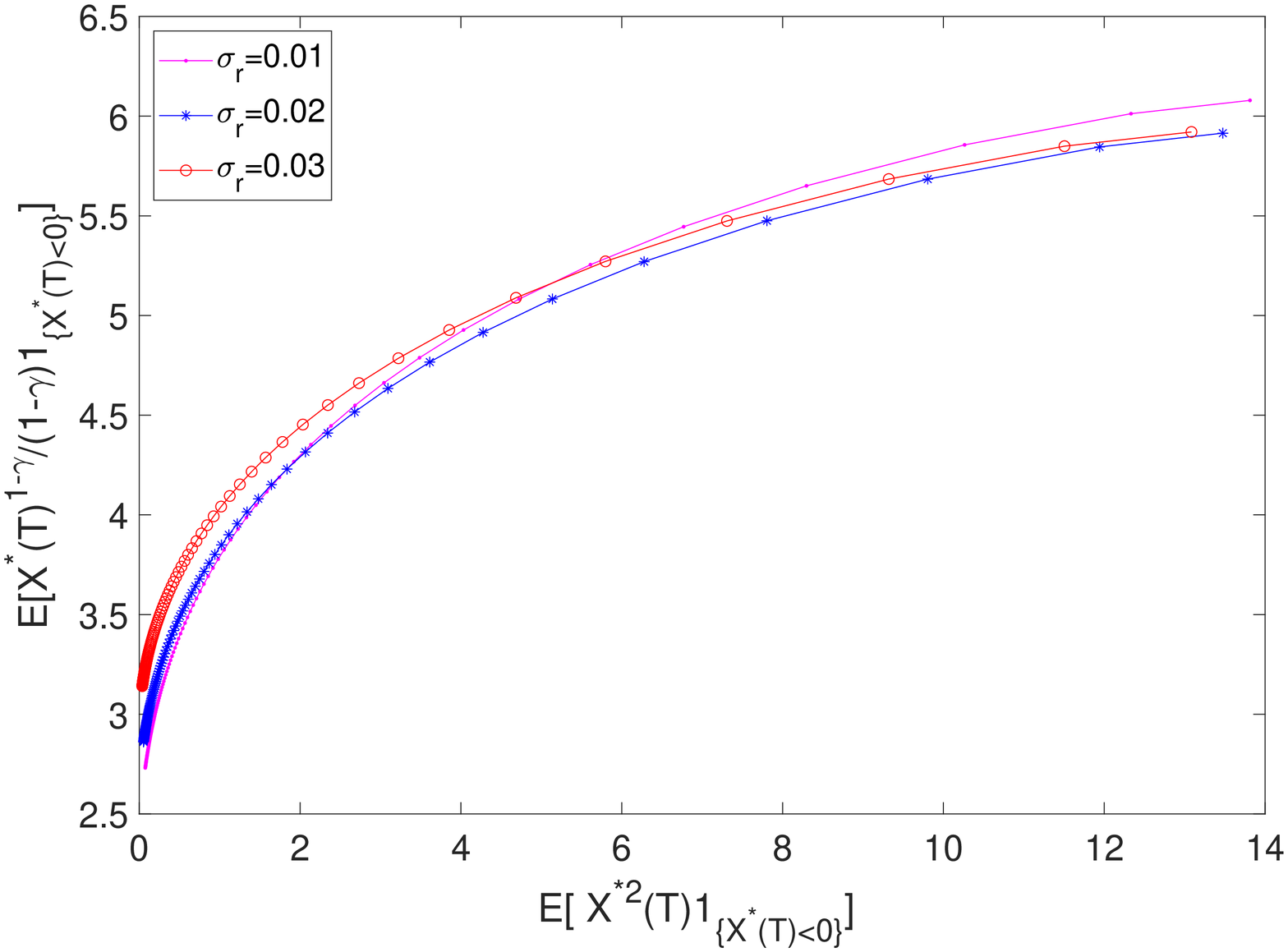}
		\caption{Effect of $\sigma_r$.}
		\label{boundry6}
	\end{minipage}\hfill
\end{figure}
The effects of $\mu$, $k$ and $\sigma_r$ are shown in Figs.~\ref{boundry4}, \ref{boundry7} and \ref{boundry6}, respectively. $\mu$ characterizes the increasing rate of the promised benefits and has two influences on the pension system. On the one hand, the liabilities of the manager increase with $\mu$, and the efficient frontier has a tendency to move downward with $\mu$. On the other hand, as the contribution rate is calculated by the spread method of fund amortization, the efficient has a tendency to move upward with $\mu$.  Fig.~\ref{boundry4} reveals that the second effect dominates. The constant $k$ reflects the rate at which surplus is amortized. When $k$ is positive and increases, the contribution rate decreases (increases) in the overfunded (underfunded) region. Fig.~\ref{boundry7} shows that the decrease of contribution rate in the overfunded region has a larger influence on the efficient frontier, i.e., the efficient frontier moves downward with $k$. The influence of interest risk on the efficient frontier is disclosed in Fig.~\ref{boundry6}, which shows that every two efficient frontiers intersect. As the interest risk impacts the actuarial discount rate and the returns from cash and bond simultaneously, the effect of $\sigma_r$ is complicated. If the solvency risk is small (large), the expected utility increases (decreases) with $\sigma_r$.
\vskip 10pt
\begin{figure}[htbp]
	\centering
	\begin{minipage}{0.5\textwidth}
		\centering
		\includegraphics[totalheight=5cm]{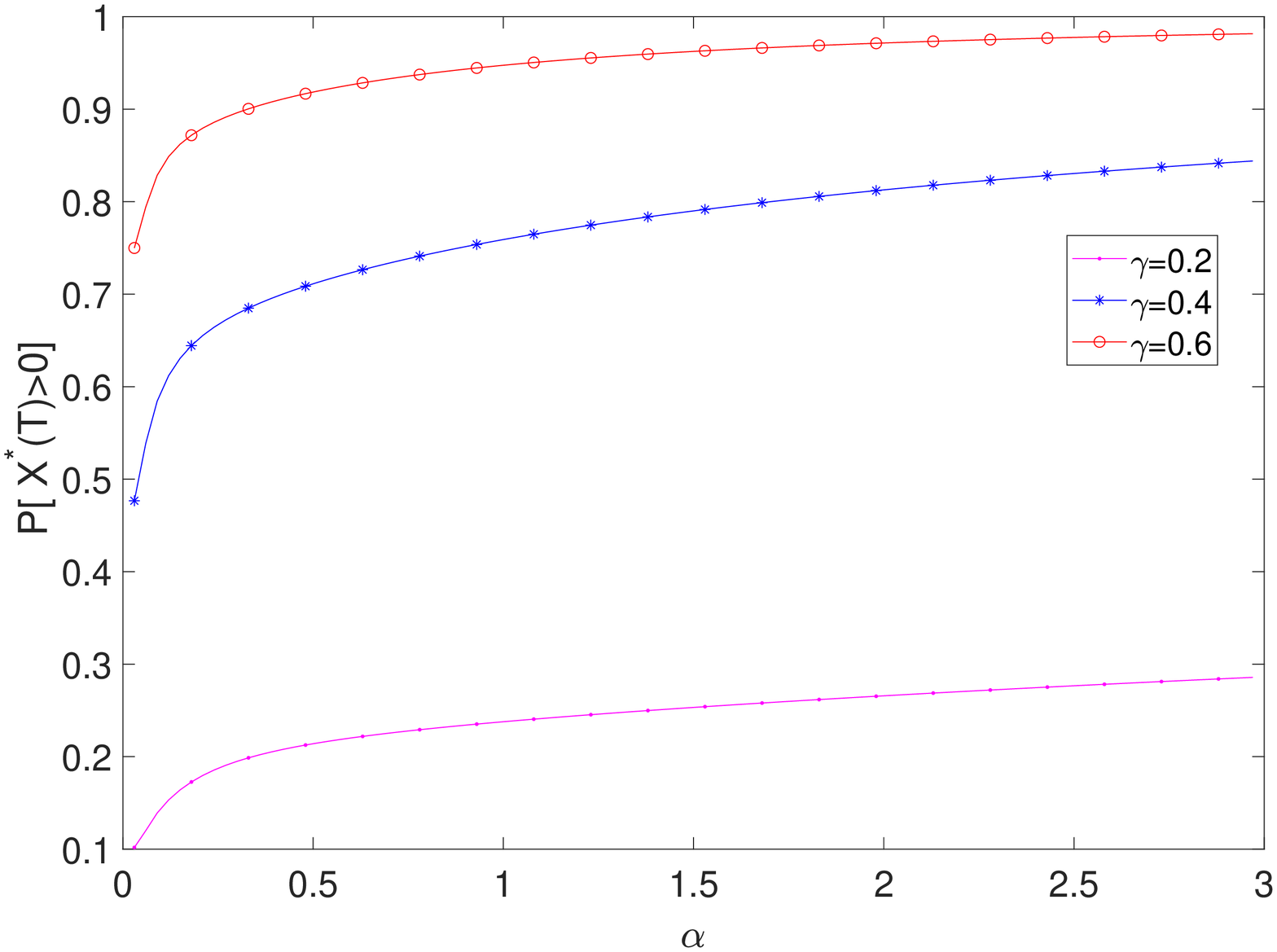}
		\caption{Probability in overfunded region w.r.t.  $\alpha$ and $\gamma$.}
		\label{p-alpha}
	\end{minipage}\hfill
	\begin{minipage}{0.5\textwidth}
		\centering
		\includegraphics[totalheight=5cm]{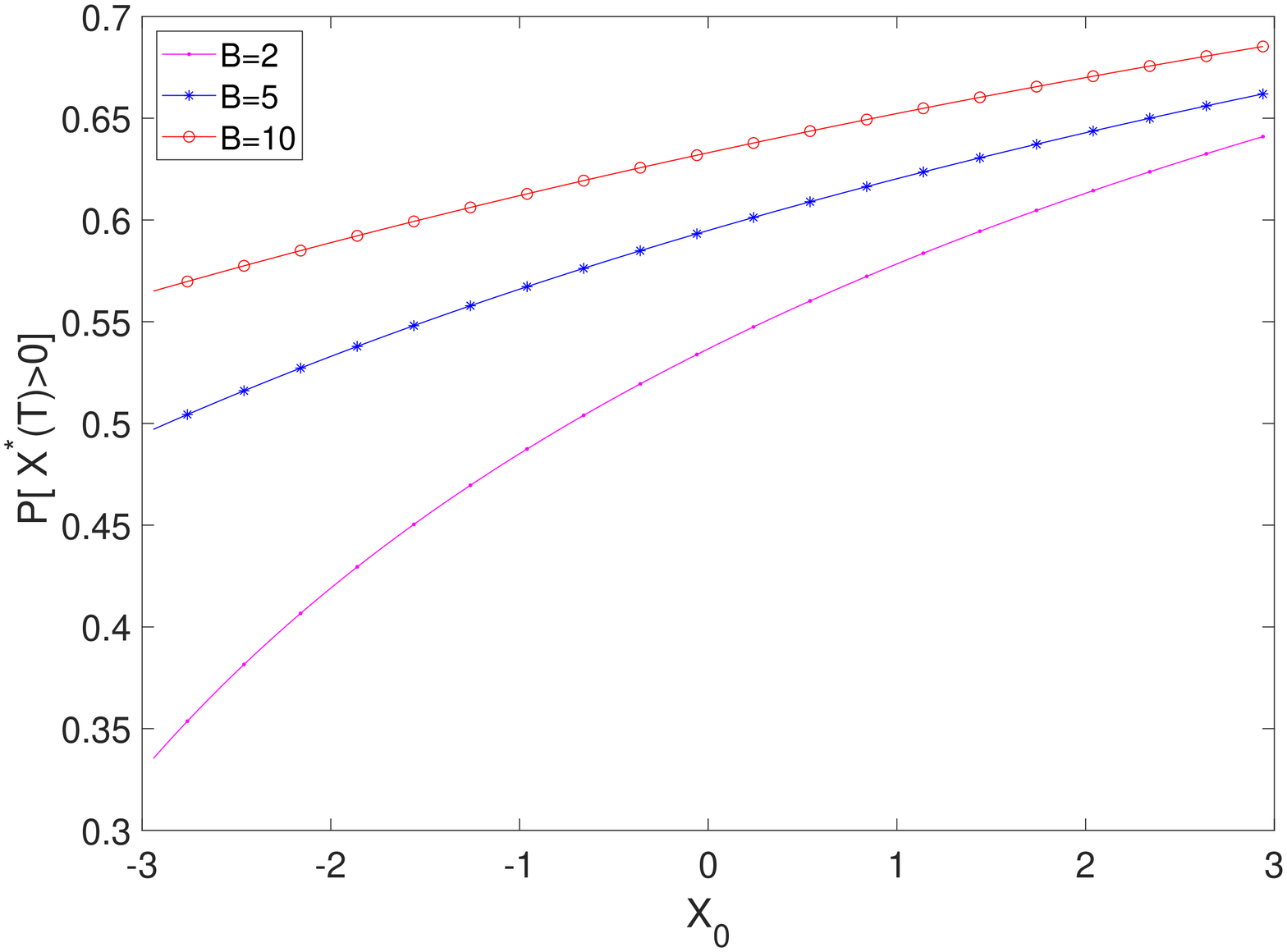}
		\caption{Probability in overfunded region w.r.t.  $X_0$ and $B$.}
		\label{p-x0}
	\end{minipage}\hfill
	\begin{minipage}{0.5\textwidth}
		\centering
		\includegraphics[totalheight=5cm]{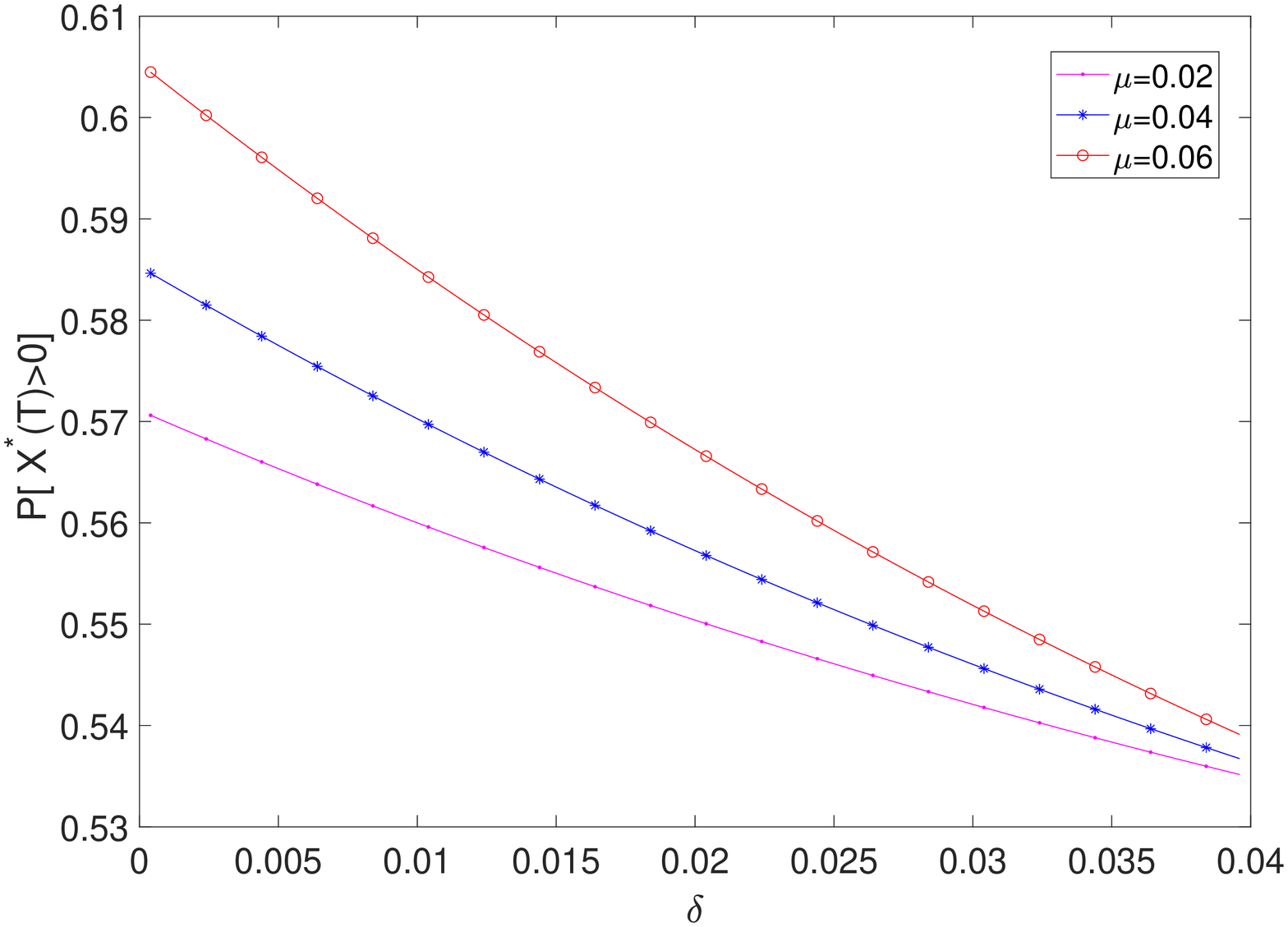}
		\caption{Probability in overfunded region w.r.t. $\delta$ and $\mu$.}
		\label{p-delta}
	\end{minipage}\hfill
	\begin{minipage}{0.5\textwidth}
		\centering
		\includegraphics[totalheight=5cm]{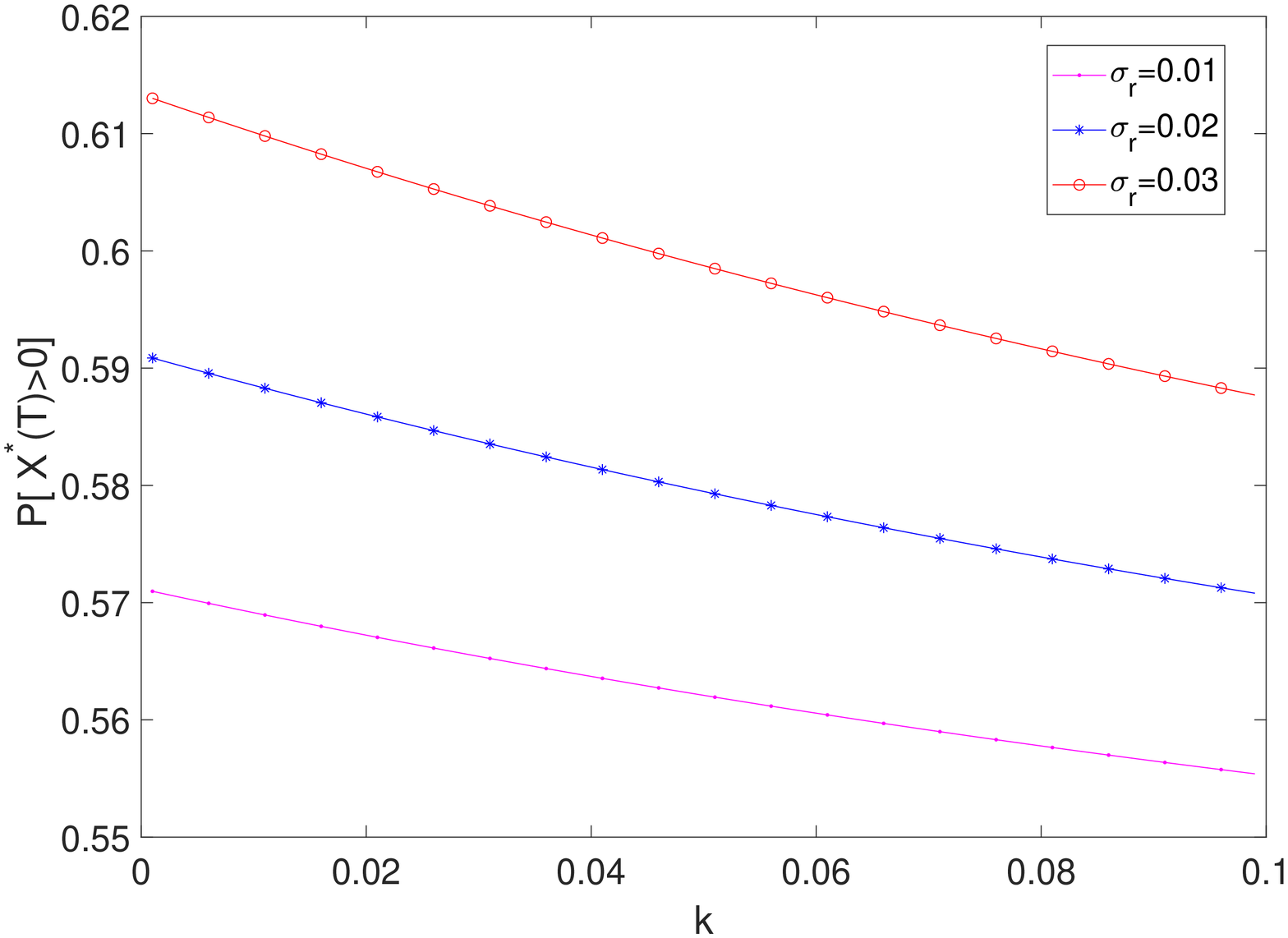}
		\caption{Probability in overfunded region w.r.t. $k$ and $\sigma_r$.}
		\label{p-k}
	\end{minipage}\hfill
\end{figure}
\subsection{\bf Probability in the overfunded region}
As the sum of the probabilities in the overfunded region and underfunded region is 1, we only plot the probability in the overfunded region. Figs.~\ref{p-alpha}-\ref{p-k} show the effects of different parameters on the probability in the overfunded region. The effects on the probability in the underfunded region are in the opposite direction. Fig.~\ref{p-alpha} illustrates the probability with different risk aversion coefficient $\gamma$ and solvency risk tolerance level $\alpha$. When $\gamma$ increases, the manager is more risk aversion towards financial risk and expects a larger probability in the overfunded region to attain higher expected utility, which is indicated by comparisons of three curves in  Fig.~\ref{p-alpha}. $\alpha$ characterizes the solvency risk tolerance level and a larger $\alpha$ indicates less tolerance towards solvency risk. As shown by Fig.~\ref{p-alpha},  the probability in the overfunded region increases with $\alpha$. Besides, we are also interested in the parameters $X_0$ and $B$, which are plotted in Fig.~\ref{p-x0}. It is natural to see that the probability in the overfunded region increases with $X_0$. Moreover, Fig.~\ref{p-x0} also shows that regardless of the initial status (overfunded/underfunded) of the fund, the terminal wealth has both the possibilities in the overfunded region and underfunded region in our framework. Moreover, the bound $B$ has a positive effect on the probability in the overfunded region. When the negative bound is big, the wealth in the underfunded region is restricted to a small horizon and the probability in the underfunded region increases. As such, the probability in the overfunded region increases with $B$.

In Fig.~\ref{p-delta}, the probabilities w.r.t. different $\delta$ and $\mu$ are illustrated. As stated above, $\delta$ and $\mu$ have two effects on the pension system. Smaller $\delta$ or larger $\mu$ indicates larger promised liabilities as well as larger contribution rate, which have the opposite effects on the evolution of the fund surplus. As such, it is not clear to clarify the effects of  $\delta$ and $\mu$. Fig.~\ref{p-delta} shows that probability in the overfunded region decreases with $\delta$  while increasing with $\mu$.  Comparing Figs.~\ref{p-k} and \ref{boundry7}, when $k$ increases, the decrease of contribution in the overfunded region has a larger effect. As such, the fund surplus has a larger probability to become underfunded with a larger $k$, which is well indicated by the decrease of probability in the overfunded region with $k$ in Fig.~\ref{p-k}. As stated in Fig.~\ref{boundry6}, the impact of interest risk $\sigma_r$ on the pension fund is complicated. Fig.~\ref{p-k} shows the positive relationship of $\sigma_r$ and the probability in the overfunded region.
\vskip 15pt	
\section{\bf Conclusion}
In this paper, we investigate the optimal management of an aggregated defined benefit pension plan. Different from previous studies, we suppose that {no matter} starting in the underfunded or overfunded region, the fund wealth may become underfunded or overfunded in the future. The goal of the manager in the underfunded region is to minimize the solvency risk as in \cite{JR10}. Meanwhile, in the overfunded region, the expected utility is considered as in \cite{JJ18}. We formulate the optimization rule of the pension manager as {a} weighted sum of the solvency risk and the expected utility.
	
The pension manager is faced with stochastic contribution rate, stochastic benefits, and {the interest risk in financial market}. In our paper, the fund surplus contains additional stochastic processes in the drift and diffusion terms. First, we replicate the additional stochastic processes in the drift term and transform the original problem into an equivalent one. {Using  martingale method} and replication technique, the explicit solutions of optimal wealth, portfolio, efficient frontier are obtained. We have four cases of the results, corresponding to high tolerance towards solvency risk, low tolerance towards solvency risk,  a specific lower bound, and high lower bound.

We also show detailed numerical examples to illustrate the manager's economic behaviors. The efficient frontier is very similar to that in the mean-variance analysis, which characterizes a balance {between  solvency risk and expected utility}. The impacts of different parameters on the optimal wealth process, optimal portfolio, efficient frontier  and probability in the overfunded region are also illustrated.

\vskip 15pt
{\bf Acknowledgements.}
The authors acknowledge the support from the National Natural Science Foundation of China (Grant  No.11901574, No.11871036). The authors thank the members of the group of Mathematical Finance and Actuarial Science at the Department of Mathematical Sciences, Tsinghua University for their feedbacks and useful conversations.

%\section{\bf Compliance with Ethical Standards}
%
%{\bf Funding }This work was supported by National Natural Science Foundation of China (Grant numbers 11901574,  and 11871036).
% 
%{\bf Conflict of Interest }The authors declare that they have no conflict of
%interest.
% 
%{\bf Ethical approval }This article does not contain any studies with human participants or animals performed by any of the authors.

\vskip 15pt
\appendix
	\renewcommand{\theequation}{\thesection.\arabic{equation}}
	\section{\bf Proof of Proposition \ref{ALNC}}\label{aALNC}
\begin{proof}
Using Eqs.~(\ref{equ-r}) and (\ref{equ-P}), we easily obtain the explicit solutions of $r(t)$ and $P(t)$ as follows:
\begin{eqnarray}
r(t)&=&b-e^{-at}\left[(b-r_0)+\sigma_r\int_0^te^{as}\rd W_r(s)\right],\label{r}\\
P(t)&=&P(0)\exp\left(\left(\mu-\frac12
\left(\sigma_{P_1}^2+\sigma_{P_2}^2\right)\right)t+\sigma_{P_1}W_r(t)+\sigma_{P_2}W_S(t)\right).\label{equ-P_1}
\end{eqnarray}
     Based on Eq.~(\ref{equ-r}), we also obtain the integral of the interest rate
%	Based on Eqs.~(\ref{equ-r}), (\ref{r}) and (\ref{equ-P_1}), we get\begin{equation*}
%		r(T)-r(t)=ab(T-t)-a\int_{t}^Tr(s)\rd s -\sigma_r(W_r(T)-W_r(t)),
%	\end{equation*} and  \begin{equation*}
%		r(T)-r(t)=(1-e^{-a(T-t)})(b-r(t))-e^{-aT}\sigma_r\int_{t}^Te^{as}\rd W_r(s),
%	\end{equation*} so we get
\begin{equation*}
		-\int_{t}^{t+d-x}r(s)\rd s=\frac{1-e^{-a(d-x)}}{a}(b-r(t))-b(d-x)+{\sigma_r}\int_{t}^{t+d-x}\frac{1-e^{-a(t+d+x-s)}}{a}\rd W_r(s).
	\end{equation*}
Combining the above equations, the conditional expectation in $AL(t)$ and $NC(t)$ can be calculated
	\begin{equation*}
		\begin{aligned}
			&\mathbb{E}\left[e^{-\int_{t}^{t+d-x} \hat{\delta}(s) d s}  P(t+d-x)  \mid \mathcal{F}_{t}\right]\\
			=&\exp\left({A(x,d)(b-r(t))-b(d-x)-\hat{\delta}(d-x)+\left(\mu-\frac12\sigma_{P_1}^2\right)(d-x)}\right)P(t)\\
			&\times\mathbb{E}\left[\exp\left({-\frac{\sigma_r}{a}e^{-a(t+d-x)}\int_{t}^{t+d-x} e^{as}\rd W_r(s)+(\frac{\sigma_r}{a}+\sigma_{P_1})(W_r(t+d-x)-W_r(t))}\right)\right]\\
&\times\mathbb{E}\left[\exp\left({\sigma_{P_2}(W_S(t+d-x)-W_S(t))-\frac12\sigma_{P_2}^2(d-x)}\right)\right]\\
=&\exp\left({A(x,d)(b-r(t))-b(d-x)-\hat{\delta}(d-x)+\left(\mu-\frac12\sigma_{P_1}^2\right)(d-x)}\right)P(t)\\
			&\times \exp \left(-\frac{\sigma_r^2}{4a}A(x,d)^2-\left(\frac{\sigma_r \sigma_{P_1}}{a}+\frac{\sigma_r^2}{2a^2}\right)A(x,d) +\left(\frac{\sigma_r \sigma_{P_1}}{a}+\frac{\sigma_r^2}{2a^2}+\frac{1}{2} \sigma_{P_1}^2\right)(d-x)\right)\times1\times1\\
=&\!\!\exp \left(-\frac{\sigma_r^2}{4a}A(x,d)^2\!+\!\left(b\!-\!r(t)\!-\!\frac{\sigma_r \sigma_{P_1}}{a}\!-\!\frac{\sigma_r^2}{2a^2}\right)A(x,d)\! +\!\left(\frac{\sigma_r \sigma_{P_1}}{a}\!+\!\frac{\sigma_r^2}{2a^2}\!-\! b\!-\!\delta\!+\!\mu\right)(d\!-\!x)\right)P(t).
		\end{aligned}
	\end{equation*}
	Denote
	\begin{equation*}\begin{aligned}
			D(x,d)= -\frac{\sigma_r^2}{4a}A(x,d)^2+\left(b-\frac{\sigma_r \sigma_{P_1}}{a}-\frac{\sigma_r^2}{2a^2}\right)A(x,d) +\left(\frac{\sigma_r \sigma_{P_1}}{a}+\frac{\sigma_r^2}{2a^2}- b-\delta+\mu\right)(d-x),
		\end{aligned}
	\end{equation*}
	 we have
	\begin{equation*}
		\begin{split}
			AL(t)=\int_{m}^{d} e^{-r(t)A(x,d)+D(x,d)}M(x)\rd x P(t),\\
			NC(t)=\int_{m}^{d} e^{-r(t)A(x,d)+D(x,d)}M'(x)\rd x P(t).
		\end{split}
	\end{equation*}
	Denote
	\begin{equation*}
		f_0(r(t))=\int_{m}^{d}  e^{-r(t)A(x,d)+D(x,d)}M(x)\rd x.
	\end{equation*}
	Because
\begin{align*}
D'(x,d)&=-\frac{\sigma_r^2}{2a}A(x,d)A'(x,d)+\left(b-\frac{\sigma_r^2}{2a^2}
-\frac{\sigma_{P_1}\sigma_r}{a}\right)A'(x,d)\\
&-\left(\frac{\sigma_r \sigma_{P_1}}{a}+\frac{\sigma_r^2}{2a^2}- b-\delta+\mu\right)\\
&=-\frac{\sigma_r^2}{2a}A(x,d)(aA(x,d)-1)+\left(b-\frac{\sigma_r^2}{2a^2}
-\frac{\sigma_{P_1}\sigma_r}{a}\right)(aA(x,d)-1)\\
&-\left(\frac{\sigma_r \sigma_{P_1}}{a}+\frac{\sigma_r^2}{2a^2}- b-\delta+\mu\right)\\
&=-\frac{\sigma_r^2}{2}A^2(x,d)+\left(ab-{\sigma_{P_1}\sigma_r}\right)A(x,d)+\left(\delta-\mu\right),
\end{align*}
we obtain \begin{equation*}
		\rd f_0(r(t))=f_1(r(t))\rd t+f_2(r(t))\rd W_r(t),
	\end{equation*}
where
\begin{equation*}
	\begin{aligned}
		f_1(r(t))=&\int_{m}^{d}  e^{-r(t)A(x,d)+D(x,d)}M(x)A(x,d)\left(\frac{\sigma_r^2}{2}A(x,d)-a(b-r(t))\right)\rd x,\\
		f_2(r(t))=&\sigma_r\int_{m}^{d}  e^{-r(t)A(x,d)+D(x,d)}M(x)A(x,d)\rd x.
	\end{aligned}
\end{equation*}
Then, by the It\^{o}'s formula, $AL$ satisfies the following SDE
\begin{equation*}\begin{aligned}
			\rd AL(t)=&[P(t)(f_1(r(t))+\sigma_{P_1}f_2(r(t)))+\mu AL(t)]\rd t\\&+[P(t)f_2(r(t))+\sigma_{P_1}AL(t)]\rd W_r(t)+\sigma_{P_2}AL(t)\rd W_S(t).
		\end{aligned}
	\end{equation*}	
	\end{proof}
\section{\bf Proof of Proposition \ref{prop:x}}\label{aprop:x}
Similar to the proof to Proposition \ref{ALNC}, denote
\begin{equation*}
	\psi_{NC}(t)=\int_{m}^{d} e^{-r(t)A(x,d)+D(x,d)}M'(x)\rd x ,
\end{equation*}
then we have \begin{align*}
	\psi_{NC}(t)=&\int_{m}^{d} e^{-r(t)A(x,d)+D(x,d)}M'(x)\rd x \\
	=&M(d)-e^{-r(t)A(m,d)+D(m,d)}M(m)\\
&-\int_{m}^{d}  e^{-r(t)A(x,d)+D(x,d)}M(x)(-r(t)A'(x,d)+D'(x,d))\rd x\\
	=&1-\int_{m}^{d}  e^{-r(t)A(x,d)+D(x,d)}M(x)(-r(t)A'(x,d)+D'(x,d))\rd x.
\end{align*}
As such,
\begin{align*}
	NC(t)=P(t)-\int_{m}^{d}  e^{-r(t)A(x,d)+D(x,d)}M(x)(-r(t)A'(x,d)+D'(x,d))\rd xP(t).
\end{align*}
Applying It\^{o}'s formula to $X(t)$, we have
\!\!\!\!\!\begin{eqnarray*}
\!\!\!\!\mathrm{d}X(t)&\!=\!&r(t)F(t)\mathrm{d}t\!+\![u_B(t)h(K)\!+\!u_S(t)\sigma_1](\lambda_r\mathrm{d}t
\!+\!\mathrm{d}W_r(t))\!+\!u_S(t)\sigma_2(\lambda_S\mathrm{d}t+\mathrm{d}W_S(t))\\
&+&\left(-\int_{m}^{d}  e^{-r(t)A(x,d)+D(x,d)}M(x)(-r(t)A'(x,d)+D'(x,d))\rd xP(t)-kX(t)\right)\mathrm{d}t\\
&-&[P(t)(f_1(r(t))+\sigma_{P_1}f_2(r(t)))+\mu AL(t)]\rd t
\\
&-&[P(t)f_2(r(t))+\sigma_{P_1}AL(t)]\rd W_r(t)-\sigma_{P_2}AL(t)\rd W_S(t)\\
&=&(r(t)-k)X(t)\rd t-AL(t)((\mu-r(t))\rd t+\sigma_{P_1}\rd W_r(t)+\sigma_{P_2}\rd W_S(t)\\
&&-P(t)\int_{m}^{d}  e^{-r(t)A(x,d)+D(x,d)}M(x)(r(t)-{\sigma_{P_1}\sigma_r}A(x,d)+\left(\delta-\mu\right))
\rd x\rd t\\
&&+[u_B(t)h(K)+u_S(t)\sigma_1](\lambda_r\mathrm{d}t+\mathrm{d}W_r(t))
+u_S(t)\sigma_2(\lambda_S\mathrm{d}t+\mathrm{d}W_S(t))\\
&&-P(t)f_2(r(t))(\sigma_{P_1}\rd t+\rd W_r(t))\\
&=&\!(r(t)\!-\!k)X(t)\rd t\!-\!P(t)f_2(r(t))\rd W_r(t)\!-\!AL(t)\left(\delta\rd t\!+\!\sigma_{P_1}\rd W_r(t)\!+\!\sigma_{P_2}\rd W_S(t)\right)\\
&&+[u_B(t)h(K)+u_S(t)\sigma_1](\lambda_r\mathrm{d}t+\mathrm{d}W_r(t))
+u_S(t)\sigma_2(\lambda_S\mathrm{d}t+\mathrm{d}W_S(t))
\end{eqnarray*}
\begin{eqnarray*}
&=&(r(t)-k)X(t)\rd t+\lambda_rP(t)f_2(r(t))\rd t+\left(\lambda_r\sigma_{P_1}+\lambda_S\sigma_{P_2}-\delta\right)P(t)f_0(r(t))\rd t\\
&&+[u_B(t)h(K)+u_S(t)\sigma_1-P(t)f_2(r(t))-\sigma_{P_1}P(t)f_0(r(t))](\lambda_r\mathrm{d}t+\mathrm{d}W_r(t))\\
&&+[u_S(t)\sigma_2-\sigma_{P_2}P(t)f_0(r(t))](\lambda_S\mathrm{d}t+\mathrm{d}W_S(t)).
\end{eqnarray*}
\section{\bf  Replication of additional drift terms}\label{replication}
\subsection{\bf Replication of $\{\lambda_rP(t)f_2(r(t))\}$}
First, we replicate the term $\lambda_rP(t)f_2(r(t))$. We introduce an auxiliary  process $G=\left\{G(t,r(t),P(t),s)|0\leqslant t\leqslant s\right\} $ to represent  the price of an asset with payment $P(s)f_2(r(s))$ at maturity time $s$. By the Markov's property of $P$ and $r$, we guess that $G$ has the form
\[G(t,r(t),P(t),s)=P(t)g(t,r(t),s).\]
By the theory of derivative pricing, $g(t,r(t),s)$ satisfies the following partial differential equation (PDE)
\begin{equation}\label{equ:g}
	g_t+\mu g+a(b-r)g_r+\frac{1}{2}\sigma_r^2g_{rr}-\sigma_r\sigma_{P_1}g_r=(r-k)g+
	\lambda_r(\sigma_{P_1}g-\sigma_rg_r)+\lambda_S\sigma_{P_2}g
\end{equation}
with boundary condition $g(s,r,s)=f_2(r)$.

Meanwhile, $G(t,r(t),P(t),s)$ satisfies the following BSDE
\begin{equation}\label{equ:G}
	\left\{\begin{aligned}
		&\rd G(t,r,P,s)\!=\!P\left\{((r\!-\!k)g\!+\!\lambda_r(\sigma_{P_1}g\!-\!\sigma_rg_r)\!+\!\lambda_S\sigma_{P_2}g)\rd t\!\!+\!\!(\sigma_{P_1}g-\sigma_rg_r)\rd W_r\!\!+\!\!\sigma_{P_2}g\rd W_s\right\}
		,\\&G(s,r(t),P(t),s)=P(t)f_2(r(t)).
	\end{aligned}\right.
\end{equation}
Then, we solve PDE~(\ref{equ:g}). By the boundary condition $g(s,r,s)=f_2(r)$ and the form of $f_2(r)$, we guess that $g$ has the following form
\begin{equation*}
	g(t,r(t),s)=\sigma_r\int_{m}^de^{-r(t)\tilde{A}(t,s,x,d)+\tilde{D}(t,s,x,d)}M(x)A(x,d)\rd x
\end{equation*}
with boundary condition $\tilde{A}(s,s,x,d)=A(x,d)$ and $\tilde{D}(s,s,x,d)=D(x,d)$.  Substituting the last equation into PDE~(\ref{equ:g}) and arranging it by the orders of $r$, $\tilde{A}(t,s,x,d)$ and $\tilde{D}(t,s,x,d)$ satisfy the following ODEs
\begin{equation*}
	\left\{\begin{aligned}
		&-\tilde{A}_t(t,s,x,d)+aA(t,s,x,d)=1,\quad \tilde{A}(s,s,x,d)=A(x,d),\\
		&\tilde{D}_t(t,s,x,d)\!+\!\mu\!+\!\frac{1}{2}\sigma_r^2\tilde{A}^2\!-\!ab\tilde{A}\!+\!\sigma_r\sigma_{P_1}\tilde{A}\!=\!\lambda_r\sigma_{P_1}\!+\!\lambda_S\sigma_{P_2}\!+\!\lambda_r\sigma_r\tilde{A}\!-\!k,\quad\tilde{D}(s,s,x,d)\!=\!D(x,d).
	\end{aligned}\right.
\end{equation*}
Solving the above equations, we obtain the following results
\begin{equation*}
	\left\{\begin{aligned}
		\tilde{A}(t,s,x,d)&=A(t,s)+A(x,d)e^{-a(s-t)},\\
		\tilde{D}(t,s,x,d)&=\frac{\sigma_r^2}{2a^3}(4e^{a(t-s)}-3-e^{2a(t-s)})+bA(t,s)\\
		&+A(x,d)b(e^{-a(s-t)}-1)+D(x,d)+(b-k-\mu+\lambda_r\sigma_{P_1}+\lambda_S\sigma_{P_2})(t-s)\\
		&+\frac{\sigma_r^2}{2a} A(x,d)^2(1-e^{2a(t-s)})-\frac{\sigma_r}{a}(\lambda_r-\sigma_{P_1})(A(x,d)(1-e^{a(t-s)})+s-t)\\
		&+\frac{\sigma_r}{a^2}(A(x,\!d)\sigma_re^{2a(t-s)}\!+\!e^{a(t-s)}(\sigma_{P_1}\!-\!\lambda_r\!-\!2\sigma_rA(x,d))\!+\!\lambda_r\!-\!\sigma_{P_1}\!+\!\sigma_r(A(x,\!d)\!-\!t\!+\!s)).
	\end{aligned}\right.
\end{equation*}
In order to replicate the cash flow with  rate of $P(t)f_2(r(t))$, we introduce a process
\begin{equation*}
	H(t,r(t),P(t),T)=\int_t^TG(t,r(t),P(t),s)\rd s=P(t)\int_t^Tg(t,r(t),s)\rd s,
\end{equation*}
where $H(t,\!r(t),\!P(t),\!T)$ represents the fair value (at time $t$) of the aggregated cash flow of rate $\lambda_rP(s)f_2(r(s)$ within $s\in[t,T]$. By It\^{o}'s formula and Eq.~(\ref{equ:G}), $H(t,r(t),P(t),T)$ satisfies the following BSDE
\begin{equation}\label{sde:H}
	\left\{	\begin{aligned}
		\rd H(t,r(t),P(t),T)=&-P(t)f_2(r(t))\rd t+(r(t)-k)H(t,r(t),P(t),T)\rd t\\
		&+\int_t^T(\sigma_{P_1}g(t,r(t),s)-\sigma_rg_r(t,r(t),s))\rd s (\lambda_r\rd t+\rd  W_r(t))\\
		&+\sigma_{P_2}\int_t^Tg(t,r(t),s)\rd s(\lambda_S\rd t+\rd W_S(t)),\\
		H(T,r(T),P(T),T)=&0.
	\end{aligned}\right.
\end{equation}
\subsection{\bf Replication of $\{(\lambda_r\sigma_{P_1}+\lambda_S\sigma_{P_2}-\delta)P(t)f_0(r(t))\}$}
We now replicate the term $(\lambda_r\sigma_{P_1}+\lambda_S\sigma_{P_2}-\delta)P(t)f_0(r(t))$ in Eq.~(\ref{SDE-X}). Compared with $f_0(r)$, $f_2(r)$ has an additional term $A(x,d)$ in the integral.  Denote $\tilde{G}=\left\{\tilde{G}(t,r(t),P(t),s)|0\leqslant t\leqslant s\right\} $ to represent  the price of an asset with payment $P(s)f_0(r(s))$ at maturity time $s$. Denote $\tilde{H}(t,r(t),P(t),T)=\int_t^T\tilde{G}(t,r(t),s)\rd s$ as the present value of future cash flow with rate of $P(s)f_0(r(s))$ within $s\in[t,T]$. We use the process $\tilde{H}=\left\{\tilde{H}(t,r(t),P(t),s)|0\leqslant t\leqslant s\right\} $ to replicate the term $P(t)f_0(r(t))$ in Eq.~(\ref{SDE-X}).  By the similar arguments of the last subsection, we omit the tedious calculation and directly show the results for $\tilde{G}$ and $\tilde{H}$.

$\tilde{G}$ is explicitly obtained as follows:
\[
\tilde{G}(t,r(t),P(t),s)=P(t)\tilde{g}(t,r(t),s),
\]
where
\begin{equation*}
	\tilde{g}(t,r(t),s)=\int_{m}^de^{-r(t)\tilde{A}(t,s,x,d)+\tilde{D}(t,s,x,d)}M(x)\rd x.
\end{equation*}
Moreover,  $\tilde{G}$ satisfies the following BSDE
\begin{equation}\label{equ:G_tilde}
	\left\{\begin{aligned}
		&\rd \tilde{G}(t,r,P,s)\!=\!P\!\left\{((r\!-\!k)\tilde{g}\!+\!\lambda_r(\sigma_{P_1}\tilde{g}\!-\!\sigma_r\tilde{g}_r)\!
		+\!\lambda_S\sigma_{P_2}\tilde{g})\rd t
		\!+\!(\sigma_{P_1}\tilde{g}\!\!-\!\sigma_r\tilde{g}_r)\rd W_r\!+\!\sigma_{P_2}\tilde{g}\rd W_s\right\}
		,\\&\tilde{G}(s,r(t),P(t),s)=P(t)f_0(r(t)).
	\end{aligned} \right.
\end{equation}
The closed form of $\tilde{H}$ is
\begin{equation*}
	\tilde{H}(t,r(t),P(t),T)=P(t)\int_t^T\tilde{g}(t,r(t),s)\rd s.
\end{equation*}
$\tilde{H}$ also satisfies the following BSDE
\begin{equation}\label{equ:h2}
	\begin{cases}
		\rd \tilde{H}(t,r(t),P(t),T)=&-P(t)f_0(r(t))\rd t+(r(t)-k)\tilde{H}(t,r(t),P(t),T)\rd t\\
		&+\int_t^T(\sigma_{P_1}\tilde{g}(t,r(t),s)-\sigma_r\tilde{g}_r(t,r(t),s))\rd s (\lambda_r\rd t+\rd  W_r(t))\\
		&+\sigma_{P_2}\int_t^T\tilde{g}(t,r(t),s)\rd s(\lambda_S\rd t+\rd W_S(t)),
		\\\tilde{H}(T,r(T),P(T),T)=&0.
	\end{cases}
\end{equation}

\section{\bf Proof of Lemma {\ref{lemma:xy}}}\label{alemma:xy}
%For given $ y $, $ {f(x)-yx} $ is a continuous function with respect to $ x $, we  derive the optimal solution by comparing the possible maximal points. The possible maximal points of Problem (\ref{DPY3}) are $$y^{-\frac{1}{\gamma}}, -\frac{y}{2\alpha}, -B.$$ There are two cases to compare the values of these there points.
Solving the following equation:\begin{equation*}
		\frac{\frac{z_1^{1-\gamma}}{1-\gamma}-(-\alpha z_2^2)}{z_1-z_2}=z_1^{-\gamma}=-2\alpha z_2\triangleq k_1,
\end{equation*}
we get $ z_1= (\frac{1-\gamma}{4\alpha\gamma})^{\frac{1}{1+\gamma}}  $, $ z_2=-\frac{1}{2\alpha}(\frac{4\alpha\gamma}{1-\gamma})^{\frac{\gamma}{1+\gamma}}  $ and $ k_1= (\frac{4\alpha\gamma}{1-\gamma})^{\frac{\gamma}{1+\gamma}} $. From the geometric meaning of the formula, we know that  $ k_1 $ is the slope of the common tangent of the graph of $ y=-\alpha x^2(x<0) $ and the graph of $ y=\frac{x^{1-\gamma}}{1-\gamma}(x>0) $.

If $ z_2>-B $, i.e., $k_1<2\alpha B $,  we see from \cite{dong2020optimal} that the concave envelope of $ f(\cdot) $ is \begin{equation*}
	f^c(x)=\left\{ \begin{aligned}
&-\alpha x^2&-B\leqslant x<z_2\\
&k_1(x-z_2)-\alpha z_2^2&z_2\leqslant x\leqslant z_1\\
&\frac{x^{1-\gamma}}{1-\gamma}&z_1<x.
	\end{aligned}\right.
\end{equation*}
In this case, $ f(\cdot) $ and $ f^c(\cdot)$ are shown in  Fig.~\ref{f1}.
\begin{figure}[tbph]
	\centering
	\includegraphics[width=0.45\linewidth]{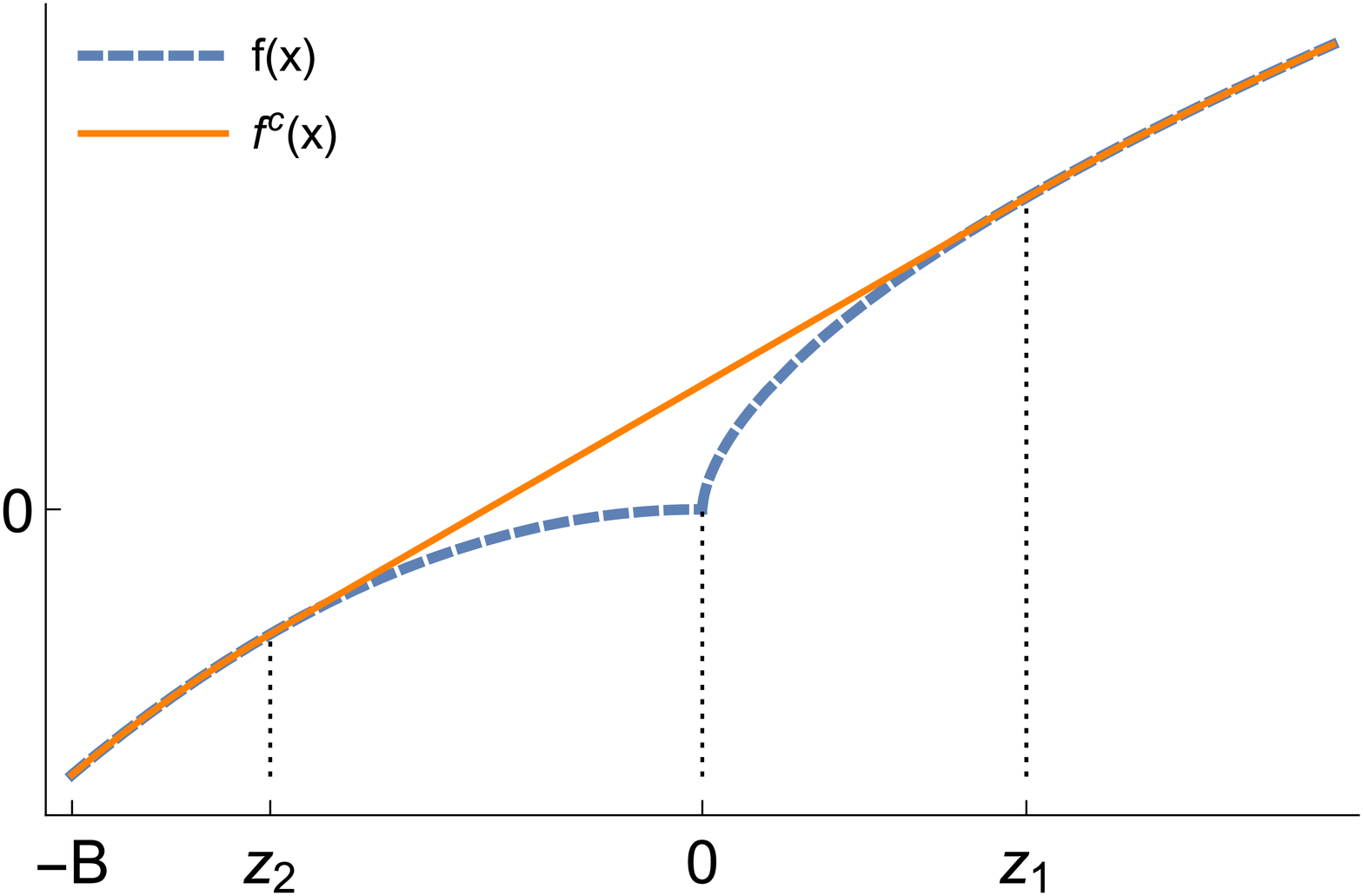}
	\caption{Graph of $ f(\cdot) $ and its concave envelope $f^c(\cdot)$ when $k_1<2\alpha B $.}
	\label{f1}
\end{figure}
 Applying subgradient method to  $f^c(x)-yx$ and noticing  that $ f(x)-yx $ and  $f^c(x)-yx$ coincide when $x\leqslant z_1$ or $x\geqslant z_2$, we find the solution to Problem (\ref{DPY3}) is
\begin{equation*}
	x^*(y)=\left\{\begin{aligned}
		&y^{-\frac{1}{\gamma}},&0<y< k_1,\\
		&k_1^{-\frac{1}{\gamma}} \text{ or } -\frac{k_1}{2\alpha},&y=k_1\\
		&-\frac{y}{2\alpha},&k_1<y\leqslant 2\alpha B,\\
		&-B,&y> 2\alpha B.
	\end{aligned}\right.
\end{equation*}

If $ z_2\leqslant -B $, i.e.,  $k_1\geqslant2\alpha B $, let $ z_0 $ denote the solution to
\begin{equation*}
	\frac{\frac{z^{1-\gamma}}{1-\gamma}-(-\alpha B^2)}{z-(-B)}=z^{-\gamma}
\end{equation*}
and let $ k_2=z_0^{-\gamma} $, then we can also see from \cite{dong2020optimal} that the concave envelope of $ f(\cdot) $ is \begin{equation*}
	f^c(x)=\left\{ \begin{aligned}
		&k_2(x+B)-\alpha B^2&-B\leqslant x\leqslant z_0\\
		&\frac{x^{1-\gamma}}{1-\gamma}&z_0<x.
	\end{aligned}\right.
\end{equation*}
  In this case, $k_2$ is also the slope of the common tangent of the graph of $ y=-\alpha x^2(-B<x<0) $ and the graph of $ y=\frac{x^{1-\gamma}}{1-\gamma}(x>0) $, $ f(\cdot) $ and  $f^c(\cdot)$ are shown in  Fig.~\ref{f2}.
\begin{figure}[tbph]
	\centering
	\includegraphics[width=0.45\linewidth]{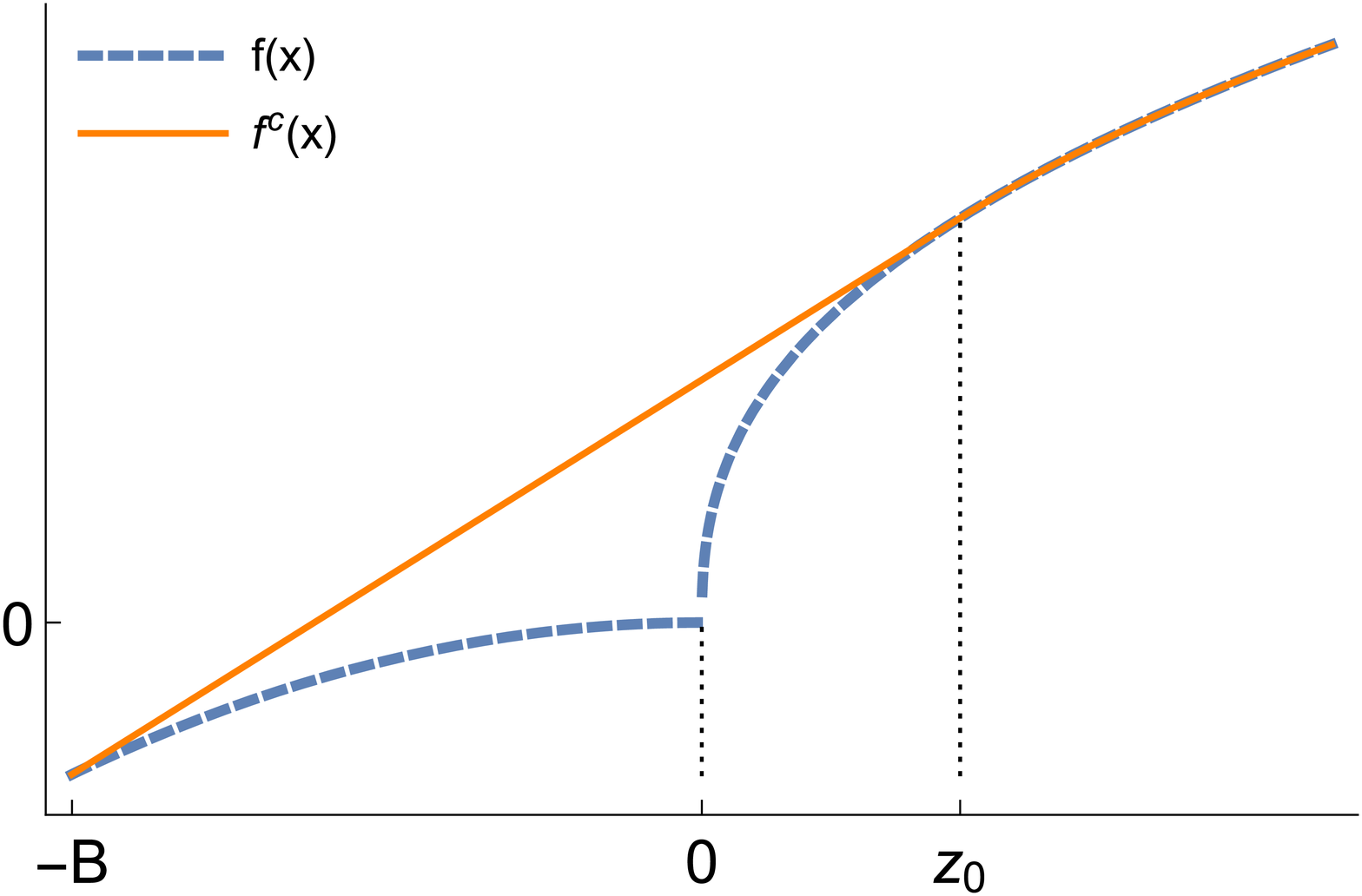}
	\caption{Graph of $ f(\cdot) $ and its concave envelope $f^c(\cdot)$ when $k_1\geqslant2\alpha B $.}
	\label{f2}
\end{figure}
 Applying subgradient method to  $f^c(x)-yx$ and noticing  that $ f(x)-yx $ and  $f^c(x)-yx$ coincide when  $x\geqslant z_0$, we find the solution to Problem (\ref{DPY3}) is
 \begin{equation*}
	x^*(y)=\left\{\begin{aligned}
		&y^{-\frac{1}{\gamma}},&0<y< k_2,\\
		&k_2^{-\frac{1}{\gamma}}\text{ or }-B,&y= k_2\\
		&-B,&y>k_2.
	\end{aligned}\right.
\end{equation*}
\section{\bf Proof of Theorem {\ref{thm:ef}}}\label{athm:ef}	
Based on the closed-form of $X^*(T)$ presented in Theorem \ref{th:y}, we  directly calculate the solvency risk (expected utility) in the underfunded (overfunded) region as follows: \\
If $k_1<2\alpha B $ and $Y_0>-B\exp\left(M+\frac{1}{2}V^2\right)$,  we have
\begin{eqnarray*}	
\mathbb{E}\left[ \frac{X^*(T)^{1-\gamma}}{1-\gamma}1_{\{X^*(T)>0\}}\right]&=&\mathbb{E}\left[ \frac{x^*(\beta^* \rho(T))^{1-\gamma}}{1-\gamma}1_{\{x^*(\beta^* \rho(T))>0\}}\right]\\
&=&\mathbb{E}\left[ \frac{(\beta^* \rho(T))^{-\frac{1-\gamma}{\gamma}}}{1-\gamma}1_{\{0<\beta^* \rho(T)\leqslant k_1\}}\right]\\		&=&\frac{{\beta^*}^{\frac{\gamma-1}{\gamma}}}
{1-\gamma}\exp\left(\frac{\gamma-1}
{\gamma}M+\frac{V^2(\gamma-1)^2}{2\gamma^2}\right)
\Phi\left(\Upsilon(k_1,\beta^*)-V\frac{\gamma-1}{\gamma}\right)
\end{eqnarray*}
and
\begin{eqnarray*}
{\mathbb{E}\left[ {X^*}^2(T)1_{\{{X^*}(T)<0\}}\right]}&=&\mathbb{E}\left[ x^*(\beta^* \rho(T))^21_{\{x^*(\beta^* \rho(T))<0\}}\right]\\
&=&\mathbb{E}\left[ \left(-\frac{1}{2\alpha}\beta^* \rho(T)\right)^21_{\{k_1<\beta^* \rho(T)\leqslant 2\alpha B\}}+(-B)^21_{\{\beta^* \rho(T)> 2\alpha B\}}\right]\\
&=&\frac{{\beta^*}^2}{4\alpha^2}\exp(2(M+V^2))\bigg[\Phi(\Upsilon(2\alpha B,\beta^*)-2V)-\Phi(\Upsilon(k_1,\beta^*)-2V))\\
&&+ B^2(1-\Phi(\Upsilon(2\alpha B,\beta^*))\bigg].
\end{eqnarray*}
If $k_1\geqslant2\alpha B $ and $Y_0>-B\exp\left(M+\frac{1}{2}V^2\right)$,  we have
\begin{eqnarray*}
\mathbb{E}\left[ \frac{X^*(T)^{1-\gamma}}{1-\gamma}1_{\{X^*(T)>0\}}\right]&=&\mathbb{E}\left[ \frac{x^*(\beta^* \rho(T))^{1-\gamma}}{1-\gamma}1_{\{x^*(\beta^* \rho(T))>0\}}\right]\\
&=&\mathbb{E}\left[ \frac{(\beta^* \rho(T))^{-\frac{1-\gamma}{\gamma}}}{1-\gamma}1_{\{0<\beta^* \rho(T)\leqslant k_2\}}\right]\\
&=&\frac{{\beta^*}^{\frac{\gamma\!-\!1}{\gamma}}}{1\!-\!\gamma}\exp\left(\frac{\gamma\!-\!1}{\gamma}M\!+\!
\frac{V^2(\gamma\!-\!1)^2}{2\gamma^2}\right)\Phi\left(\Upsilon(k_2,\beta^*)\!-\!\frac{\gamma\!-\!1}
{\gamma}V\right)
\end{eqnarray*}
and
\begin{eqnarray*}
{\mathbb{E}\left[ {X^*}^2(T)1_{\{{X^*}(T)<0\}}\right]}&=&\mathbb{E}\left[ x^*(\beta^* \rho(T))^21_{\{x^*(\beta^* \rho(T))<0\}}\right]\\
&=&\mathbb{E}\left[ (-B)^21_{\{\beta^* \rho(T)> k_2\}}\right]\\
		&=& B^2(1-\Phi(\Upsilon(k_2,\beta^*))).
\end{eqnarray*}
In the last two cases, the results are straightforward.
	\bibliographystyle{apalike}
	\bibliography{DB}	

\begin{thebibliography}{}

\bibitem[Asthana, 1999]{asthana1999determinants}
Asthana, S. (1999).
\newblock Determinants of funding strategies and actuarial choices for
  defined-benefit pension plans.
\newblock {\em Contemporary Accounting Research}, 16(1):39--74.

\bibitem[Beaudoin et~al., 2010]{beaudoin2010potential}
Beaudoin, C., Chandar, N., and Werner, E.~M. (2010).
\newblock Are potential effects of {SFAS} 158 associated with firms' decisions
  to freeze their defined benefit pension plans?
\newblock {\em Review of Accounting and Finance}.

\bibitem[Berkelaar et~al., 2004]{berkelaar2004optimal}
Berkelaar, A.~B., Kouwenberg, R., and Post, T. (2004).
\newblock Optimal portfolio choice under loss aversion.
\newblock {\em Review of Economics and Statistics}, 86(4):973--987.

\bibitem[Blake et~al., 2013]{blake2013target}
Blake, D., Wright, D., and Zhang, Y. (2013).
\newblock Target-driven investing: Optimal investment strategies in defined
  contribution pension plans under loss aversion.
\newblock {\em Journal of Economic Dynamics and Control}, 37(1):195--209.

\bibitem[Boulier et~al., 2001]{boulier2001optimal}
Boulier, J.-F., Huang, S., and Taillard, G. (2001).
\newblock Optimal management under stochastic interest rates: the case of a
  protected defined contribution pension fund.
\newblock {\em Insurance: Mathematics and Economics}, 28(2):173--189.

\bibitem[Cairns et~al., 2006]{cairns2006stochastic}
Cairns, A.~J., Blake, D., and Dowd, K. (2006).
\newblock Stochastic lifestyling: Optimal dynamic asset allocation for defined
  contribution pension plans.
\newblock {\em Journal of Economic Dynamics and Control}, 30(5):843--877.

\bibitem[Carassus and Pham, 2009]{carassus2009portfolio}
Carassus, L. and Pham, H. (2009).
\newblock Portfolio optimization for piecewise concave criteria functions (the
  8th workshop on stochastic numerics).
\newblock {\em RIMS Kokyuroku}, 1620:81--108.

\bibitem[C{\'a}rdenas et~al., 2014]{cardenas2014my}
C{\'a}rdenas, J.~C., De~Roux, N., Jaramillo, C.~R., and Martinez, L.~R. (2014).
\newblock Is it my money or not? an experiment on risk aversion and the
  house-money effect.
\newblock {\em Experimental Economics}, 17(1):47--60.

\bibitem[Carroll and Niehaus, 1998]{carroll1998pension}
Carroll, T.~J. and Niehaus, G. (1998).
\newblock Pension plan funding and corporate debt ratings.
\newblock {\em Journal of Risk and Insurance}, pages 427--443.

\bibitem[Cox and Huang, 1989]{cox1989optimal}
Cox, J.~C. and Huang, C.-F. (1989).
\newblock Optimal consumption and portfolio policies when asset prices follow a
  diffusion process.
\newblock {\em Journal of Economic Theory}, 49(1):33--83.

\bibitem[Cox et~al., 2013]{cox2013managing}
Cox, S.~H., Lin, Y., Tian, R., and Yu, J. (2013).
\newblock Managing capital market and longevity risks in a defined benefit
  pension plan.
\newblock {\em Journal of Risk and Insurance}, 80(3):585--620.

\bibitem[Dong and Zheng, 2020]{dong2020optimal}
Dong, Y. and Zheng, H. (2020).
\newblock Optimal investment with {S}-shaped utility and trading and {V}alue at
  {R}isk constraints: An application to defined contribution pension plan.
\newblock {\em European Journal of Operational Research}, 281(2):341--356.

\bibitem[Eaton and Nofsinger, 2004]{eaton2004effect}
Eaton, T.~V. and Nofsinger, J.~R. (2004).
\newblock The effect of financial constraints and political pressure on the
  management of public pension plans.
\newblock {\em Journal of Accounting and Public Policy}, 23(3):161--189.

\bibitem[Emms, 2012]{emms2012lifetime}
Emms, P. (2012).
\newblock Lifetime investment and consumption using a defined-contribution
  pension scheme.
\newblock {\em Journal of Economic Dynamics and Control}, 36(9):1303--1321.

\bibitem[Franzoni and Marin, 2006]{franzoni2006pension}
Franzoni, F. and Marin, J.~M. (2006).
\newblock Pension plan funding and stock market efficiency.
\newblock {\em The Journal of Finance}, 61(2):921--956.

\bibitem[Guan and Liang, 2016]{guan2016optimal}
Guan, G. and Liang, Z. (2016).
\newblock Optimal management of {DC} pension plan under loss aversion and
  {V}alue-at-{R}isk constraints.
\newblock {\em Insurance: Mathematics and Economics}, 69:224--237.

\bibitem[Hainaut and Deelstra, 2011]{hainaut2011optimal}
Hainaut, D. and Deelstra, G. (2011).
\newblock Optimal funding of defined benefit pension plans.
\newblock {\em Journal of Pension Economics and Finance}, 10(1):31--52.

\bibitem[Hull and White, 1990]{hull1990pricing}
Hull, J. and White, A. (1990).
\newblock Pricing interest-rate-derivative securities.
\newblock {\em The Review of Financial Studies}, 3(4):573--592.

\bibitem[Josa-Fombellida et~al., 2018]{JJ18}
Josa-Fombellida, R., L{\'o}pez-Casado, P., and Rinc{\'o}n-Zapatero, J.~P.
  (2018).
\newblock Portfolio optimization in a defined benefit pension plan where the
  risky assets are processes with constant elasticity of variance.
\newblock {\em Insurance: Mathematics and Economics}, 82:73--86.

\bibitem[Josa-Fombellida and Rinc{\'o}n-Zapatero, 2006]{josa2006optimal}
Josa-Fombellida, R. and Rinc{\'o}n-Zapatero, J.~P. (2006).
\newblock Optimal investment decisions with a liability: The case of defined
  benefit pension plans.
\newblock {\em Insurance: Mathematics and Economics}, 39(1):81--98.

\bibitem[Josa-Fombellida and Rinc{\'o}n-Zapatero, 2010]{JR10}
Josa-Fombellida, R. and Rinc{\'o}n-Zapatero, J.~P. (2010).
\newblock Optimal asset allocation for aggregated defined benefit pension funds
  with stochastic interest rates.
\newblock {\em European Journal of Operational Research}, 201(1):211--221.

\bibitem[Josa-Fombellida and Rinc{\'o}n-Zapatero, 2012]{josa2012stochastic}
Josa-Fombellida, R. and Rinc{\'o}n-Zapatero, J.~P. (2012).
\newblock Stochastic pension funding when the benefit and the risky asset
  follow jump diffusion processes.
\newblock {\em European Journal of Operational Research}, 220(2):404--413.

\bibitem[Kahneman and Tversky, 1979]{kahneman1979prospect}
Kahneman, D. and Tversky, A. (1979).
\newblock Prospect theory: An analysis of decision under risk.
\newblock {\em Econometrica: Journal of the Econometric Society}.

\bibitem[Kapinos, 2009]{kapinos2009determinants}
Kapinos, K.~A. (2009).
\newblock On the determinants of defined benefit pension plan conversions.
\newblock {\em Journal of Labor Research}, 30(2):149--167.

\bibitem[Li et~al., 2021]{li2021alpha}
Li, D., Bi, J., and Hu, M. (2021).
\newblock Alpha-robust mean-variance investment strategy for {DC} pension plan
  with uncertainty about jump-diffusion risk.
\newblock {\em RAIRO-Operations Research}, 55:S2983--S2997.

\bibitem[March, 1996]{march1996learning}
March, J.~G. (1996).
\newblock Learning to be risk averse.
\newblock {\em Psychological review}, 103(2):309.

\bibitem[Siegmann, 2007]{siegmann2007optimal}
Siegmann, A. (2007).
\newblock Optimal investment policies for defined benefit pension funds.
\newblock {\em Journal of Pension Economics and Finance}, 6(1):1--20.

\bibitem[Stone, 1987]{stone1987financing}
Stone, M. (1987).
\newblock A financing explanation for overfunded pension plan terminations.
\newblock {\em Journal of Accounting Research}, pages 317--326.

\bibitem[Sundaresan and Zapatero, 1997]{sundaresan1997valuation}
Sundaresan, S. and Zapatero, F. (1997).
\newblock Valuation, optimal asset allocation and retirement incentives of
  pension plans.
\newblock {\em The Review of Financial Studies}, 10(3):631--660.

\bibitem[Temocin et~al., 2018]{temocin2018constant}
Temocin, B.~Z., Korn, R., and Selcuk-Kestel, A.~S. (2018).
\newblock Constant proportion portfolio insurance in defined contribution
  pension plan management.
\newblock {\em Annals of Operations Research}, 266(1):329--348.

\bibitem[Thomas, 1989]{thomas1989firms}
Thomas, J.~K. (1989).
\newblock Why do firms terminate their overfunded pension plans?
\newblock {\em Journal of Accounting and Economics}, 11(4):361--398.

\bibitem[Tversky and Kahneman, 1991]{tversky1991loss}
Tversky, A. and Kahneman, D. (1991).
\newblock Loss aversion in riskless choice: A reference-dependent model.
\newblock {\em The Quarterly Journal of Economics}, 106(4):1039--1061.

\bibitem[Tversky and Kahneman, 1992]{tversky1992advances}
Tversky, A. and Kahneman, D. (1992).
\newblock Advances in prospect theory: Cumulative representation of
  uncertainty.
\newblock {\em Journal of Risk and Uncertainty}, 5(4):297--323.

\bibitem[Zeng et~al., 2018]{zeng2018ambiguity}
Zeng, Y., Li, D., Chen, Z., and Yang, Z. (2018).
\newblock Ambiguity aversion and optimal derivative-based pension investment
  with stochastic income and volatility.
\newblock {\em Journal of Economic Dynamics and Control}, 88:70--103.

\end{thebibliography}
\end{document}